\documentclass[twocolumn]{IEEEtran}

\usepackage{multirow}
\usepackage{graphicx}
\usepackage{xcolor}
\usepackage{epstopdf}
\usepackage{caption}
\usepackage{subcaption}
\usepackage{amsmath}
\usepackage{bm}
\usepackage{pgfplots}
\usepackage{footnote}
\usepackage{stfloats}
\usepackage{placeins}
\usepackage{color,soul}
\usepackage{dsfont}
\usepackage{amssymb}
\usepackage[boxruled]{algorithm2e}
\usepackage{pbox}
\usepackage{enumerate}

\usepackage{etextools}
\usepackage{mathtools}

\usepackage[hidelinks]{hyperref}

\usepackage{amsthm}

\newtheorem{innercustomgeneric}{\customgenericname}
\providecommand{\customgenericname}{}
\newcommand{\newcustomtheorem}[2]{%
  \newenvironment{#1}[1]
  {%
   \renewcommand\customgenericname{#2}%
   \renewcommand\theinnercustomgeneric{##1}%
   \innercustomgeneric
  }
  {\endinnercustomgeneric}
}

\newcustomtheorem{customthm}{Theorem}
\newcustomtheorem{customlemma}{Lemma}
\newcustomtheorem{customdef}{Definition}

\theoremstyle{plain}
\newtheorem{thm}{Theorem} % reset theorem numbering for each chapter

\newtheorem*{theorem*}{Theorem}
\newtheorem*{lemma*}{Lemma}
    
\mathtoolsset{showonlyrefs}

\newtheorem{lemma}{Lemma}

\theoremstyle{definition}
\newtheorem{defn}[thm]{Definition} % definition numbers are dependent on theorem numbers
\newtheorem{exmp}[thm]{Example} % same for example numbers

\newlength
\figureheight
\newlength
\figurewidth

\def\x{{\mathbf x}}
\def\X{{\mathbf X}}

\def\d{{\mathbf d}}

\def\uh{ \overline{\mathbf{u}}}
\def\u{ \mathbf{u}}
\def\r{ \mathbf{r}}
\def\xi{\x_i}

\def\Y{{\mathbf Y}}
\def\z{{\mathbf z}}
\def\Z{{\mathbf Z}}
\def\u{{\mathbf u}}

\def\S{{\mathbf S}}
\def\G{{\mathbf G}}
\def\D{{\mathbf D}}

\def\A{{\mathbf A}}

\def\M{{\mathbf Z}}
\def\P{{\mathbf P}}

\def\alfa{{\boldsymbol \alpha}}
\def\gama{{\boldsymbol \gamma}}
\def\Gama{{\boldsymbol \Gamma}}
\def\Delt{{\boldsymbol \Delta}}
\def\delt{{\boldsymbol \delta}}
\def\O{{\boldsymbol \Omega}}

\def\Poi{{P_{0,\infty}}}
\def\Loi{{\ell_{0,\infty}}}
\def\Poie{{P_{0,\infty}^\epsilon}}

\def\C{{\mathbf C}}

\def\Lamda{{\boldsymbol \Lambda}}

\def\R{{\mathbf R}}

\def\E{{\mathbf E}}
\def\DT{\D_\mathcal{T}}
\def\DTB{\D_{\overline{\mathcal{T}}}}
\def\x{{\mathbf x}}
\def\X{{\mathbf X}}

\def\d{{\mathbf d}}

\def\uh{ \bar{\mathbf{u}}}
\def\u{ \mathbf{u}}
\def\r{ \mathbf{r}}
\def\xi{\x_i}

\def\Y{{\mathbf Y}}
\def\z{{\mathbf z}}
\def\Z{{\mathbf Z}}
\def\u{{\mathbf u}}

\def\S{{\mathbf S}}
\def\G{{\mathbf G}}
\def\D{{\mathbf D}}

\def\A{{\mathbf A}}

\def\M{{\mathbf Z}}
\def\P{{\mathbf P}}

\def\alfa{{\boldsymbol \alpha}}
\def\gama{{\boldsymbol \gamma}}
\def\Gama{{\boldsymbol \Gamma}}
\def\Delt{{\boldsymbol \Delta}}
\def\delt{{\boldsymbol \delta}}
\def\O{{\boldsymbol \Omega}}

\def\Poi{{P_{0,\infty}}}
\def\Loi{{\ell_{0,\infty}}}
\def\Poie{{P_{0,\infty}^\epsilon}}

\def\C{{\mathbf C}}

\def\Lamda{{\boldsymbol \Lambda}}

\def\R{{\mathbf R}}

\def\E{{\mathbf E}}

\DeclarePairedDelimiterX\set[1]\lbrace\rbrace{\def\given{\;\delimsize\vert\;}#1}

\begin{document}

\title{Working Locally Thinking Globally: Theoretical Guarantees for Convolutional Sparse Coding}
\author{Vardan~Papyan*,
	Jeremias~Sulam*,
	Michael~Elad

\thanks{*The authors contributed equally to this work. 
	
All authors are with the Computer Science Department, the Technion - Israel Institute of Technology.}}

\maketitle

\begin{abstract}
The celebrated sparse representation model has led to remarkable results in various signal processing tasks in the last decade.
However, \mbox{despite} its initial purpose of serving as a global prior for entire signals, it has been commonly used for modeling low dimensional patches due to the computational constraints it entails when deployed with learned dictionaries.
A way around this problem has been recently proposed, adopting a convolutional sparse representation model.
This approach assumes that the global dictionary is a concatenation of banded Circulant matrices.
While several works have presented algorithmic solutions to the global pursuit problem under this new model, very few truly-effective guarantees are known for the success of such methods.
In this work, we address the theoretical aspects of the convolutional sparse model providing the first meaningful answers to questions of uniqueness of solutions and success of pursuit algorithms, both greedy and convex relaxations, in ideal and noisy regimes.
To this end, we generalize mathematical quantities, such as the $\ell_0$ norm, mutual coherence, Spark and RIP to their counterparts in the convolutional setting, intrinsically capturing local measures of the global model.
On the algorithmic side, we demonstrate how to solve the global pursuit problem by using simple local processing, thus offering a first of its kind bridge between global modeling of signals and their patch-based local treatment.
\end{abstract}


%% ---------------------------------------------------------------------------------------------------------------
%% ---------------------------------------------------------------------------------------------------------------
\section{Introduction}

A popular choice for a signal model, which has proven to be very effective in a wide range of applications, is the celebrated sparse representation prior \cite{Bruckstein2009,Mairal2014,Romano2014,Dong2011}. In this framework, one assumes a signal $\X \in \mathbb{R}^N$ to be a sparse combination of a few columns (atoms) $\d_i$ from a collection $\D \in \mathbb{R}^{N\times M}$, termed dictionary. In other words, $\X = \D\Gama$ where $\Gama\in \mathbb{R}^{M}$ is a sparse vector. Finding such a vector can be formulated as the following optimization problem:
\begin{equation}
\min_{\Gama}\ g(\Gama)\ \text{ s.t. } \D\Gama = \X,
\label{Eq:Global Sparse Model}
\end{equation}
where $g(\cdot)$ is a function which penalizes dense solutions, such as the $\ell_1$ or $\ell_0$ ``norms''\footnote{Despite the $\ell_0$ not being a norm (as it does not satisfy the homogeneity property), we will use this jargon throughout this paper for the sake of simplicity.}. For many years, analytically defined matrices or operators were used as the dictionary $\D$ \cite{Mallat1993,Elad2005}. However, designing a model from real examples by some learning procedure has proven to be more effective, providing sparser solutions \cite{Engan1999,Aharon2006,Mairal2009a}. This led to vast work that deploys dictionary learning in a variety of applications \cite{Li2011,Zhang2010,Dong2011,Gao2012,Yang2014}.

Generally, solving a pursuit problem is a computationally challenging task. As a consequence, most such recent successful methods have been deployed on relatively small dimensional signals, commonly referred to as \emph{patches}. Under this \emph{local} paradigm, the signal is broken into overlapped blocks and the above defined sparse coding problem is reformulated as
\begin{equation}
\forall \ i \quad \min_{\alfa}\ g(\alfa)\ \text{ s.t. }\ \D_L\alfa = \R_i\X,
\label{eq:first_equation}
\end{equation}
where $\D_L \in \mathbb{R}^{n\times m}$ is a local dictionary, and $\R_i\in \mathbb{R}^{n\times N}$ is an operator which extracts a small local patch of length $n\ll N$ from the global signal $\X$. In this set-up, one processes each patch independently and then aggregates the estimated results using plain averaging in order to recover the global reconstructed signal. A local-global gap naturally arises when solving global tasks with this local approach, which ignores the correlation between overlapping patches. The reader is referred to \cite{Sulam2015,Romano2015b,Romano2015,Papyan2016,batenkov2017global,Zoran2011} for further insights on this dichotomy.

The above discussion suggests that in order to find a consistent global representation for the signal, one should propose a global sparse model. However, employing a general global (unconstrained) dictionary is infeasible due the computational complexity involved, and training this model suffers from the curse of dimensionality. An alternative is a (constrained) global model in which the signal is composed as a superposition of local atoms. The family of dictionaries giving rise to such signals is a concatenation of banded Circulant matrices. This global model benefits from having a local shift invariant structure -- a popular assumption in signal and image processing -- suggesting an interesting connection to the above-mentioned local modeling.

When the dictionary $\D$ has this structure of a concatenation of banded Circulant matrices, the pursuit problem in \eqref{Eq:Global Sparse Model} is usually known as convolutional sparse coding \cite{Grosse2007}. Recently, several works have addressed the problem of using and training such a model in the context of image inpainting, super-resolution, and general image representation \cite{Bristow2013,Heide2015,Kong2014,Wohlberg2014,Gu2015}. These methods usually exploit an ADMM formulation \cite{Boyd2011} while operating in the Fourier domain in order to search for the sparse codes and train the dictionary involved. Several variations have been proposed for solving the pursuit problem, yet there has been no theoretical analysis of their success. 

Assume a signal is created by multiplying a sparse vector by a convolutional dictionary. In this work, we consider the following set of questions:
\begin{enumerate}
	\item Can we guarantee the uniqueness of such a global (convolutional) sparse vector?
	\item Can global pursuit algorithms, such as the ones suggested in recent works, be guaranteed to find the true underlying sparse code, and if so, under which conditions?
	\item Can we guarantee a stability of the sparse approximation problem, and a stability of corresponding pursuit methods in a noisy regime?; And
	\item Can we solve the global pursuit by restricting the process to local pursuit operations?
\end{enumerate}

A na\"ive approach to address such theoretical questions is to apply the fairly extensive results for sparse representation and compressed sensing to the above defined model \cite{Elad_Book}. However, as we will show throughout this paper, this strategy provides nearly useless results and bounds from a global perspective. Therefore, there exists a true need for a deeper and alternative analysis of the sparse coding problem in the convolutional case which would yield meaningful bounds. 

In this work, we will demonstrate the futility of the $\ell_0$-norm in providing meaningful bounds in the convolutional model. This, in turn, motivates us to propose a new localized measure -- the $\Loi$ norm. Based on it, we redefine our pursuit into a problem that operates locally while thinking globally. To analyze this problem, we extend useful concepts, such as the Spark and mutual coherence, to the convolutional setting. We then provide claims for uniqueness of solutions and for the success of pursuit methods in the noiseless case, both for greedy algorithms and convex relaxations. Based on these theoretical foundations, we then extend our analysis to a more practical scenario, handling noisy data and model deviations. We generalize and tie past theoretical constructions, such as the Restricted Isometry Property (RIP) \cite{Candes2005} and the Exact Recovery Condition (ERC) \cite{Tropp2004}, to the convolutional framework proving the stability of this model in this case as well. 

 %and generalize these into stability guarantees in the noisy regime. By leveraging on the local structure of the model, we further show that the global pursuit can be solved by decomposing it into local problems.

This paper is organized as follows. We begin by reviewing the unconstrained global (traditional) sparse representation model in Section \ref{Sec:Preliminaries}, followed by a detailed description of the convolutional structure in Section \ref{Sec:Conv_Model}. Section \ref{Sec:Global2Local} briefly motivates the need of a thorough analysis of this model, which is then provided in Section \ref{Sec:TheoStudy}. We introduce additional mathematical tools in Section \ref{Sec:ShiftedMutualCoherence}, which provide further insight into the convolutional model. The noisy scenario is then considered in Section \ref{Sec:NoisyRegime} and analyzed in Section \ref{Sec:TheoreticalAnalysis_noisy}, where we assume the global signal to be contaminated with norm-bounded error. We then bridge the local and global models in Section \ref{Sec:FromGlobal2LocalProcessing}, proposing local algorithmic solutions to tackle the convolutional pursuit. Finally, we conclude this work in Section \ref{Sec:Conclusions}, proposing exciting future directions.

%% ---------------------------------------------------------------------------------------------------------------
%% ---------------------------------------------------------------------------------------------------------------
\section{The Global Sparse Model -- Preliminaries}
\label{Sec:Preliminaries}

Consider the constrained $P_0$ problem, a special case of Eq. \eqref{Eq:Global Sparse Model}, given by
% the case where $g(\Gama) = \|\Gama\|_0$ in Eq. \eqref{Eq:Global Sparse Model}. This problem, known as Sparse Coding, attempts to solve the constrained $P_0$ problem:
\begin{equation}
(P_0): \quad \underset{\Gama}{\min}\ \|\Gama\|_0 \ \text{ s.t. } \  \D\Gama = \X. 
\label{Eq:P0problem}
\end{equation}
Several results have shed light on the theoretical aspects of this problem, claiming a unique solution under certain circumstances. These guarantees are given in terms of properties of the dictionary $\D$, such as the \emph{Spark}, defined as the minimum number of linearly dependent columns (atoms) in $\D$ \cite{Donoho2003}. Formally,
\begin{equation}
\sigma(\D) = \underset{\Gama}{\min}\ \|\Gama\|_0 \ \text{ s.t. } \ \D\Gama = \mathbf{0}, \ \Gama \neq \mathbf{0}.
\end{equation}
Based on this property, a solution obeying $\|\Gama\|_0 < \sigma(\D)/2$ is necessarily the sparsest one \cite{Donoho2003}. Unfortunately, this bound is of little practical use, as computing the Spark of a matrix is a combinatorial problem -- and infeasible in practice.  

Another guarantee is given in terms of the \emph{mutual coherence} of the dictionary, $\mu(\D)$, which quantifies the similarity of atoms in the dictionary. Assuming hereafter that $\|\d_i\|_2 = 1 \ \forall i$, this was defined in \cite{Donoho2003} as:
\begin{equation}
\mu(\D) = \underset{i\neq j}{\max} \|\d_i^T \d_j\|.
\end{equation}
A relation between the Spark and the mutual coherence was also shown in \cite{Donoho2003}, stating that $\sigma(\D) \geq 1 + \frac{1}{\mu(\D)}$. This, in turn, enabled the formulation of a practical uniqueness bound guaranteeing that $\Gama$ is the unique solution of the $P_0$ problem if $\|\Gama\|_0 < \frac{1}{2} \left(1 + 1/\mu(\D)\right)$.

Solving the $P_0$ problem is NP-hard in general. Nevertheless, its solution can be approximated by either greedy pursuit algorithms, such as the Orthogonal Matching Pursuit (OMP) \cite{Pati1993a,Chen1989}, or convex relaxation approaches like Basis Pursuit (BP) \cite{Chen2001}. Despite the difficulty of this problem, these methods (and other similar ones) have been proven to recover the true solution if $\|\Gama\|_0 < \frac{1}{2} \left(1 + 1/\mu(\D)\right)$ \cite{Donoho2006,Tropp2004,Donoho2003,Gribonval2003}.

When dealing with natural signals, the $P_0$ problem is often relaxed to consider model deviations as well as measurement noise. In this set-up one assumes $\Y = \D\Gama + \E$, where $\E$ is a nuisance vector of bounded energy, $\| \E\|_2 \leq \epsilon$. The corresponding recovery problem can then be stated as follows:
\begin{equation}
(P_0^\epsilon): \quad \underset{\Gama}{\min} \ \|\Gama\|_0 \ \text{ s.t. } \|\D\Gama - \Y\|_2 \leq \epsilon.
\end{equation}
Unlike the noiseless case, given a solution to the above problem, one can not claim its uniqueness in solving the $P_0^\epsilon$ problem but instead can guarantee that it will be close enough to the true vector $\Gama$ that generated the signal $\Y$.
This kind of stability results have been derived in recent years by leveraging the Restricted Isometry Property (RIP) \cite{Candes2005}. A matrix $\D$ is said to have a k-RIP with constant $\delta_k$ if this is the smallest quantity such that
\begin{equation}
(1-\delta_k) \| \Gama \|_2^2 \leq \| \D\Gama \|_2^2 \leq (1+\delta_k) \| \Gama \|_2^2,
\end{equation}
for every $\Gama$ satisfying $\| \Gama \|_0=k$. Based on this property, it was shown that assuming $\Gama$ is sparse enough, the distance between $\Gama$ and all other solutions to the $P_0^\epsilon$ problem is bounded \cite{Elad_Book}. Similar stability claims can be formulated in terms of the mutual coherence also, by exploiting its relationship with the RIP property \cite{Elad_Book}. 

Success guarantees of practical algorithms, such as the Orthogonal Matching Pursuit (OMP) and the Basis Pursuit Denoising (BPDN), have also been derived under this regime. In the same spirit of the aforementioned stability results, the work in \cite{Donoho2006} showed that these methods recover a solution close to the true sparse vector as long as some sparsity constraint, relying on the mutual coherence of the dictionary, is met.

\begin{figure*}[t]
	\centering
	\includegraphics[trim = -40 60 40 30 ,width=0.9\textwidth]{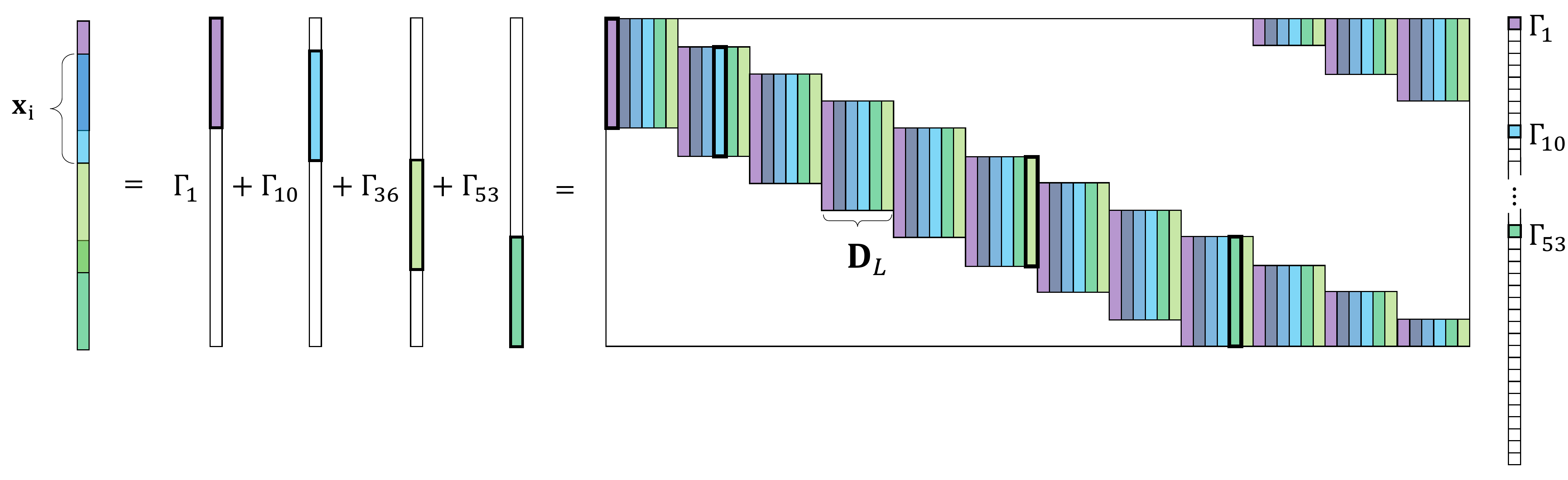}
	\caption{The convolutional model description, and its composition in terms of the local dictionary $\D_L$.}
	\label{GlobalDiagram}
\end{figure*}

Another useful property for analyzing the success of pursuit methods, initially proposed in \cite{Tropp2004}, is the Exact Recovery Condition (ERC). Formally, one says that the ERC is met for a support $\mathcal{T}$ with a constant $\theta$ whenever
\begin{equation}
	%\text{ERC}(\mathcal{T},\D): \quad
	\theta =  1 - \underset{i \notin \mathcal{T}}{\max} \|\D^{\dagger}_{\mathcal{T}} \d_i \|_1 > 0,
\end{equation}
where we have denoted by $\D_\mathcal{T}^{\dagger}$ the Moore-Penrose pseudoinverse of the dictionary restricted to support $\mathcal{T}$, and $\d_i$ refers to the $i^{th}$ atom in $\D$. Assuming the above is satisfied, the stability of both the OMP and BP was proven in \cite{Tropp2006}. Moreover, in an effort to provide a more intuitive result, the ERC was shown to hold whenever the total number of non-zeros in $\mathcal{T}$ is less than a certain number, which is a function of the mutual coherence.

%% ---------------------------------------------------------------------------------------------------------------
%% ---------------------------------------------------------------------------------------------------------------
\section{The Convolutional Sparse Model}
\label{Sec:Conv_Model}

Consider now the global dictionary to be a concatenation of $m$ banded Circulant matrices\footnote{Each of these matrices is constructed by shifting a single column, supported on $n$ subsequent entries, to all possible shifts. This choice of Circulant matrices comes to alleviate boundary problems.}, where each such matrix has a band of width $n \ll N$. As such, by simple permutation of its columns, such a dictionary consists of all shifted versions of a \emph{local} dictionary $\D_L$ of size $n\times m$. This model is commonly known as Convolutional Sparse Representation \cite{Grosse2007,Bristow2014,Heide2015}. Hereafter, whenever we refer to the global dictionary $\D$, we assume it has this structure. Assume a signal $\X$ to be generated as $\D\Gama$. In Figure \ref{GlobalDiagram} we describe such a global signal, its corresponding dictionary that is of size $N \times mN$ and its sparse representation, of length $mN$. We note that $\Gama$ is built of $N$ distinct and independent sparse parts, each of length $m$, which we will refer to as the local sparse vectors $\alfa_i$.
 
Consider a sub-system of equations extracted from $\X=\D\Gama$ by multiplying this system by the patch extraction\footnote{Denoting by $\mathbf{0}_{( a\times b)}$ a zeros matrix of size $a \times b$, and $\mathbf{I}_{( n\times n )}$ an identity matrix of size $n \times n$, then $\R_i = \left[ \mathbf{0}_{( n\times (i-1) )},\mathbf{I}_{( n\times n )},\mathbf{0}_{( n \times (N-i-n+1) )}\right]$.} operator $\R_i \in \mathbb{R}^{n \times N}$. The resulting system is $\x_i$ = $\R_i \X = \R_i \D \Gama$, where $\x_i$ is a patch of length $n$ extracted from $\X$ from location $i$. Observe that in the set of extracted rows, $\R_i \D$, there are only $(2n-1)m$ columns that are non-trivially zero. Define the operator $\S_i\in \mathbb{R}^{(2n-1)m \times mN}$ as a columns' selection operator\footnote{An analogous definition can be written for this operator as well.}, such that $\R_i \D \S_i^T$ preserves all the non-zero columns in $\R_i \D$. Thus, the subset of equations we get is essentially 
\begin{equation} \label{eq:xyz}
\x_i = \R_i \X = \R_i \D \Gama = \R_i \D \S_i^T \S_i \Gama.
\end{equation}
\begin{defn}
Given a global sparse vector $\Gama$, define $\gama_i = \S_i \Gama$ as its $i^{th}$ stripe representation.
\end{defn}
\noindent
Note that a stripe $\gama_i$ can be also seen as a group of $2n-1$ adjacent local sparse vectors $\alfa_j$ of length $m$ from $\Gama$, centered at location $\alfa_i$.
\begin{defn}
Consider a convolutional dictionary $\D$ defined by a local dictionary $\D_L$ of size $n \times m$. Define the stripe dictionary $\O$ of size $n \times (2n - 1)m$, as the one obtained by extracting $n$ consecutive rows from $\D$, followed by the removal of its zero columns, namely $\O = \R_i \D \S_i^T$.
\end{defn}
\noindent
Observe that $\O$, depicted in Figure \ref{PartialStripe}, is independent of $i$, being the same for all locations due to the union-of-Circulant-matrices structure of $\D$. In other words, the shift invariant property is satisfied for this model -- all patches share the same stripe dictionary in their construction. Armed with the above two definitions, Equation \eqref{eq:xyz} simply reads $\x_i = \O \gama_i$. 

% Each $\gama_i$ can be decomposed into $2n-1$ non-overlapping chunks, that we will refer to as \emph{shifts}, of length $m$. These local sparse codes will be denoted by $\gama_{i,s}$, where $s$ corresponds to the shift within $\gama_i$, and $s = -n+1,\dots,n-1$. These elements are depicted in Figure .

From a different perspective, one can synthesize the signal $\X$ by considering $\D$ as a concatenation of $N$ vertical stripes of size $N\times m$ (see Figure \ref{GlobalDiagram}), where each can be represented as $\R_i^T \D_L$. In other words, the vertical stripe is constructed by taking the small and local dictionary $\D_L$ and positioning it in the $i^{th}$ row. The same partitioning applies to $\Gama$, leading to the $\alfa_i$ ingredients. Thus,  
\begin{equation}
\X = \sum_i \R_i^T \D_L \alfa_i.
\end{equation}
Since $\alfa_i$ play the role of local sparse vectors, $\D_L \alfa_i$ are reconstructed patches (which are not the same as $\x_i = \O \gama_i$), and the sum above proposes a patch averaging approach as practiced in several works \cite{Aharon2006,Zoran2011,Sulam2015}. This formulation provides another local interpretation of the convolutional model.

Yet a third interpretation of the very same signal construction can be suggested, in which the signal is seen as resulting from a sum of local/small atoms which appear in a small number of locations throughout the signal. This can be formally expressed as 
\begin{equation}
\X = \sum_{i=1}^m \d_i \ast \z_i, 
\end{equation}
where the vectors $\z_i\in \mathbb{R}^N$ are sparse maps encoding the location and coefficients convolved with the $i^{th}$ atom \cite{Grosse2007}. In our context, $\Gama$ is simply the interlaced concatenation of all $\z_i$.

This model (adopting the last convolutional interpretation) has received growing attention in recent years in various applications. In \cite{Morup2008} a convolutional sparse coding framework was used for pattern detection in images and the analysis of instruments in music signals, while in \cite{Zhu2015} it was used for the reconstruction of 3D trajectories. The problem of learning the local dictionary $\D_L$ was also studied in several works \cite{Zeiler2010,Kavukcuoglu2010,Bristow2014,Wohlberg2014,Huang2015}.
Different methods have been proposed for solving the convolutional sparse coding problem under an $\ell_1$-norm penalty. Commonly, these methods rely on the ADMM algorithm \cite{Boyd2011}, exploiting multiplications of vectors by the global dictionary in the Fourier domain in order to reduce the computational cost involved \cite{Bristow2014}. An alternative is the deployment of greedy algorithms of the Matching Pursuit family \cite{Mallat1993}, which suggest an $\ell_0$ constraint on the global sparse vector. The reader is referred to the work of \cite{Wohlberg2014} and references therein for a thorough discussion on these methods.
In essence, all the above works are solutions to the minimization of a global pursuit under the convolutional structure. As a result, the theoretical results in our work will apply to these methods, providing guarantees for the recovery of the underlying sparse vectors.

\begin{figure}[t]
	\centering
	\includegraphics[trim = 0 50 0 20,width=0.5\textwidth]{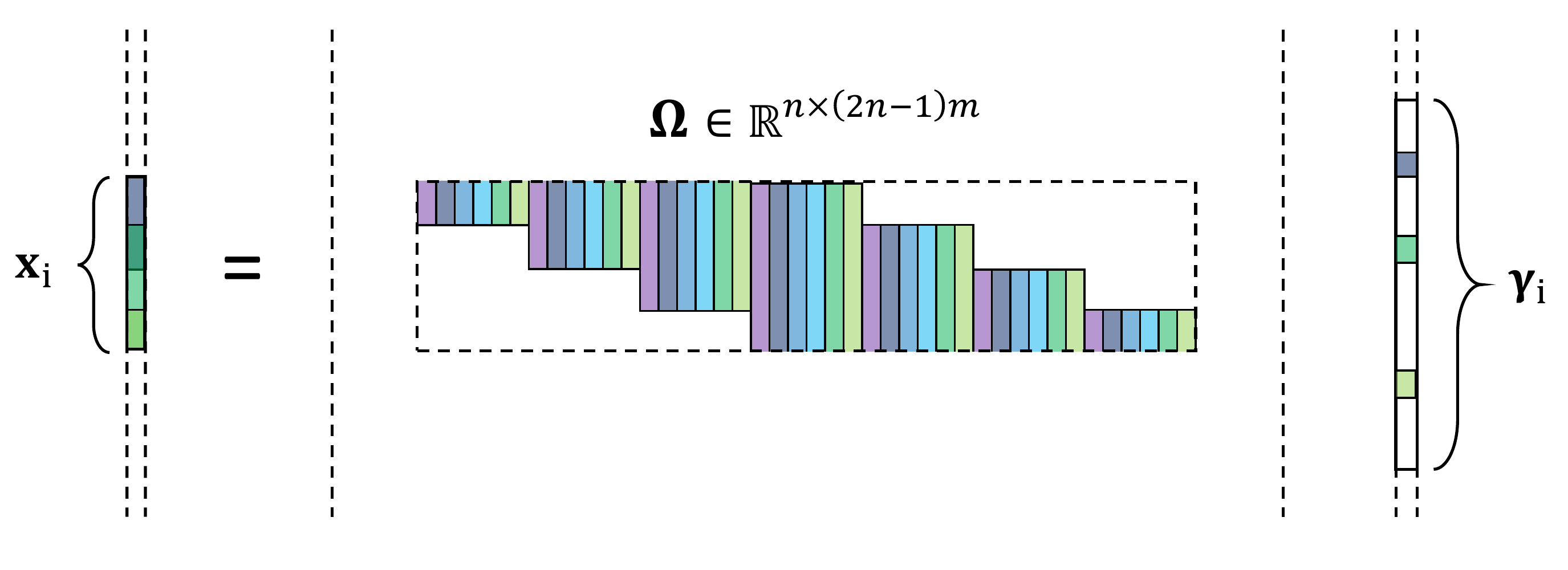}
	\caption{Stripe Dictionary}
	\label{PartialStripe}
	\vspace{-0.3cm}
\end{figure}
%% ---------------------------------------------------------------------------------------------------------------
%% ---------------------------------------------------------------------------------------------------------------
\section{From Global to Local Analysis}
\label{Sec:Global2Local}

Consider a sparse vector $\Gama$ of size $mN$ which represents a global (convolutional) signal. Assume further that this vector has a few $k \ll N$ non-zeros. If these were to be clustered together in a given stripe $\gama_i$, the local patch corresponding to this stripe would be very complex, and pursuit methods would likely fail in recovering it. On the contrary, if these $k$ non-zeros are spread all throughout the vector $\Gama$, this would clearly imply much simpler local patches, facilitating their successful recovery. This simple example comes to show the futility of the traditional global $\ell_0$-norm in assessing the success of convolutional pursuits, and it will be the pillar of our intuition throughout our work.

%% ---------------------------------------------------------------------------------------------------------------
\subsection{The $\ell_{0,\infty}$ Norm and the $P_{0,\infty}$ Problem}

Let us now introduce a measure that will provide a local notion of sparsity within a global sparse vector.
\begin{defn}
Define the $\Loi$ pseudo-norm of a global sparse vector $\Gama$ as
\begin{equation}
\|\Gama\|_{0,\infty} = \max_i \|\gama_i\|_0.
\end{equation}
\end{defn}
\noindent
In words, this quantifies the number of non-zeros in the densest stripe $\gama_i$ of the global $\Gama$. This is equivalent to extracting all stripes from the global sparse vector $\Gama$, arranging them column-wise into a matrix $\A$ and applying the usual $\|\A\|_{0,\infty}$ norm -- thus, the name. By constraining the $\Loi$ norm to be low, we are essentially limiting all stripes $\gama_i$ to be sparse, and their corresponding patches $\R_i\X$ to have a sparse representation under a shift-invariant local dictionary $\O$. This is one of the underlying assumptions in many signal and image processing algorithms. As for properties of this norm, similar to $\ell_0$ case, in the $\Loi$ the non-negativity and triangle inequality properties hold, while homogeneity does not.

Armed with the above definition, we now move to define the $\Poi$ problem:
\begin{equation}
(\Poi): \quad \min_\Gama \quad \|\Gama\|_{0,\infty} \ \text{ s.t. }\ \D\Gama=\X.
\end{equation}
When dealing with a global signal, instead of solving the $P_0$ problem (defined in Equation \eqref{Eq:P0problem}) as is commonly done, we aim to solve the above defined objective instead. The key difference is that we are not limiting the overall number of zeros in $\Gama$, but rather putting a restriction on its local density. 

%% ---------------------------------------------------------------------------------------------------------------
\subsection{Global versus Local Bounds}
\label{GlobalVsLocal}
As mentioned previously, theoretical bounds are often given in terms of the mutual coherence of the dictionary. In this respect, a lower bound on this value is much desired. In the case of the convolution sparse model, this value quantifies not only the correlation between the atoms in $\D_L$, but also the correlation between their shifts. Though in a different context, a bound for this value was derived in \cite{Welch1974}, and it is given by
\begin{equation} \label{eq:mu_bound}
\mu(\D)\geq\sqrt{\frac{m-1}{m(2n-1)-1}}.
\end{equation}
%For example, if $m=1$ (one local atom with all its shifts), this suggests that $\D$ might be an orthogonal matrix, and thus $\mu(\D)=0$. Going to the other extreme, 
For a large value of $m$, one obtains that the best possible coherence is \mbox{$\mu(\D)\approx \frac{1}{\sqrt{2n}}$}. % -- this is a very high value (e.g., if $n=128$, this coherence bound is $1/16$), considering the fact that it characterizes the whole global dictionary. 
This implies that if we are to apply BP or OMP to recover the sparsest $\Gama$ that represents $\X$, the classical sparse approximation results \cite{Bruckstein2009} would allow merely $O(\sqrt{n})$ non-zeros in \textbf{all} $\Gama$, for any $N$, no matter how long $\X$ is!
As we shall see next, the situation is not as grave as it may seem, due to our migration from $P_0$ to $\Poi$.
Leveraging the previous definitions, we will provide global recovery guarantees that will have a local flavor, and the bounds will be given in terms of the number of non-zeros in the densest stripe. This way, we will show that the guarantee conditions can be significantly enhanced to $O(\sqrt{n})$ non-zeros \emph{locally} rather than \emph{globally}.

%% ---------------------------------------------------------------------------------------------------------------
%% ---------------------------------------------------------------------------------------------------------------
\vspace{-0.2cm}
\section{Theoretical Study of Ideal Signals}
\label{Sec:TheoStudy}
%% ---------------------------------------------------------------------------------------------------------------
%\subsection{Notations}
As motivated in the previous section, the concerns of uniqueness, recovery guarantees and stability of sparse solutions in the convolutional case require special attention. We now formally address these questions by following the path taken in \cite{Elad_Book}, carefully generalizing each and every statement to the global-local model discussed here.

Before proceeding onto theoretical grounds, we briefly summarize, for the convenience of the reader, all notations used throughout this work in Table \ref{Table:Notations}. Note the somewhat unorthodox choice of capital letters for global vectors and lowercase for local ones.

\vspace{0.5cm}

\begin{table}[t] \centering
\begin{tabular}{|p{1.85cm}@{:\quad}l|} \hline
$N$ & length of the global signal. \\ \hline
$n$ & size of a local atom or a local signal patch. \\ \hline
$m$ & \pbox{20cm}{number of unique local atoms (filters) or the number\\ of Circulant matrices.} \\[0.2cm] \hline
$\X$, $\Y$ and $\E$ & \pbox{20cm}{global signals of length $N$, where generally\\ $\Y=\X+\E$.} \\ \hline
$\D$ & global dictionary of size $N\times mN$. \\ \hline
$\Gama$ and $\Delt$ & global sparse vectors of length $mN$. \\ \hline
$\Gamma_i$ and $\Delta_i$ & the $i^{th}$ entry in $\Gama$ and $\Delt$, respectively. \\ \hline
$\D_L$ & local dictionary of size $n\times m$. \\ \hline
$\O$ & \pbox{20cm}{stripe dictionary of size $n\times(2n-1)m$, which \\ contains all possible shifts of $\D_L$.} \\ \hline
$\alfa_i$ & local sparse code of size $m$. \\ \hline
$\gama_i$ and $\delt_i$ & \pbox{20cm}{a stripe of length $(2n-1)m$ extracted \\ from the global vectors $\Gama$ and $\Delt$, respectively.} \\ \hline
$\gama_{i,s}$ and $\delt_{i,s}$ & \pbox{20cm}{a local sparse vector of length $m$ which corresponds \\ to the $s^{th}$ portion inside $\gama_i$ and $\delt_i$, respectively.} \\ \hline
\end{tabular}
\caption{Summary of notations used throughout the paper.}
\label{Table:Notations}
\vspace{-0.4cm}
\end{table}

%% ---------------------------------------------------------------------------------------------------------------
\vspace{-0.6cm}
\subsection{Uniqueness and Stripe-Spark}

Just as it was initially done in the general sparse model, one might ponder about the uniqueness of the sparsest representation in terms of the $\Loi$ norm. More precisely, does a unique solution to the $\Poi$ problem exist? and under which circumstances? In order to answer these questions we shall first extend our mathematical tools, in particular the characterization of the dictionary, to the convolutional scenario.

In Section \ref{Sec:Preliminaries} we recalled the definition of the Spark of a general dictionary $\D$. In the same spirit, we propose the following:
\begin{defn}
Define the Stripe-Spark of a convolutional dictionary\ $\D$ as
\begin{equation}
\sigma_\infty(\D)= \min_\Delt \quad \|\Delt\|_{0,\infty} \ \text{ s.t. }\  \Delt\neq 0,\ \D\Delt=0.
\end{equation}
\end{defn}
\noindent
In words, the Stripe-Spark is defined by the sparsest non-zero vector, in terms of the $\Loi$ norm, in the null space of $\D$. 
Next, we use this definition in order to formulate an uncertainty and a uniqueness principle for the $\Poi$ problem that emerges from it. The proof of this and the following theorems are described in detail in the Supplementary Material.

\begin{thm}{(Uncertainty and uniqueness using Stripe-Spark):}
	Let $\D$ be a convolutional dictionary. If a solution $\Gama$ obeys \mbox{$\|\Gama\|_{0,\infty}<\frac{1}{2}\sigma_\infty$}, then this is necessarily the global optimum for the $\Poi$ problem for the signal $\D\Gama$.
\end{thm}

%% ---------------------------------------------------------------------------------------------------------------
\subsection{Lower Bounding the Stripe-Spark}

In general, and similar to the Spark, calculating the Stripe-Spark is computationally intractable. Nevertheless, one can bound its value using the global mutual coherence defined in Section \ref{Sec:Preliminaries}. Before presenting such bound, we formulate and prove a Lemma that will aid our analysis throughout this paper.

\begin{lemma}{}
\label{lemma:GershgorinConvolutional}
Consider a convolutional dictionary $\D$, with mutual coherence $\mu(\D)$, and a support $\mathcal{T}$ with $\Loi$ norm\footnote{Note that specifying the $\Loi$ of a support rather than a sparse vector is a slight abuse of notation, that we will nevertheless use for the sake of simplicity.} equal to $k$. Let $\G^\mathcal{T} = \D_\mathcal{T}^T \D_\mathcal{T}$, where $\D_\mathcal{T}$ is the matrix $\D$ restricted to the columns indicated by the support $\mathcal{T}$. Then, the eigenvalues of this Gram matrix, given by $\lambda_i \left(\G^\mathcal{T}\right)$, are bounded by
\begin{equation}
1-(k-1)\mu(\D) \ \leq \ \lambda_i \left(\G^\mathcal{T}\right) \ \leq \ 1+(k-1)\mu(\D).
\end{equation}
\end{lemma}

\begin{proof}
	% The Gram matrix $\G^\mathcal{T}$ corresponds to a portion extracted from the global Gram matrix $\D^T\D$.
	From Gerschgorin's theorem, the eigenvalues of the Gram matrix $\G^\mathcal{T}$ reside in the union of its Gerschgorin circles. The $j^{th}$ circle, corresponding to the $j^{th}$ row of $\G^\mathcal{T}$, is centered at the point $\G^{\mathcal{T}}_{j,j}$ (belonging to the Gram's diagonal) and its radius equals the sum of the absolute values of the off-diagonal entries; i.e., $\sum_{i,i\neq j} |\G^{\mathcal{T}}_{j,i}|$. Notice that both indices $i,j$ correspond to atoms in the support $\mathcal{T}$. Because the atoms are normalized, $\forall\ j,\ \G^{\mathcal{T}}_{j,j} = 1$, implying that all Gershgorin disks are centered at $1$. Therefore, all eigenvalues reside inside the circle with the largest radius. Formally,
	\begin{equation} \label{eq:max_radius}
	\big| \lambda_i\left(\G^\mathcal{T}\right) - 1 \big| \leq \max_j \sum_{i,i\neq j}  \big| \G^{\mathcal{T}}_{j,i} \big| = \max_j \sum_{\substack{i,i\neq j \\ i,j \in \mathcal{T}}} | \d_j^T\d_i \big|.
	\end{equation}
	On the one hand, from the definition of the mutual coherence, the inner product between atoms that are close enough to overlap is bounded by $\mu(\D)$. On the other hand, the product $\d_j^T\d_i$ is zero for atoms $\d_i$ too far from $\d_j$ (i.e., out of the stripe centered at the $j^{th}$ atom). Therefore, we obtain:
	\begin{equation}
	\sum_{\substack{i,i\neq j \\ i,j \in \mathcal{T}}} | \d_j^T\d_i | \leq (k - 1)\ \mu(\D),
	\end{equation}
	where $k$ is the maximal number of non-zero elements in a stripe, defined previously as the $\Loi$ norm of $\mathcal{T}$. Note that we have subtracted $1$ from $k$ because we must omit the entry on the diagonal. Putting this back in Equation \eqref{eq:max_radius}, we obtain
	\begin{equation}
	\big| \lambda_i\left(\G^\mathcal{T}\right) - 1 \big| \leq \max_j \sum_{\substack{i,i\neq j \\ i,j \in \mathcal{T}}} | \d_j^T\d_i \big| \leq\ (k - 1)\ \mu(\D).
	\end{equation}
	From this we obtain the desired claim.
\end{proof}

We now dive into the next theorem, whose proof relies on the above Lemma.
\begin{thm}{(Lower bounding the Stripe-Spark via the mutual coherence):}
	For a convolutional dictionary $\D$ with mutual coherence $\mu(\D)$, the Stripe-Spark can be lower-bounded by
	\begin{equation}
	\sigma_\infty(\D)\geq 1+\frac{1}{\mu(\D)}.
	\end{equation}
\end{thm}
\noindent
Using the above derived bound and the uniqueness based on the Stripe-Spark we can now formulate the following theorem:
\begin{thm}{(Uniqueness using mutual coherence):}
	Let $\D$ be a convolutional dictionary with mutual coherence $\mu(\D)$. If a solution $\Gama$ obeys $\|\Gama\|_{0,\infty}<\frac{1}{2}(1+\frac{1}{\mu(\D)})$, then this is necessarily the sparsest (in terms of $\Loi$ norm) solution to $\Poi$ with the signal $\D\Gama$.
\end{thm}

At the end of Section \ref{Sec:Global2Local} we mentioned that for $m\gg 1$, the classical analysis would allow an order of $O(\sqrt{n})$ non-zeros all over the vector $\Gama$, regardless of the length of the signal $N$. In light of the above theorem, in the convolutional case, the very same quantity of non-zeros is allowed locally per stripe, implying that the overall number of non-zeros in $\Gama$ grows linearly with the global dimension $N$.

%% ---------------------------------------------------------------------------------------------------------------
\subsection{Recovery Guarantees for Pursuit Methods}
\label{Sec:BPnoiseless}

In this subsection, we attempt to solve the $\Poi$ problem by employing two common, but very different, pursuit methods: the Orthogonal Matching Pursuit (OMP) and the Basis Pursuit (BP) -- the reader is referred to \cite{Elad_Book} for a detailed description of these formulations and respective algorithms. Leaving aside the computational burdens of running such algorithms, which will be addressed in the second part of this work, we now consider the theoretical aspects of their success.

Previous works \cite{Donoho2003,Tropp2004} have shown that both OMP and BP succeed in finding the sparsest solution to the $P_0$ problem if the cardinality of the representation is known a priori to be lower than $\frac{1}{2}(1+\frac{1}{\mu(\D)})$. That is, we are guaranteed to recover the underlying solution as long as the \emph{global sparsity} is less than a certain threshold. In light of the discussion in Section \ref{GlobalVsLocal}, these values are pessimistic in the convolutional setting. By migrating from $P_0$ to the $\Poi$ problem, we show next that both algorithms are in fact capable of recovering the underlying solutions under far weaker assumptions.

\begin{thm}{(Global OMP recovery guarantee using $\Loi$ norm):} \label{thm:OMPSuccess}
	Given the system of linear equations $\X = \D\Gama$, if a solution $\Gama$ exists satisfying
	\begin{equation} \label{eq:Loi_condition}
	\|\Gama\|_{0,\infty} < \ \frac{1}{2}\left(1+\frac{1}{\mu(\D)}\right),
	\end{equation}
	then OMP is guaranteed to recover it.
\end{thm}
Note that if we assume $\|\Gama\|_{0,\infty} < \ \frac{1}{2}\left(1+\frac{1}{\mu(\D)}\right)$, according to our uniqueness theorem, the solution obtained by the OMP is the unique solution to the $\Poi$ problem. Interestingly, under the same conditions the BP algorithm is guaranteed to succeed as well. 
\begin{thm}{(Global Basis Pursuit recovery guarantee using the $\Loi$ norm):} \label{thm:BPSuccess}
	For the system of linear equations $\D\Gama=\X$, if a solution $\Gama$ exists obeying
	\begin{equation}
		\|\Gama\|_{0,\infty}<\frac{1}{2}\left(1+\frac{1}{\mu(\D)}\right),
	\end{equation}
	then Basis Pursuit is guaranteed to recover it.
\end{thm}

The recovery guarantees for both pursuit methods have now become \emph{independent of the global signal dimension and sparsity}. Instead, the condition for success is given in terms of the \emph{local} concentration of non-zeros of the global sparse vector. Moreover, the number of non-zeros allowed per stripe under the current bounds is in fact the same number previously allowed globally. As a remark, note that we have used these two algorithms in their natural form, being oblivious to the $\Loi$ objective they are serving. Further work is required to develop OMP and BP versions that are aware of this specific goal, potentially benefiting from it.  

\subsection{Experiments}
In this subsection we intend to provide numerical results that corroborate the above presented theoretical bounds. While doing so, we will shed light on the performance of the OMP and BP algorithms in practice, as compared to our previous analysis.

In \cite{Soltanalian2014} an algorithm was proposed to construct a local dictionary such that all its aperiodic auto-correlations and cross-correlations are low. This, in our context, means that the algorithm attempts to minimize the mutual coherence of the dictionary $\D_L$ and all of its shifts, decreasing the global mutual coherence as a result. We use this algorithm to numerically build a dictionary consisting of two atoms ($m=2$) with patch size $n=64$. The theoretical lower bound on the $\mu(\D)$ presented in Equation \eqref{eq:mu_bound} under this setting is approximately $0.063$, and we manage to obtain a mutual coherence of $0.09$ using the aforementioned method. With these atoms we construct a convolutional dictionary with global atoms of length $N = 640$.

Once the dictionary is fixed, we generate sparse vectors with random supports of (global) cardinalities in the range $\left[1,300\right]$. The non-zero entries are drawn from random independent and identically-distributed Gaussians with mean equal to zero and variance equal to one. Given these sparse vectors, we compute their corresponding global signals and attempt to recover them using the global OMP and BP. We perform $500$ experiments per each cardinality and present the probability of success as a function of the representation's $\Loi$ norm. We define the success of the algorithm as the full recovery of the true sparse vector. The results for the experiment are presented in Figure \ref{fig:phase_transition_both}. The theorems provided in the previous subsection guarantee the success of both OMP and BP as long as the $\|\Gama \|_{0,\infty} \leq 6$. 
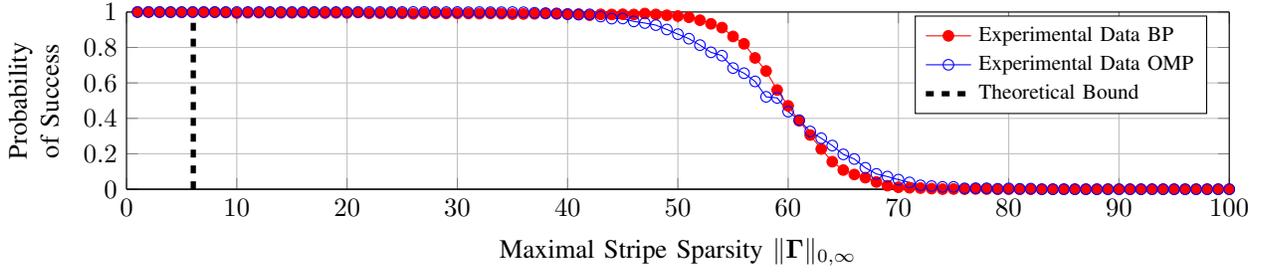
\begin{figure*}[t]%[!htbp]
	\vspace{-.5cm}
	\centering
	\setlength\figureheight{0.13\textwidth}
	\setlength\figurewidth{0.85\textwidth}
\begin{tikzpicture}

\begin{axis}[%
width=0.951\figurewidth,
height=\figureheight,
at={(0\figurewidth,0\figureheight)},
scale only axis,
separate axis lines,
every outer x axis line/.append style={black},
every x tick label/.append style={font=\color{black}},
xmin=0,
xmax=100,
xlabel={Maximal Stripe Sparsity $\|\Gama \|_{0,\infty}$},
xmajorgrids,
every outer y axis line/.append style={black},
every y tick label/.append style={font=\color{black}},
ymin=0,
ymax=1,
ylabel={Probability of Success},
ylabel style={align=center,text width=3cm},
ymajorgrids,
axis background/.style={fill=white},
legend style={legend cell align=left,align=left,draw=black}
]

\addplot[color=red,solid,mark=*,mark options={solid}] plot table[row sep=crcr] {%
1	1\\
2	1\\
3	0.999250374812594\\
4	0.998520710059172\\
5	0.997893258426966\\
6	0.999328408327737\\
7	0.998756991920448\\
8	0.998751560549313\\
9	0.997043169722058\\
10	0.994051160023795\\
11	0.99578567128236\\
12	0.996468510888758\\
13	0.997116493656286\\
14	0.995505617977528\\
15	0.995614035087719\\
16	0.996095928611266\\
17	0.996694214876033\\
18	0.994114499732477\\
19	0.995565410199557\\
20	0.994612068965517\\
21	0.997251236943375\\
22	0.992972972972973\\
23	0.99354491662184\\
24	0.994324045407637\\
25	0.993036957686127\\
26	0.990231362467866\\
27	0.99238578680203\\
28	0.993407707910751\\
29	0.991979949874687\\
30	0.992150706436421\\
31	0.991778006166495\\
32	0.990950226244344\\
33	0.992401215805471\\
34	0.990673575129534\\
35	0.988174807197943\\
36	0.988939165409754\\
37	0.989717223650386\\
38	0.991517436380773\\
39	0.99000999000999\\
40	0.988977955911824\\
41	0.988301119023398\\
42	0.982073643410853\\
43	0.986973947895792\\
44	0.986124876114965\\
45	0.987317073170732\\
46	0.985983566940551\\
47	0.991546494281452\\
48	0.986342943854325\\
49	0.981629769194536\\
50	0.97735460627895\\
51	0.970366886171214\\
52	0.953748782862707\\
53	0.9335232668566\\
54	0.912030432715169\\
55	0.862102217936355\\
56	0.82055063913471\\
57	0.741045498547919\\
58	0.667281956622058\\
59	0.559924206537186\\
60	0.470616570327553\\
61	0.386941910705713\\
62	0.305853658536585\\
63	0.226923076923077\\
64	0.155466399197593\\
65	0.108235294117647\\
66	0.0830564784053156\\
67	0.0652372262773723\\
68	0.041018387553041\\
69	0.0202764976958525\\
70	0.0111867704280156\\
71	0.00831024930747922\\
72	0.00648748841519926\\
73	0.0034330554193232\\
74	0.000525762355415352\\
75	0.00113250283125708\\
76	0.000628930817610063\\
77	0.00151285930408472\\
78	0.000956022944550669\\
79	0\\
80	0\\
81	0\\
82	0\\
83	0\\
84	0\\
85	0\\
86	0\\
87	0\\
88	0\\
89	0\\
90	0\\
91	0\\
92	0\\
93	0\\
94	0\\
95	0\\
96	0\\
97	0\\
98	0\\
99	0\\
100	0\\
};
\addlegendentry{\footnotesize{Experimental Data BP}};

\addplot[color=blue,solid,mark=o,mark options={solid}] plot table[row sep=crcr] {%
1	1\\
2	1\\
3	1\\
4	1\\
5	1\\
6	1\\
7	1\\
8	1\\
9	1\\
10	1\\
11	1\\
12	1\\
13	1\\
14	1\\
15	1\\
16	1\\
17	1\\
18	1\\
19	1\\
20	1\\
21	1\\
22	1\\
23	1\\
24	1\\
25	1\\
26	1\\
27	1\\
28	1\\
29	0.999478079331942\\
30	1\\
31	0.99897066392177\\
32	0.999479708636837\\
33	0.999003984063745\\
34	0.998980112187659\\
35	0.999491353001017\\
36	0.998481781376518\\
37	0.995529061102832\\
38	0.993496748374187\\
39	0.991826923076923\\
40	0.987593052109181\\
41	0.985074626865672\\
42	0.986875315497224\\
43	0.974531475252283\\
44	0.962776659959759\\
45	0.963321446765155\\
46	0.947072599531616\\
47	0.937806072477963\\
48	0.926133469179827\\
49	0.901277013752456\\
50	0.875121477162293\\
51	0.84952380952381\\
52	0.813233223838573\\
53	0.772528007793473\\
54	0.753664302600473\\
55	0.683736367946894\\
56	0.654767726161369\\
57	0.608381502890173\\
58	0.52144578313253\\
59	0.515503875968992\\
60	0.438036224976168\\
61	0.389518413597734\\
62	0.325969563082965\\
63	0.287747839349263\\
64	0.246468582562104\\
65	0.196759259259259\\
66	0.170582706766917\\
67	0.120941176470588\\
68	0.0865518845974872\\
69	0.071360153256705\\
70	0.0542778288868445\\
71	0.0386847195357834\\
72	0.0219728845254792\\
73	0.0175867122618466\\
74	0.0132382892057026\\
75	0.013172338090011\\
76	0.0051150895140665\\
77	0.00593471810089021\\
78	0.00552995391705069\\
79	0.00238095238095238\\
80	0.00470957613814757\\
81	0.00210526315789474\\
82	0.00335570469798658\\
83	0\\
84	0\\
85	0\\
86	0\\
87	0\\
88	0\\
89	0\\
90	0\\
91	0\\
92	0\\
93	0\\
94	0\\
95	0\\
96	0\\
97	0\\
98	0\\
99	0\\
100	0\\
};
\addlegendentry{\footnotesize{Experimental Data OMP}};

\addplot [color=black,solid,forget plot]
  table[row sep=crcr]{%
0	0\\
100	0\\
};
\addplot [color=black,dashed,line width=2.0pt]
  table[row sep=crcr]{%
6.05373654502998	0\\
6.05373654502998	1\\
};
\addlegendentry{\footnotesize{Theoretical Bound}};

\end{axis}
\end{tikzpicture}%
	\caption{Probability of success of OMP and BP at recovering the true convolutional sparse code. The theoretical guarantee is presented on the same graph.}
	\label{fig:phase_transition_both}
	\vspace{-0.3cm}
\end{figure*}

As can be seen from these results, the theoretical bound is far from being tight. However, in the traditional sparse representation model the corresponding bounds have the same loose flavor \cite{Bruckstein2009}. This kind of results is in fact expected when using such a worst-case analysis. Tighter bounds could likely be obtained by a probabilistic study, which we leave for future work.

%This bound, though far from the empirical performance of both pursuit methods, is in fact meaningful. Note that the traditional bound relying on the global number of non-zeros guarantees the recovery of a sparse vector having up to 6 non-zeros. Our bound, on the other hand, guarantees the recovery of signals with up to 30 non-zeros\footnote{There are $\frac{N}{2n-1}$ non-overlapping stripes in the global sparse vector. Hence, one can allow up to $6\cdot \frac{N}{2n-1} \approx 30$ non-zeros in this particular case.}, for this value on $N$. The higher the global dimension, the bigger the gap between these bounds will be. For example, if $N = 6400$, we could then guarantee up to 300 non-zeros (compare this with the previous bound of 6).

\section{Shifted Mutual Coherence and Stripe Coherence}
\label{Sec:ShiftedMutualCoherence}
When considering the mutual coherence $\mu(\D)$, one needs to look at the maximal correlation between every pair of atoms in the global dictionary. One should note, however, that atoms having a non-zero correlation must have overlapping supports, and $\mu(\D)$ provides a bound for these values independently of the amount of overlap.
One could go beyond this characterization of the convolutional dictionary by a single value and propose to bound all the inner products between atoms for a \emph{given shift}. As a motivation, in several applications one can assume that signals are built from local atoms separated by some minimal lag, or shift. In radio communications, for example, such a situation appears when there exists a minimal time between consecutive transmissions on the same channel \cite{He2009}. In such cases, knowing how the correlation between the atoms depends on their shifts is fundamental for the design of the dictionary, its utilization and its theoretical analysis. 

In this section we briefly explore this direction of analysis, introducing a stronger characterization of the convolutional dictionary, termed shifted mutual coherence. By being a considerably more informative measure than the standard mutual coherence, this will naturally lead to stronger bounds. We will only present the main points of these results here for the sake of brevity; the interested reader can find a more detailed discussion on this matter in the Supplementary Material.

Recall that $\O$ is defined as a stripe extracted from the global dictionary $\D$. Consider the sub-system given by $\x_i=\O\gama_i$, corresponding to the $i^{th}$ patch in $\X$. Note that $\O$ can be split into a set of $2n-1$ blocks of size $n\times m$, where each block is denoted by $\O_s$, i.e., 
\begin{equation}
\O=[\O_{-n+1},\dots,\O_{-1},\O_0,\O_1,\dots,\O_{n-1}],
\end{equation}
as shown previously in Figure \ref{PartialStripe}. 
\begin{defn}
Define the shifted mutual coherence $\mu_s$ by
\begin{equation}
\mu_s = \underset{i,j}{\max} \quad |\langle \d^0_i,\d^s_j  \rangle|,
\end{equation}
where $\d^0_i$ is a column extracted from $\O_0$, $\d^s_j$ is extracted from $\O_s$, and we require\footnote{The condition $i\neq j$ if $s=0$ is necessary so as to avoid the inner product of an atom by itself.} that $i\neq j$ if $s=0$.
\end{defn}
\noindent
The above definition can be seen as a generalization of the mutual coherence for the shift-invariant local model presented in Section \ref{Sec:Conv_Model}. Indeed, $\mu_s$ characterizes $\O$ just as $\mu(\D)$ characterizes the coherence of a general dictionary. Note that if $s=0$ the above definition boils down to the traditional mutual coherence of $\D_L$, i.e., $\mu_0=\mu(\D_L)$. It is important to stress that the atoms used in the above definition {\em are normalized globally} according to $\D$ and not $\O$. In the Supplementary Material we comment on several interesting properties of this measure.

Similar to $\O$, $\gama_i$ can be split into a set of $2n-1$ vectors of length $m$, each denoted by $\gama_{i,s}$ and corresponding to $\O_s$. In other words, $\gama_i=[\gama_{i,-n+1}^T,\dots,\gama_{i,-1}^T,\gama_{i,0}^T,\gama_{i,1}^T,\dots,\gama_{i,n-1}^T]^T$. Note that previously we denoted local sparse vectors of length $m$ by $\alfa_j$. Yet, we will also denote them by $\gama_{i,s}$ in order to emphasize the fact that they correspond to the $s^{th}$ shift within $\gama_i$. Denote the number of non-zeros in $\gama_i$ as $n_i$. We can also write $n_i=\displaystyle\sum_{s=-n+1}^{n-1}n_{i,s}$, where $n_{i,s}$ is the number of non-zeros in each $\gama_{i,s}$. With these definitions, we can now propose the following measure.
\begin{defn}
Define the stripe coherence as
\begin{equation}
\zeta(\gama_i) = \sum_{s=-n+1}^{n-1} n_{i,s}\ \mu_s.
\end{equation}
\end{defn}
\noindent
According to this definition, each stripe has a coherence given by the sum of its non-zeros weighted by the shifted mutual coherence. As a particular case, if all $k$ non-zeros correspond to atoms in the center sub-dictionary, $\D_L$, this becomes $\mu_0k$. Note that unlike the traditional mutual coherence, this new measure depends on the location of the non-zeros in $\Gama$ -- it is a function of the support of the sparse vector, and not just of the dictionary. As such, it characterizes the correlation between the atoms participating in a given stripe. In what follows, we will use the notation $\zeta_i$ for $\zeta(\gama_i)$.

Having formalized these tighter constructions, we now leverage them to improve the previous results. Although these theorems are generally sharper, they are harder to grasp. We begin with a recovery guarantee for the OMP and BP algorithms, followed by a discussion on their implications.

\begin{thm}{(Global OMP recovery guarantee using the stripe coherence):}
	\label{thm:OMPSuccess_stripeCoherence}
	Given the system of linear equations $\X = \D\Gama$, if a solution $\Gama$ exists satisfying \begin{equation} \label{eq:stripe_coherence_condition3}
	\max_i\ \zeta_i = \max_i\sum_{s=-n+1}^{n-1} n_{i,s} \mu_s < \ \frac{1}{2}\left(1+\mu_0\right),
	\end{equation}
	then OMP is guaranteed to recover it.
\end{thm}

\begin{thm}{(Global BP recovery guarantee using the stripe coherence):}
	Given the system of linear equations $\X = \D\Gama$, if a solution $\Gama$ exists satisfying \begin{equation}
	\max_i\ \zeta_i = \max_i\sum_{s=-n+1}^{n-1} n_{i,s} \mu_s < \ \frac{1}{2}\left(1+\mu_0\right),
	\end{equation}
	then Basis Pursuit is guaranteed to recover it.
\end{thm}
\noindent The corresponding proofs are similar to their counterparts presented in the preceding section but require a more delicate analysis; one of them is thoroughly discussed in the Supplementary Material. 

In order to provide an intuitive interpretation for these results, the above bounds can be tied to a concrete number of non-zeros per stripe. First, notice that requiring the maximal stripe coherence to be less than a certain threshold is equal to requiring the same for every stripe:
\begin{equation}
\forall i\quad\sum_{s=-n+1}^{n-1} n_{i,s} \mu_s< \ \frac{1}{2}\left(1+\mu_0\right).
\end{equation}
Multiplying and dividing the left-hand side of the above inequality by $n_i$ and rearranging the resulting expression, we obtain
\begin{equation}
\forall i \quad n_i< \ \frac{1}{2}\frac{1+\mu_0}{\sum_{s=-n+1}^{n-1} \frac{n_{i,s}}{n_i} \mu_s }.
\end{equation}
Define $\bar{\mu}_i=\sum_{s=-n+1}^{n-1}\frac{n_{i,s}}{n_i} \mu_s$. Recall that $\sum_{s=-n+1}^{n-1} \frac{n_{i,s}}{n_i}=1$ and as such $\bar{\mu}_i$ is simply the (weighted) average shifted mutual coherence in the $i^{th}$ stripe. Putting this definition into the above condition, the inequality becomes
\begin{equation}
\forall i\quad n_i< \frac{1}{2}\left(\frac{1}{\bar{\mu}_i}+\frac{\mu_0}{\bar{\mu}_i}\right).
\end{equation}
Thus, the condition in \eqref{eq:stripe_coherence_condition3} boils down to requiring the sparsity of all stripes to be less than a certain number. Naturally, this inequality resembles the one presented in the previous section for the OMP and BP guarantees. In the Supplementary Material we prove that under the assumption that $\mu(\D)=\mu_0$, the shifted mutual coherence condition is at least as strong as the original one. 

%% ---------------------------------------------------------------------------------------------------------------
%% ---------------------------------------------------------------------------------------------------------------

\section{From Global to Local Stability Analysis} \label{Sec:NoisyRegime}
One of the cardinal motivations for this work was a series of recent practical methods addressing the convolutional sparse
coding problem; and in particular, the need for their theoretical foundation. However, our results are as of yet not directly applicable to these, as we have restricted our analysis to the ideal case of noiseless signals. This is the path we undertake in the following sections, exploring the question of whether the convolutional model remains stable in the presence of noise.

Assume a clean signal $\X$, which admits a sparse representation $\Gama$ in terms of the convolutional dictionary $\D$, is contaminated with noise $\E$ (of bounded energy, $\| \E \|_2 \leq \epsilon$)  to create $\Y=\D\Gama+\E$. Given this noisy signal, one could propose to recover the true representation $\Gama$, or a vector close to it, by solving the $P_0^\epsilon$ problem. In this context, as mentioned in the previous section, several theoretical guarantees have been proposed in the literature. As an example, consider the stability results presented in the seminal work of \cite{Donoho2006}. Therein, it was shown that assuming the total number of non-zeros in $\Gama$ is less than $\frac{1}{2}\left(1+\frac{1}{\mu(\D)}\right)$, the distance between the solution to the $P_0^\epsilon$ problem, $\overline{\Gama}$, and the true sparse vector, $\Gama$, satisfies
\begin{equation}
	\|\overline{\Gama} - \Gama\|_2^2 \leq \frac{4\epsilon^2}{1-\mu(\D)(2\|\Gama\|_0-1)}.
	\label{Eq:OriginalStability}
\end{equation}
In the context of our convolutional setting, however, this result provides a weak bound as it constrains the total number of non-zeros to be below a certain threshold, which scales with the local filter size $n$. 

We now re-define the $P_0^{\epsilon}$ problem into different one, capturing the convolutional structure by relying on the $\Loi$ norm instead. Consider the problem:
\begin{equation}
\quad (\Poie): \quad \underset{\Gama}{\min} \quad \|\Gama\|_{0,\infty} \ \text{ s.t. }\ \|\Y-\D\Gama\|_2\leq\epsilon.
\end{equation}
In words, given a noisy measurement $\Y$, we seek for the $\Loi$-sparsest representation vector that explains this signal up to an $\epsilon$ error. In what follows, we address the theoretical aspects of this problem and, in particular, study the stability of its solutions and practical yet secured ways for retrieving them.

%% ---------------------------------------------------------------------------------------------------------------
%% ---------------------------------------------------------------------------------------------------------------
\section{Theoretical Analysis of Corrupted Signals}
\label{Sec:TheoreticalAnalysis_noisy}

\subsection{Stability of the $\Poie$ Problem}
As expected, one cannot guarantee the uniqueness of the solution to the $\Poie$ problem, as was done for the $\Poi$. Instead, in this subsection we shall provide a stability claim that guarantees the found solution to be close to the underlying sparse vector that generated $\Y$. In order to provide such an analysis, we commence by arming ourselves with the necessary mathematical tools. 

\begin{defn}
Let $\D$ be a convolutional dictionary. Consider all the sub matrices $\D_\mathcal{T}$, obtained by restricting the dictionary $\D$ to a support $\mathcal{T}$ with an $\Loi$ norm equal to $k$. Define $\delta_k$ as the smallest quantity such that
\begin{equation}
\forall \Delt \quad (1-\delta_k)\|\Delt\|_2^2\leq\|\D_\mathcal{T} \Delt\|_2^2\leq(1+\delta_k)\|\Delt\|_2^2
\end{equation}
holds true for any choice of the support. Then, $\D$ is said to satisfy $k$-SRIP (Stripe-RIP) with constant $\delta_k$.
\end{defn}

Given a matrix $\D$, similar to the Stripe-Spark, computing the SRIP is hard or practically impossible. Thus bounding it using the mutual coherence is of practical use. 
\begin{thm}{(Upper bounding the SRIP via the mutual coherence):}
	For a convolutional dictionary $\D$ with global mutual coherence $\mu(\D)$, the SRIP can be upper-bounded by
	\begin{equation}
	\delta_k\leq(k-1)\mu(\D).
	\end{equation}
\end{thm}

Assume a sparse vector $\Gama$ is multiplied by $\D$ and then contaminated by a vector $\E$, generating the signal $\Y=\D\Gama+\E$, such that $\|\Y-\D\Gama\|_2^2\leq\epsilon^2$. Suppose we solve the $\Poie$ problem and obtain a solution $\hat{\Gama}$. How close is this solution to the original $\Gama$? The following theorem provides an answer to this question.

\begin{thm}{(Stability of the solution to the $\Poie$ problem):}
Consider a sparse vector $\Gama$ such that $\|\Gama\|_{0,\infty} = k < \frac{1}{2}\left( 1 + \frac{1}{\mu(\D)} \right) $, and a convolutional dictionary $\D$ satisfying the SRIP property for $\Loi=2k$ with coefficient $\delta_{2k}$. Then, the distance between the true sparse vector $\Gama$ and the solution to the $\Poie$ problem $\hat{\Gama}$ is bounded by
\begin{equation}
\|\Gama-\hat{\Gama}\|_2^2\leq \frac{4\epsilon^2}{1-\delta_{2k}}\leq\frac{4\epsilon^2}{1-(2k-1)\mu(\D)}.	
\label{Eq:NewStability}
\end{equation}
\end{thm}

One should wonder if the new guarantee presents any advantage when compared to the bound based on the traditional RIP. Looking at the original stability claim for the global system, as discussed in Section \ref{Sec:Global2Local}, 
%we obtain that $\|\Gama-\overline{\Gama}\|_2^2\leq\frac{4\epsilon^2}{1-(2\|\Gama\|_0-1)\mu(\D)}$, with $\overline{\Gama}$ being the solution to the $P_0^\epsilon$ problem (whereas $\hat{\Gama}$ is the solution to the $\Poie$ problem).
%In order to analyze the difference between both bounds, 
the reader should compare the assumptions on the sparse vector $\Gama$, as well as the obtained bounds on the distance between the estimates and the original vector. The stability claim in the $P_0^\epsilon$ problem is valid under the condition
\begin{equation}
\|\Gama\|_0<\frac{1}{2}\left(1+\frac{1}{\mu(\D)}\right).
\end{equation}
In contrast, the stability claim presented above holds whenever
\begin{equation}
\|\Gama\|_{0,\infty}<\frac{1}{2}\left(1+\frac{1}{\mu(\D)}\right).
\end{equation}
This allows for significantly more non-zeros in the global signal. Furthermore, as long as the above hold, and comparing Equations \eqref{Eq:OriginalStability} and \eqref{Eq:NewStability}, we have that
\begin{equation}
\frac{4\epsilon^2}{1-(2\|\Gama\|_{0,\infty}-1)\mu(\D)} \ll \frac{4\epsilon^2}{1-(2\|\Gama\|_0-1)\mu(\D)},
\end{equation}
since generally $\|\Gama\|_{0,\infty}\ll\|\Gama\|_0$. This inequality implies that the above developed bound is (usually much) lower than the traditional one. In other words, the bound on the distance to the true sparse vector is much tighter and far more informative under the $\Loi$ setting.

\vspace{-0.1cm}
\subsection{Stability Guarantee of OMP}
Hitherto, we have shown that the solution to the $\Poie$ problem will be close to the true sparse vector $\Gama$. However, it is also important to know whether this solution can be approximated by pursuit algorithms. In this subsection, we address such a question for the OMP, extending the analysis presented to the noisy setting.

In \cite{Donoho2006}, a claim was provided for the OMP, guaranteeing the recovery of the true support of the underlying solution if
\begin{equation}
	\|\Gama\|_0 < \frac{1}{2}\left(1+\frac{1}{\mu(\D)}\right) - \frac{1}{\mu(\D)}\cdot\frac{\epsilon}{|\Gamma_{min}|},
\end{equation}
$|\Gamma_{min}|$ being the minimal absolute value of a (non-zero) coefficients in $\Gama$. This result comes to show the importance of both the sparsity of $\Gama$ and the signal-to-noise ratio, which relates to the term ${\epsilon}/{|\Gamma_{min}|}$. In the context of our convolutional setting, this result provides a weak bound for two different reasons. First, the above bound restricts the total number of non-zeros in the representation of the signal. From Section \ref{Sec:TheoStudy}, it is natural to seek for an alternative condition for the success of this pursuit relying on the $\Loi$ norm instead. Second, notice that the rightmost term in the above bound divides the global error energy by the minimal coefficient (in absolute value) in $\Gama$. In the convolutional scenario, the energy of the error $\epsilon$ is a \textit{global} quantity, while the minimal coefficient $|\Gamma_{min}|$ is a \textit{local} one -- thus making this term enormous, and the corresponding bound nearly meaningless. As we show next, one can harness the inherent locality of the atoms in order to replace the global quantity in the numerator with a local one: $\epsilon_L$.

\begin{thm}{(Stable recovery of global OMP in the presence of noise):} \label{Theorem:StabilityOMP}
	Suppose a clean signal $\X$ has a representation $\D\Gama$, and that it is contaminated with noise $\E$ to create the signal $\Y=\X+\E$, such that $\|\Y-\X\|_2\leq\epsilon$. Denote by $\epsilon_{_L}$ the highest energy of all $n$-dimensional local patches extracted from $\E$. Assume $\Gama$ satisfies
	\begin{equation} \label{omp_hypothesis}
	\|\Gama\|_{0,\infty} < \frac{1}{2}\left( 1+\frac{1}{\mu(\D)} \right)-\frac{1}{\mu(\D)}\cdot\frac{\epsilon_{_L}}{|\Gamma_{min}|},
	\end{equation}
	where $|\Gamma_{min}|$ is the minimal entry in absolute value of the sparse vector $\Gama$.
	Denoting by $\Gama_\text{OMP}$ the solution obtained by running OMP for $\|\Gama\|_0$ iterations, we are guaranteed that
	\begin{enumerate}[ a) ]
	\item OMP will find the correct support; And,
	\item $\|\Gama_\text{OMP}-\Gama\|_2^2\leq\frac{\epsilon^2}{1-\mu(\D)(\|\Gama\|_{0,\infty}-1)}$.
	\end{enumerate}
\end{thm}

The proof of this theorem is presented in the Supplementary Material, and the derivations therein are based on the analysis presented in \cite{Donoho2006}, generalizing the study to the convolutional setting. Note that we have assumed that the OMP algorithm runs for $\|\Gama\|_0$ iterations. We could also propose a different approach, however, using a stopping criterion based on the norm of the residual. Under such setting, the OMP would run until the energy of the global residual is less than the energy of the noise, given by $\epsilon^2$. 

\vspace{-0.1cm}
\subsection{Stability Guarantee of Basis Pursuit Denoising via ERC}
 A theoretical motivation behind relaxing the $\Loi$ norm to the convex $\ell_1$ was already established in Section \ref{Sec:TheoStudy}, showing that if the former is low, the BP algorithm is guaranteed to succeed. When moving to the noisy regime, the BP is naturally extended to the Basis Pursuit DeNoising (BPDN) algorithm\footnote{Note that an alternative to the BPDN extension is that of the Dantzig Selector algorithm. One can envision a similar analysis to the one presented here for this algorithm as well.}, which in its Lagrangian form is defined as follows
\begin{equation} \label{eq:lagrangian_BP}
\min_{\Gama} \frac{1}{2}\|\Y - \D \Gama \|^2_2 +\lambda \| \Gama \|_1.
\end{equation}
Similar to how BP was shown to approximate the solution to the $\Poi$ problem, in what follows we will prove that the BPDN manages to approximate the solution to the $\Poie$ problem.

Assuming the ERC is met, the stability of BP was proven under various noise models and formulations in \cite{Tropp2006}. By exploiting the convolutional structure used throughout our analysis, we now show that the ERC is met given that the $\Loi$ norm is small, tying the aforementioned results to our story.

\begin{thm}{(ERC in the convolutional sparse model):}
\label{Theorem:ERC_Loi}
	For a convolutional dictionary $\D$ with mutual coherence $\mu(\D)$, the ERC condition is met for every support $\mathcal{T}$ that satisfies
	\begin{equation}
		\|\mathcal{T}\|_{0,\infty} < \frac{1}{2}\left(1+\frac{1}{\mu(\D)}\right).
	\end{equation}
\end{thm}

Based on this and the analysis presented in \cite{Tropp2006}, we present a stability claim for the Lagrangian formulation of the BP problem as stated in Equation \eqref{eq:lagrangian_BP}.
\begin{thm}{(Stable recovery of global Basis Pursuit in the presence of noise):} \label{Theorem:StabilityBP}
	Suppose a clean signal $\X$ has a representation $\D\Gama$, and that it is contaminated with noise $\E$ to create the signal $\Y=\X+\E$. Denote by $\epsilon_{_L}$ the highest energy of all $n$-dimensional local patches extracted from $\E$. Assume $\Gama$ satisfies
	\begin{equation} \label{eq:BP_assumption}
	\|\Gama\|_{0,\infty}\leq\frac{1}{3} \left( 1 + \frac{1}{\mu(\D)} \right).
	\end{equation}
	Denoting by $\Gama_{\text{BP}}$ the solution to the Lagrangian BP formulation with parameter $\lambda=4\epsilon_L$, we are guaranteed that
	\begin{enumerate}
	\item The support of $\Gama_{\text{BP}}$ is contained in that of $\Gama$.
	\item $\|\Gama_{\text{BP}}-\Gama\|_\infty<\frac{15}{2}\epsilon_L$.
	\item In particular, the support of $\Gama_{\text{BP}}$ contains every index $i$ for which $|\Gamma_i|>\frac{15}{2}\epsilon_L$.
	\item The minimizer of the problem, $\Gama_{\text{BP}}$, is unique.
	\end{enumerate}
\end{thm}
\noindent
The proof for both of the above, inspired by the derivations in \cite{Elad_Book} and \cite{Tropp2006}, are presented in the Supplementary Material.

The benefit of this over traditional claims is, once again, the replacement of the $\ell_0$ with the $\Loi$ norm. Moreover, this result bounds the difference between the entries in $\Gama_{\text{BP}}$ and $\Gama$ in terms of a local quantity -- the local noise level $\epsilon_L$. As a consequence, all atoms with coefficients above this local measure are guaranteed to be recovered.

The implications of the above theorem are far-reaching as it provides a sound theoretical back-bone for all works that have addressed the convolutional BP problem in its Lagrangian form \cite{Bristow2013,Wohlberg2016,Bristow2014,Heide2015,Kong2014}. In Section \ref{Sec:FromGlobal2LocalProcessing} we will propose two additional algorithms for solving the global BP efficiently by working locally, and these methods would benefit from this theoretical result as well. As a last comment, a different and perhaps more appropriate convex relaxation for the $\Loi$ norm could be suggested, such as the $\ell_{1,\infty}$ norm. This, however, remains one of our future work challenges. 

\subsection{Experiments}

Following the above analysis, we now provide a numerical experiment demonstrating the above obtained bounds. The global dictionary employed here is the same as the one used for the noiseless experiments in Section \ref{Sec:TheoStudy}, with mutual coherence $\mu(\D)=0.09$, local atoms of length $n=64$ and global ones of size $N = 640$. We sample random sparse vectors with cardinality between $1$ and $500$, with entries drawn from a uniform distribution with range $\left[-a,a\right]$, for varying values of $a$. Given these vectors, we construct global signals and contaminate them with noise. The noise is sampled from a zero-mean unit-variance white Gaussian distribution, and then normalized such that $\| \E \|_2 = 0.1$.

%\begin{figure*}[t]
%	\centering
%	\includegraphics[trim = 50 0 50 0, width = .85\textwidth]{graphics/PhaseTransition_OMP_Noisy.pdf}
%	\caption{The ratio $\epsilon_{_L}/|\Gama_{\min}|$ as a function of the $\Loi$ norm, and the theoretical bound for the successful recovery of the support, for the OMP algorithm.}
%	\label{fig:PhaseTransition_OMP_Noisy}
%\end{figure*}
%\begin{figure*}[t]
%	\centering
%	\includegraphics[trim = 50 0 50 0, width = .85\textwidth]{graphics/PhaseTransition_BP_Noisy.pdf}
%	\caption{The ratio $\epsilon_{_L}/|\Gama_{\min}|$ as a function of the $\Loi$ norm, and the theoretical bound for the successful recovery of the support, for the BP algorithm.}
%	\label{fig:PhaseTransition_BP_Noisy}
%\end{figure*}
\begin{figure}[t]
\centering
\includegraphics[trim = 50 0 30 20, width = 0.35\textwidth]{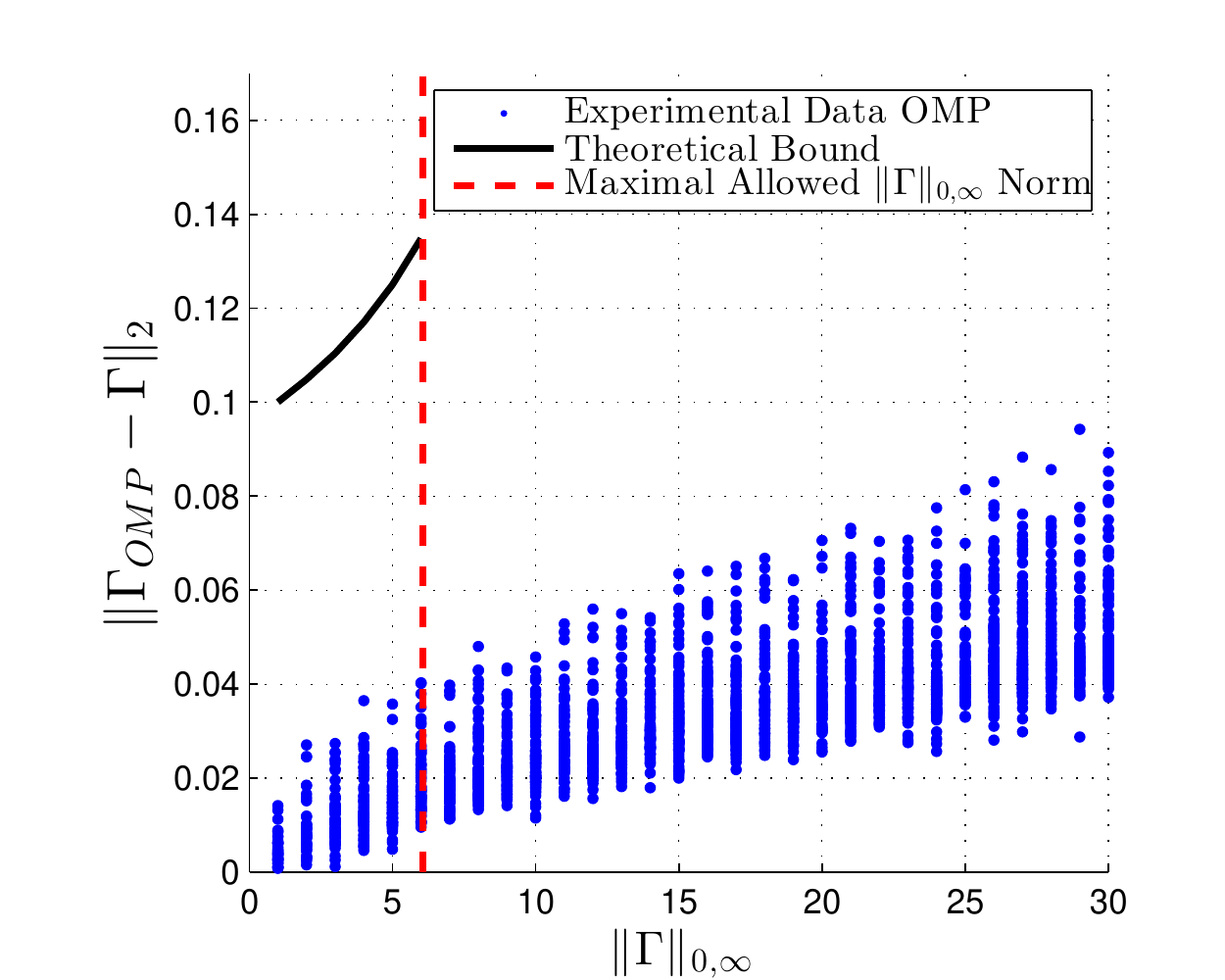}
\caption{The distance $\|\Gama_{\text{OMP}} - \Gama\|_2$ as a function of the $\Loi$ norm, and the corresponding theoretical bound.}
\label{fig:L2Dist_OMP_Noisy}
\vspace{-0.2cm}
\end{figure}

In what follows, we will first center our attention on the bounds obtained for the OMP algorithm, and then proceed to the ones corresponding to the BP. Given the noisy signals, we run OMP with a sparsity constraint, obtaining $\Gama_{\text{OMP}}$. For each realization of the global signal, we compute the minimal entry (in absolute value) of the global sparse vector, $|\Gamma_{min}|$, and its $\Loi$ norm. In addition, we compute the maximal local energy of the noise, $\epsilon_L$, corresponding to the highest energy of a $n$-dimensional patch of $\E$.

Recall that the theorem in the previous subsection poses two claims: 1) the stability of the result in terms of $\|\Gama_{\text{OMP}} - \Gama\|_2$; and 2) the success in recovering the correct support. In Figure \ref{fig:L2Dist_OMP_Noisy} we investigate the first of these points, presenting the distance between the estimated and the true sparse codes as a function of the $\Loi$ norm of the original vector. As it is clear from the graph, the empirical distances are below the theoretical bound depicted in black, given by $\frac{\epsilon^2}{1-\mu(\D)(\|\Gama\|_{0,\infty}-1)}$. According to the theorem's assumption, the sparse vector should satisfy $\|\Gama\|_{0,\infty} < \frac{1}{2}\left(1 + \frac{1}{\mu(\D)} \right) - \frac{1}{\mu(\D)}\cdot\frac{\epsilon_{_L}}{|\Gamma_{\min}|}$. The red dashed line delimits the area where this is met, with the exception that we omit the second term in the previous expression, as done previously in \cite{Donoho2006}. This disregards the condition on the $|\Gamma_{\min}|$ and $\epsilon_{_L}$ (which depends on the realization). Yet, the empirical results remain stable. 

In order to address the successful recovery of the support, we compute the ratio $\frac{\epsilon_{_L}}{|\Gamma_{\min}|}$ for each realization in the experiment. In Figure \ref{fig:PhaseTransition_OMP_Noisy}, for each sample we denote by $\bullet$ or $\times$ the success or failure in recovering the support, respectively. Each point is plotted as a function of its $\Loi$ norm and its corresponding ratio. The theoretical condition for the success of the OMP can be rewritten as $\frac{\epsilon_{_L}}{|\Gamma_{min}|} < \frac{\mu(\D)}{2}\left( 1+\frac{1}{\mu(\D)} \right) - \mu(\D)\|\Gama\|_{0,\infty}$, presenting a bound on the ratio $\frac{\epsilon_{_L}}{|\Gamma_{\min}|}$ as a function of the $\Loi$ norm. This bound is depicted with a blue line, indicating that the empirical results agree with the theoretical claims.

\begin{figure}[t]
	\centering
	\begin{subfigure}[t]{.5\textwidth}
		\centering
		\includegraphics[trim = 15 0 15 0,width = 1\textwidth]{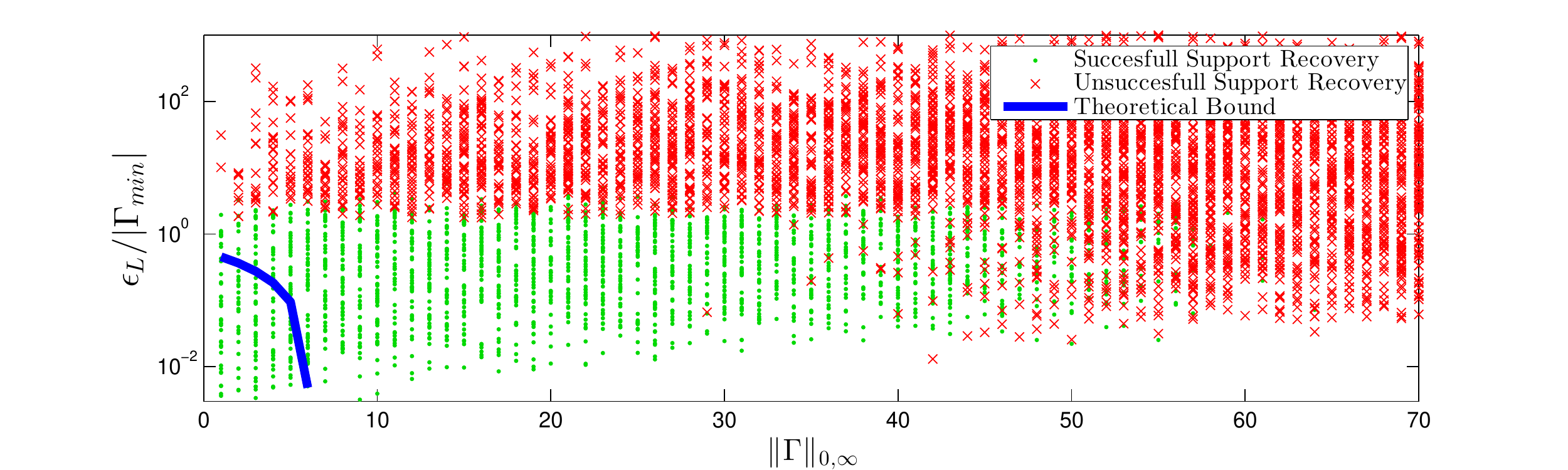}
		\caption{Orthogonal Matching Pursuit.}
		\label{fig:PhaseTransition_OMP_Noisy}
	\end{subfigure} 
	\\[.15cm]
	\begin{subfigure}[t]{.5\textwidth}
		\centering
		\includegraphics[trim = 15 0 15 0,width = 1\textwidth]{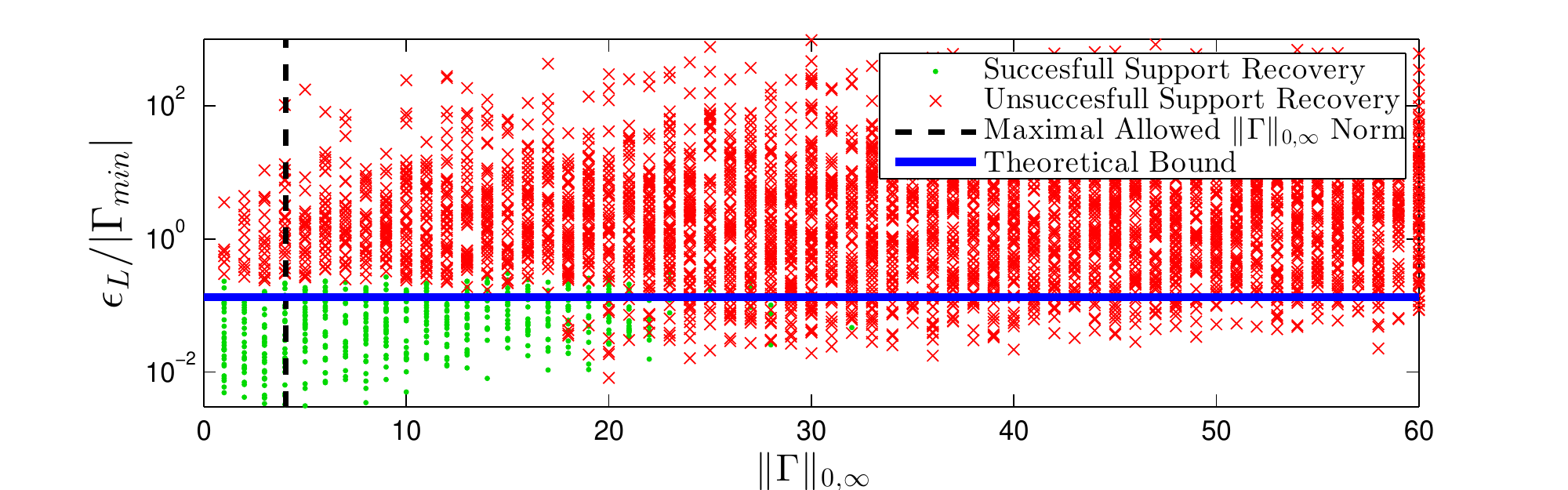}
		\caption{Basis Pursuit.}
		\label{fig:PhaseTransition_BP_Noisy}
	\end{subfigure}
	\caption{The ratio $\epsilon_{_L}/|\Gamma_{\min}|$ as a function of the $\Loi$ norm, and the theoretical bound for the successful recovery of the support, for both the OMP and BP algorithms.}
	\vspace{-0.3cm}
\end{figure}

One can also observe two distinct phase transitions in Figure \ref{fig:PhaseTransition_OMP_Noisy}. On the one hand, noting that the $y$ axis can be interpreted as the inverse of the noise-to-signal ratio (in some sense), we see that once the noise level is too high, OMP fails in recovering the support\footnote{Note that the abrupt change in this phase-transition area is due to the log scale of the $y$ axis.}. On the other hand, similar to what was presented in the noiseless case, once the $\Loi$ norm becomes too large, the algorithm is prone to fail in recovering the support.

We now shift to the empirical verification of the guarantees obtained for the BP. We employ the same dictionary as in the experiment above, and the signals are constructed in the same manner. We use the implementation of the LARS algorithm within the SPAMS package\footnote{Freely available from http://spams-devel.gforge.inria.fr/.} in its Lagrangian formulation with the theoretically justified parameter $\lambda=4\epsilon_L$, obtaining $\Gama_{\text{BP}}$. Once again, we compute the quantities: $|\Gamma_{min}|$, $\|\Gama\|_{0,\infty}$ and $\epsilon_L$.

Theorem \ref{Theorem:StabilityBP} states that the $\ell_\infty$ distance between the BP solution and the true sparse vector is below $\frac{15}{2}\epsilon_L$. In Figure \ref{fig:LInfDist_BP_Noisy} we depict the ratio $\frac{\|\Gama_{\text{BP}}-\Gama\|_\infty}{\epsilon_L}$ for each realization, verifying it is indeed below $\frac{15}{2}$ as long as the $\Loi$ norm is below $\frac{1}{3}\left(1+\frac{1}{\mu(\D)}\right) \approx 4$. Next, we would like to corroborate the assertions regarding the recovery of the true support. To this end, note that the theorem guarantees that all entries satisfying $|\Gamma_i|>\frac{15}{2}\epsilon_L$ shall be recovered by the BP algorithm. Alternatively, one can state that the complete support must be recovered as long as $\frac{\epsilon_L}{|\Gamma_{\min}|}<\frac{2}{15}$. To verify this claim, we plot this ratio for each realization as function of the $\Loi$ norm in Figure \ref{fig:PhaseTransition_BP_Noisy}, marking every point according to the success or failure of BP (in recovering the complete support).
As evidenced in \cite{Elad_Book}, OMP seems to be far more accurate than the BP in recovering the true support. As one can see by comparing Figure \ref{fig:PhaseTransition_OMP_Noisy} and \ref{fig:PhaseTransition_BP_Noisy}, BP fails once the $\Loi$ norm goes beyond $20$, while OMP succeeds all the way until $\|\Gama\|_{0,\infty}=40$. 

%% ---------------------------------------------------------------------------------------------------------------
%% ---------------------------------------------------------------------------------------------------------------
\section{From Global Pursuit to Local Processing}
\label{Sec:FromGlobal2LocalProcessing}

We now turn to analyze the practical aspects of solving the $\Poie$ problem given the relationship $\Y = \D\Gama + \E$. Motivated by the theoretical guarantees of success derived in the previous sections, the first na\"ive approach would be to employ global pursuit methods such as OMP and BP. However, these are computationally demanding as the dimensions of the convolutional dictionary are prohibitive for high values of $N$, the signal length.  

As an alternative, one could attempt to solve the $\Poie$ problem using a patch-based processing scheme. In this case, for example, one could suggest to solve a local and relatively cheaper pursuit for every patch in the signal (including overlaps) using the local dictionary $\D_L$. It is clear, however, that this approach will not work well under the convolutional model, because atoms used in overlapping patches are simply not present in $\D_L$. On the other hand, one could turn to employ $\O$ as the \emph{local} dictionary, but this is prone to fail in recovering the correct support of the atoms. To see this more clearly, note that there is no way to distinguish between any of the atoms having only one entry different than zero; i.e., those appearing on the extremes of $\O$ in Figure \ref{PartialStripe}.

As we can see, neither the na\"ive global approach, nor the simple patch-based processing, provide an effective strategy. Several questions arise from this discussion: Can we solve the global pursuit problem using local patch-based processing? Can the proposed algorithm rely merely on the low dimensional dictionaries $\D_L$ or $\O$ while still fully solving the global problem? If so, in what form should the local patches communicate in order to achieve a global consensus? In what follows, we address these issues and provide practical and globally optimal answers.

\begin{figure}[t]
	\centering
	\includegraphics[trim = 50 0 30 20, width = 0.35\textwidth]{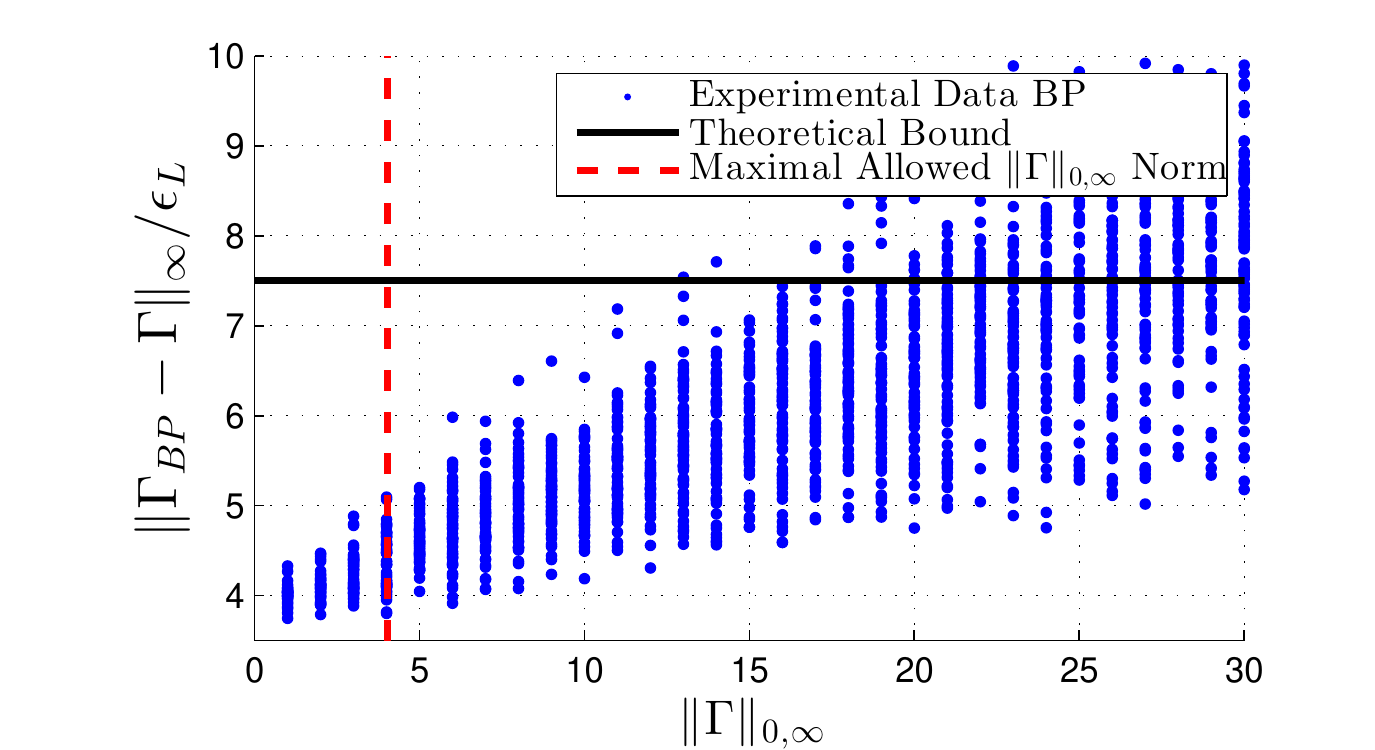}
	\caption{The distance $\|\Gama_{\text{BP}} - \Gama\|_\infty/ \epsilon_L$ as a function of the $\Loi$ norm, and the corresponding theoretical bound.}
	\label{fig:LInfDist_BP_Noisy}
	\vspace{-0.3cm}
\end{figure}

\vspace{-0.1cm}
\subsection{Global to Local Through Bi-Level Consensus}

When dealing with global problems which can be solved locally, a popular tool of choice is the Alternating Direction Method of Multipliers (ADMM) \cite{Boyd2011} in its consensus formulation. In this framework, a global objective can be decomposed into a set of local and distributed problems which attempt to reach a global agreement. We will show that this scheme can be effectively applied in the convolutional sparse coding context, providing an algorithm with a bi-level consensus interpretation.

The ADMM has been extensively used throughout the literature in convolutional sparse coding. However, as mentioned in the introduction, it has been usually applied in the Fourier domain. As a result, the sense of locality is lost in these approaches and the connection to traditional (local) sparse coding is non-existent. On the contrary, the pursuit method we propose here is carried out in a localized fashion in the original domain, while still benefiting from the advantages of ADMM.

Recall the $\ell_1$ relaxation of the global pursuit, given in Eq. \eqref{eq:lagrangian_BP}.  Note that the noiseless model is contained in this formulation as a particular case when $\lambda$ tends to zero.
Using the separability of the $\ell_1$ norm, $\| \Gama \|_1 = \sum_i \|\alfa_i\|_1$, where $\alfa_i$ are $m-$dimensional local sparse vectors, as previously defined. In addition, using the fact that $\R_i\D\Gama = \O \gama_i$, we apply a local decomposition on the first term as well. This results in 
\begin{equation}
\min_{\{\alfa_i\},\{\gama_i\}} \quad \frac{1}{2n} \sum_i \| \R_i \Y - \O \gama_i \|^2_2 +\lambda \sum_i \|\alfa_i\|_1,
\end{equation}
where we have divided the first sum by the number of contributions per entry in the global signal, which is equal to the patch size $n$.
Note that the above minimization is not equivalent to the original problem in Equation \eqref{eq:lagrangian_BP} since no explicit consensus is enforced between the local variables. Recall that the different $\gama_i$ overlap, and so we must enforce them to agree. In addition, $\alfa_i$ should be constrained to be equal to the center of the corresponding $\gama_i$. Based on these observations, we modify the above problem by adding the appropriate constraints, obtaining
\begin{align}
\min_{\{\alfa_i\},\{\gama_i\},\Gama} \quad \frac{1}{2n} \sum_i \| \R_i \Y  &- \O \gama_i \|^2_2 +\lambda \sum_i \|\alfa_i\|_1 \\
& \text{ s.t. } \begin{cases} \mathbf{Q} \gama_i = \alfa_i \\ \S_i \Gama = \gama_i \end{cases} \forall i,
\end{align}
where $\mathbf{Q}$ extracts the center $m$ coefficients corresponding to $\alfa_i$ from $\gama_i$, and $\S_i$ extracts the $i^{th}$ stripe $\gama_i$ from $\Gama$.

Defining $f_i(\gama_i) = \frac{1}{2n} \| \R_i \Y-  \O \gama_i \|^2_2 $ and $g(\alfa_i) = \lambda \| \alfa_i \|_1 $, the above problem can be minimized by employing the ADMM algorithm, as depicted in Algorithm \ref{alg:ADMM_algo}.
% \begin{equation}
% \min_{\{\alfa_i\},\{\gama_i\},\Gama}  \sum_i f_i(\gama_i) +  g(\alfa_i)\ \text{ s.t. }  \begin{cases} \mathbf{Q}\gama_i = \alfa_i \\ \S_i \Gama = \gama_i \end{cases} \forall i.
% \end{equation}
This is a two-level local-global consensus formulation:  each $m$ dimensional vector $\alfa_i$ is enforced to agree with the center of its corresponding $(2n-1)m$ dimensional $\gama_i$, and in addition, all $\gama_i$ are required to agree with each other as to create a global $\Gama$. 
The above can be shown to be equivalent to the standard two-block ADMM formulation \cite{Boyd2011}. 
% Writing the augmented Lagrangian (in its scaled form), we obtain the following objective for the problem above
% \begin{align*}
% \min_{\Gama,\{\alfa_i\},\{\gama_i\},\{\u_i\},\{\uh_i\}} \sum_i f_i(\gama_i) + g(\alfa_i) & + \frac{\rho}{2} \|  \mathbf{Q}\gama_i - \alfa_i + \u_i  \|^2_2 \\
% & + \frac{\rho}{2} \| \S_i \Gama - \gama_i  + \uh_i  \|^2_2,
% \end{align*}
% which can be minimized with the method depicted in Algorithm \ref{alg:ADMM_algo}.
% We have introduced the (scaled) Lagrange multipliers $\u_i$ and $\uh_i$ corresponding to the variables $\alfa_i$ and $\gama_i$, respectively, and have denoted by $\rho$ the step size in the algorithm. 
Each iteration of this method can be divided into four steps:
\begin{enumerate}
\item Local sparse coding that updates $\alfa_i$ (for all $i$), which amounts to a simple soft thresholding operation.
\item Solution of a linear system of equations for updating $\gama_i$ (for all $i$), which boils down to a simple multiplication by a constant matrix.
\item Update of the global sparse vector $\Gama$, which aggregates the $\gama_i$ by averaging.
\item Update of the dual variables.
\end{enumerate}
As can be seen, the ADMM provides a simple way of breaking the global pursuit into local operations. Moreover, the local coding step is just a projection problem onto the $\ell_1$ ball, which can be solved through simple soft thresholding, implying that there is no complex pursuit involved.

\begin{algorithm}[t]
 \While{not converged}{
 
 \vspace{0.3cm}
 Local Thresholding: $\alfa_i \leftarrow \underset{\alfa}{\min}\ \lambda \| \alfa \|_1 + \frac{\rho}{2} \|  \mathbf{Q}\gama_i - \alfa + \u_i  \|^2_2$ \;
  
  \vspace{0.3cm}
  Stripe Projection:
  \begin{align}
  \hspace{-0.7321cm} \gama_i \leftarrow \M^{-1} \left( \frac{1}{n}\O^T\R_i\Y \right. & + {\rho}( \S_i\Gama + \uh_i) \\
  & + \rho \mathbf{Q}^T ( \alfa_i-\u_i ) \Big),
  \end{align}
  where $\M = \rho \mathbf{Q}^T\mathbf{Q} + \frac{1}{n}\O^T\O + {\rho} \mathbf{I}$\;
  
  \vspace{0.3cm}
  Global Update:\newline
  $\Gama \leftarrow \left( \sum_i \S_i^T \S_i \right)^{-1} \sum_i \S_i^T (\gama_i - \uh_i) $ \;
 
  \vspace{0.3cm}
  Dual Variables Update:\newline
  $\u_i \leftarrow \u_i + (\mathbf{Q} \gama_i - \alfa_i) $ \; 
  $\uh_i \leftarrow \uh_i + (\S_i \Gama - \gama_i) $ \;
 }
 \caption{Locally operating global pursuit via ADMM.}
\label{alg:ADMM_algo}
\end{algorithm}

Since we are in the $\ell_1$ case, the function $g$ is convex, as are the functions $f_i$. Therefore, the above is guaranteed to converge to the minimizer of the global BP problem. As a result, we benefit from the theoretical guarantees derived in previous sections. One could attempt, in addition, to enforce an $\ell_0$ penalty instead of the $\ell_1$ norm on the global sparse vector. Despite the fact that no convergence guarantees could be claimed under such formulation, the derivation of the algorithm remains practically the same, with the only exception that the soft thresholding is replaced by a hard one.

\vspace{-0.3cm}
\subsection{An Iterative Soft Thresholding Approach}
While the above algorithm suggests a way to tackle the global problem in a local fashion, the matrix involved in the stripe projection stage, $\Z^{-1}$, is relatively large when compared to the dimensions of $\D_L$. As a consequence, the bi-level consensus introduces an extra layer of complexity to the algorithm. In what follows, we propose an alternative method based on the Iterative Soft Thresholding (IST) algorithm that relies solely on multiplications by $\D_L$ and features a simple intuitive interpretation and implementation. A similar approach for solving the convolutional sparse coding problem was suggested in \cite{Chalasani2013}. Our main concern here is to provide insights into local alternatives for the global sparse coding problem and their guarantees, whereas the work in \cite{Chalasani2013} focused on the optimizations aspects of this pursuit from an entirely global perspective.

\begin{algorithm}[t]
	$\forall i \quad \r_i^0 = \R_i \Y, \quad \alfa_i^0 = \mathbf{0}$\;
	k = 1\;
	\While{not converged}{
		\vspace{0.3cm}
		Local Coding:\newline $\forall i \quad \alfa_i^k = \mathcal{S}_{\lambda/c}\left( \alfa_i^{k-1} + \frac{1}{c} \ \D_L^T \ \r_i^{k-1} \right)$ \;
		
		\vspace{0.3cm}
		Computation of the Patch Averaging Aggregation:\newline
		$ \widehat{\X}^k = \sum_i \R_i^T\D_L\alfa_i^k $ \; 
				
		\vspace{0.3cm}
		Update of the Residuals:\newline
		$\forall i \quad \r_i^k = \R_i \left( \Y - \widehat{\X}^k \right)$ \; 
		
		\vspace{0.3cm}
		$k=k+1$\;
	}
	\caption{Global pursuit using local processing via iterative soft thresholding.}
	\label{alg:IST_algo}
\end{algorithm}

Let us consider the IST algorithm \cite{Daubechies2004} which minimizes the global objective in Equation \eqref{eq:lagrangian_BP}, by iterating the following updates
\begin{equation}
\Gama^k = \mathcal{S}_{\lambda/c}\left( \Gama^{k-1} + \frac{1}{c} \D^T(\Y-\D\Gama^{k-1}) \right),
\end{equation}
where $\mathcal{S}$ applies an entry-wise soft thresholding operation with threshold $\lambda/c$. Interpreting the above as a projected gradient descent, the coefficient $c$ relates to the gradient step size and should be set according to the maximal singular value of the matrix $\D$ in order to guarantee convergence \cite{Daubechies2004}.

The above algorithm might at first seem undesirable due to the multiplications of the residual $\Y-\D\Gama^{k-1}$ with the global dictionary $\D$. Yet, as we show in the Supplementary Material, such a multiplication does not need to be carried out explicitly due to the convolutional structure imposed on our dictionary. In fact, the above is mathematical equivalent to an algorithm that performs local updates given by 
\begin{equation}
\alfa_i^k = \mathcal{S}_{\lambda/c}\left( \alfa_i^{k-1} + \frac{1}{c} \ \D_L^T \ \r_i^{k-1} \right),
\end{equation}
where $\r_i^k=\R_i (\Y-\D\Gama^{k-1})$ is a patch from the global residual. This scheme is depicted in Algorithm \ref{alg:IST_algo}.

\begin{figure}[!htb]
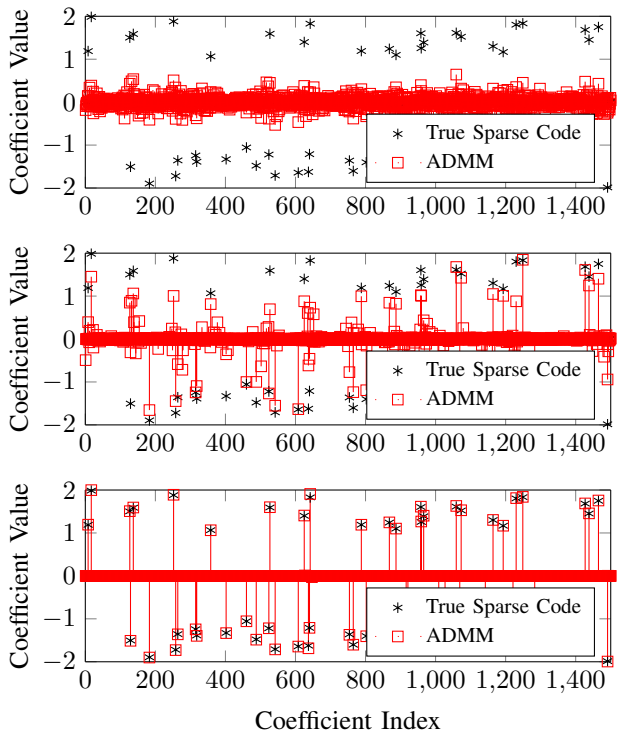

	\centering
	\begin{subfigure}{.45\textwidth}
		%\caption{Iteration number 20.}
		\setlength\figureheight{.28\textwidth}
		\setlength\figurewidth{.9\textwidth}
		\input{graphics/ADMM_L1_noiseless_codes_1.tikz}
	\end{subfigure}
	
	\centering
	\vspace{0.1cm}
	\begin{subfigure}{.45\textwidth}
		%\caption{Iteration number 200.}
		\setlength\figureheight{.28\textwidth}
		\setlength\figurewidth{.9\textwidth}
		\input{graphics/ADMM_L1_noiseless_codes_2.tikz}
	\end{subfigure}
	
	\centering
	\vspace{0.1cm}
	\begin{subfigure}{.45\textwidth}
		%\caption{Iteration number 1000.}
		\setlength\figureheight{.28\textwidth}
		\setlength\figurewidth{.9\textwidth}
		\input{graphics/ADMM_L1_noiseless_codes_3.tikz}
	\end{subfigure}
	\caption{The sparse vector $\Gama$ after the global update stage in the ADMM algorithm at iterations $20$ (top), $200$ (middle) and $1000$ (bottom). An $\ell_1$ norm formulation was used for this experiment, in a noiseless setting.}
	\label{fig:ADMM_L1_noiseless_codes}
	
	\end{figure}

From an optimization point of view, one can interpret each iteration of the above as a scatter and gather process: local residuals are first extracted and scattered to different nodes where they undergo shrinkage operations, and the results are then gathered for the re-computation of the global residual. From an image processing point of view, this algorithm decomposes a signal into overlapping patches, {\sl restores} these separately and then aggregates the result for the next iteration. Notably, this is very reminiscent of the patch averaging scheme, as described in the introduction, and it shows for the first time the relation between patch averaging and the convolutional sparse model. While the former processes every patch once and independently, the above algorithm indicates that one must iterate this process if one is to reach global consensus.

Assuming the step size is chosen appropriately, the above algorithm is also guaranteed to converge to the solution of the global BP. As such, our theoretical analysis holds in this case as well. Alternatively, one could attempt to employ an $\ell_0$ approach, using a global iterative hard thresholding algorithm. In this case, however, there are no theoretical guarantees in terms of the $\Loi$ norm. Still, we believe that a similar analysis to the one taken throughout this work could lead to such claims.

\vspace{-0.3cm}
\subsection{Experiments}
Next, we proceed to provide empirical results for the above described methods. To this end, we take an undercomplete DCT dictionary of size $25\times 5$, and use it as $\D_L$ in order to construct the global convolutional dictionary $\D$ for a signal of length $N = 300$. We then generate a random global sparse vector $\Gama$ with $50$ non-zeros, with entries distributed uniformally in the range $[-2,-1]\ \cup\ [1,2]$, creating the signal $\X = \D\Gama$.

\begin{figure}[t]
	\centering
	\includegraphics[trim = 30 10 10 10, width = 0.48\textwidth]{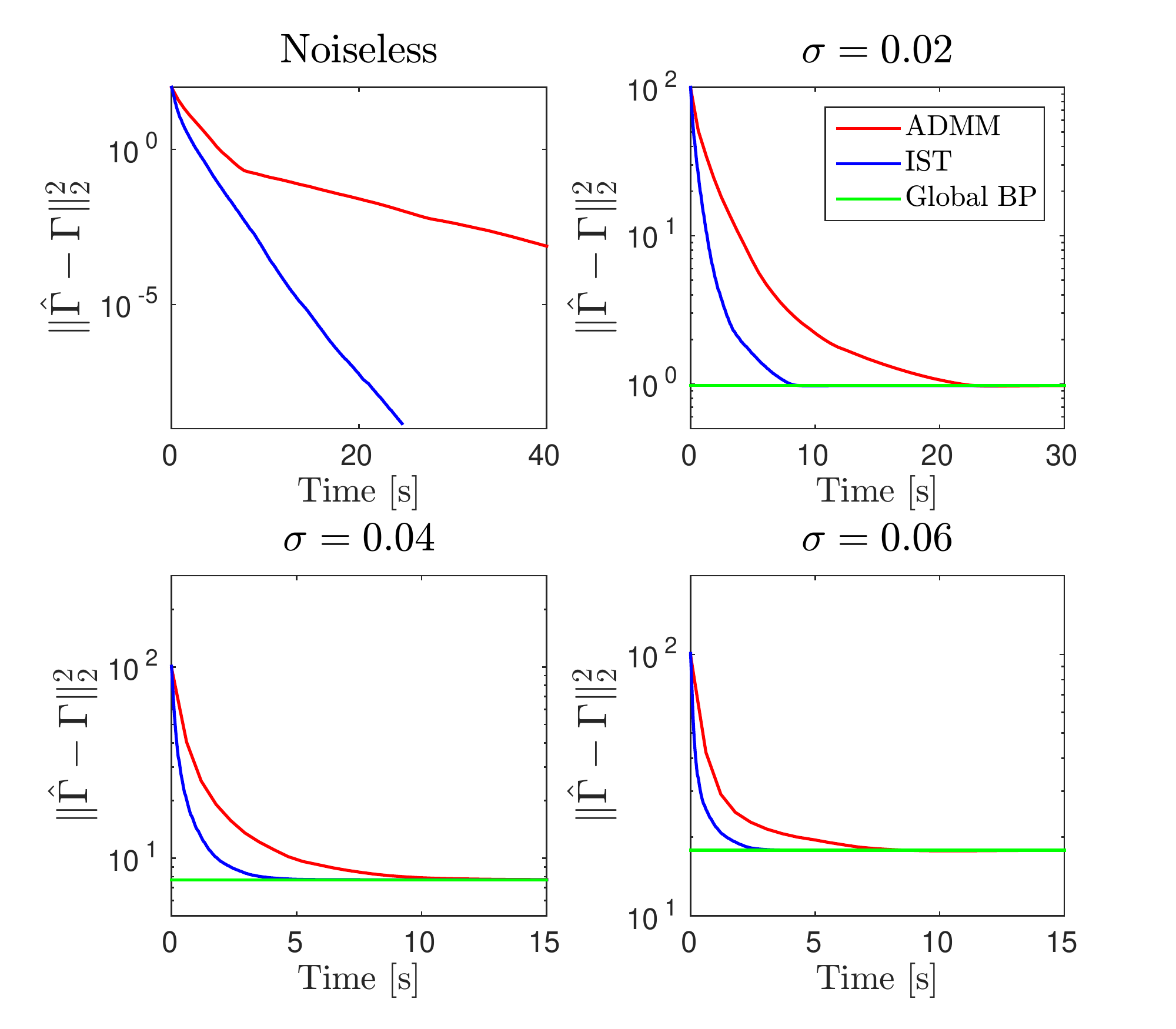}
	\caption{Distance between the estimate $\hat{\Gama}$ and the underlying solution $\Gama$ as a function of time for the IST and the ADMM algorithms compared to the solution obtained by solving the global BP.}
	\label{fig:ADMM_IT}
	\vspace{-0.3cm}
\end{figure}

We first employ the ADMM and IST algorithms in a noiseless scenario in order to minimize the global BP and find the underlying sparse vector. Since there is no noise added in this case, we decrease the penalty parameter $\lambda$ progressively throughout the iterations, making this value tend to zero as suggested in the previous subsection. In Figure \ref{fig:ADMM_L1_noiseless_codes} we present the evolution of the estimated $\hat{\Gama}$ for the ADMM solver throughout the iterations, after the global update stage. Note how the algorithm progressively increases the consensus and eventually recovers the true sparse vector. Equivalent plots are obtained for the IST method, and these are therefore omitted.

To extend the experiment to the noisy case, we contaminate the previous signal with additive white Gaussian noise of different standard deviations: $\sigma=0.02,0.04,0.06$. We then employ both local algorithms to solve the corresponding BPDN problems, and analyze the $\ell_2$ distance between their estimated sparse vector and the true one, as a function of time. These results are depicted in Figure \ref{fig:ADMM_IT}, where we include for completion the distance of the solution achieved by the global BP in the noisy cases. A few observations can be drawn from these results. Note that both algorithms converge to the solution of the global BP in all cases. In particular, the IST converges significantly faster than the ADMM method. Interestingly, despite the later requiring a smaller number of iterations to converge, these are relatively more expensive than those of the IST, which employs only multiplications by the small $\D_L$.

%The $\ell_2$ norm squared between these, $\|\hat{\Gama}-\Gama\|^2_2$, and the feasibility values are shown in Figure \ref{fig:ADMM_L1_noiseless_obj_feas}. Due to the bi-level consensus, we have both local and global feasibility terms, which correspond to the increments of the dual variables $\u_i$ and $\overline{\u}_i$, respectively. Formally, the first is given by $\sqrt{\sum_i \|\mathbf{Q} \gama_i - \alfa_i \|_2^2}$, whereas the later is defined as $\sqrt{\sum_j \|\mathbf{S}_j \Gama - \gama_j \|_2^2}$.
%
%Before concluding this section, we revisit the experiment carried for the ADMM algorithm. We run the IST on the same signal, minimizing the BP objective. We avoid plotting the obtained sparse codes as these look the same as the ones obtained by the ADMM. 
%Regarding the number of iterations, the IST requires considerably more iterations than the ADMM, and thus we run it for $50,000$ iterations. However, these iterations are a significantly cheaper due to the use of the small $\D_L$. To demonstrate this, we plot in Figure \ref{fig:ADMMvsIT} the $\|\hat{\Gama}-\Gama\|_2^2$ as a function of time for both algorithms. These implementations use un-optimized code, and improvement can probably be expected for both cases.

\vspace{-0.25cm}
\section{Conclusion and Future Work}
\label{Sec:Conclusions}

In this work we have presented a formal analysis of the convolutional sparse representation model. In doing so, we have reformulated the objective of the global pursuit, introducing the $\Loi$ norm and the corresponding $\Poi$ problem, and proven the uniqueness of its solution. By migrating from the $P_0$ to the $\Poi$ problem, we were able to provide meaningful guarantees for the success of popular algorithms in the noiseless case, improving on traditional bounds that were shown to be very pessimistic under the convolutional case. In order to achieve such results, we have generalized a series of concepts such as Spark and the mutual coherence to their counterparts in the convolutional setting. 

Striding on the foundations paved in the first part of this work, we moved on to present a series of stability results for the convolutional sparse model in the presence of noise, providing guarantees for corresponding pursuit algorithms. These were possible due to our migration from the $\ell_0$ to the $\Loi$ norm, together with the generalization and utilization of concepts such as RIP and ERC. Seeking for a connection between traditional patch-based processing and the convolutional sparse model, we finally proposed two efficient methods that solve the global pursuit while working locally. 

We envision many possible directions of future work, and here we outline some of them:
\begin{itemize}
	\item We could extend our study, which considers only worst-case scenarios, to an average-performance analysis. By assuming more information about the model, it might be possible to quantify the probability of success of pursuit methods in the convolutional case. Such results would close the gap between current bounds and empirical results.
	
	\item From an application point of view, we envision that interesting algorithms could be proposed to tackle real problems in signal and image processing while using the convolutional model. We note that while convolutional sparse coding has been applied to various problems, simple inverse problems such as denoising have not yet been properly addressed. We believe that the analysis presented in this work could facilitate the development of such algorithms by showing how to leverage on the subtleties of this model.	

	\item Interestingly, even though we have declared the $\Poi$ problem as our goal, at no point have we actually attempted to tackle it directly. What we have shown instead is that popular algorithms succeed in finding its solution. One could perhaps propose an algorithm specifically tailored for solving this problem -- or its convex relaxation ($\ell_{1,\infty}$). Such a method might be beneficial from both a theoretical and a practical aspect.
\end{itemize}
All these points, and more, are matter of current research.

\vspace{-0.25cm}
\section{Acknowledgements}
The research leading to these results has received funding from the European Research Council under European Union’s Seventh Framework Programme, ERC Grant agreement no. 320649. The authors would like to thank Dmitry Batenkov, Yaniv Romano and Raja Giryes for the prolific conversations and most useful advice which helped shape this work.

\bibliographystyle{ieeetr}
\bibliography{MyBib}

\appendix

\section{On the $\Loi$ Norm} 
\label{sect:TraingIneqLoi}
\begin{thm}
The triangle inequality holds for the $\Loi$ norm.
\end{thm}

\begin{proof}
Let $\Gama^1$ and $\Gama^2$ be two global sparse vectors. Denote the $i^{th}$ stripe extracted from each as $\gama^1_i$ and $\gama^2_i$, respectively. Notice that
\begin{align*}
\| \Gama^1+\Gama^2 \|_{0,\infty} =\max_i\|\gama^1_i+\gama^2_i\|_0 & \leq\max_i \left(\|\gama^1_i\|_0+\|\gama^2_i\|_0\right)\\
\leq \max_i\|\gama^1_i\|_0 +\max_i\|\gama^2_i\|_0 & =\|\Gama^1\|_{0,\infty}+\|\Gama^2\|_{0,\infty}.
\end{align*}
In the first inequality we have used the triangle inequality of the $\ell_0$ norm.
\end{proof}

%% ---------------------------------------------------------------------------------------------------------------

\section{Theoretical Analysis of Ideal Signals}

\begin{customthm}{5}{(Uncertainty and uniqueness using Stripe-Spark):}
	Let $\D$ be a convolutional dictionary with Stripe-Spark $\sigma_\infty$. If a solution $\Gama$ obeys \mbox{$\|\Gama\|_{0,\infty}<\frac{1}{2}\sigma_\infty$}, then this is necessarily the global optimum for the $\Poi$ problem for the signal $\D\Gama$.
\end{customthm}

\begin{proof}
	Let $\hat{\Gama}\neq\Gama$ be an alternative solution. Then $\D\left(\Gama-\hat{\Gama}\right)=0$. By definition of the Stripe-Spark
	\begin{equation}
	\|\Gama-\hat{\Gama}\|_{0,\infty}\geq\sigma_\infty.
	\end{equation}
	Using the triangle inequality of the $\Loi$ norm,
	\begin{equation}
	\|\Gama\|_{0,\infty}+\|\hat{\Gama}\|_{0,\infty}\geq\|\Gama-\hat{\Gama}\|_{0,\infty}\geq\sigma_\infty.
	\end{equation}
	This result poses an uncertainty principle for $\Loi$ sparse solutions of the system $\X = \D\Gama$, suggesting that if a very sparse solution is found, all alternative solutions must be much denser. Since $\|\Gama\|_{0,\infty}<\frac{1}{2}\sigma_\infty$, we must have that $\|\hat{\Gama}\|_{0,\infty}>\frac{1}{2}\sigma_\infty$, or in other words, every solution other than $\Gama$ has higher $\Loi$ norm, thus making $\Gama$ the global solution for the $\Poi$ problem.
\end{proof}

\begin{customthm}{6}{(Lower bounding the Stripe-Spark via the mutual coherence):}
	For a convolutional dictionary $\D$ with mutual coherence $\mu(\D)$, the Stripe-Spark can be lower-bounded by
	\begin{equation}
	\sigma_\infty(\D)\geq 1+\frac{1}{\mu(\D)}.
	\end{equation}
\end{customthm}

\begin{proof}
	Let $\Delt$ be a vector such that $\Delt\neq\mathbf{0}$ and $\D\Delt=\mathbf{0}$.
	Note that we can write
	\begin{equation} \label{eq:relation}
	\D_\mathcal{T} \Delt_\mathcal{T} = \mathbf{0},
	\end{equation}
	where $\Delt_\mathcal{T}$ is the vector $\Delt$ restricted to its support $\mathcal{T}$, and $\D_\mathcal{T}$ is the dictionary composed of the corresponding atoms.
	Consider now the Gram matrix, $\G^{\mathcal{T}} = \D_\mathcal{T}^T \D_\mathcal{T}$, which corresponds to a portion extracted from the global Gram matrix $\D^T\D$. The relation in Equation \eqref{eq:relation} suggests that $\D_\mathcal{T}$ has a nullspace, which implies that its Gram matrix must have at least one eigenvalue equal to zero. Using Lemma 1, the lower bound on the eigenvalues of $\G^{\mathcal{T}}$ is given by $1-(k-1)\mu(\D)$, where $k$ is the $\Loi$ norm of $\Delt$. Therefore, we must have that $1-(k-1)\mu(\D) \leq 0$, or equally $k \geq 1 + \frac{1}{\mu(\D)}$. We conclude that a vector $\Delt$, which is in the null-space of $\D$, must always have an $\Loi$ norm of at least $1+\frac{1}{\mu(\D)}$, and so the Stripe-Spark $\sigma_\infty$ is also bounded by this number.
\end{proof}

\begin{customthm}{8}{(Global OMP recovery guarantee using $\Loi$ norm):}
	Given the system of linear equations $\X = \D\Gama$, if a solution $\Gama$ exists satisfying
	\begin{equation}
	\|\Gama\|_{0,\infty} < \ \frac{1}{2}\left(1+\frac{1}{\mu(\D)}\right),
	\end{equation}
	then OMP is guaranteed to recover it.
\end{customthm}

\begin{proof}
	Denoting by $\mathcal{T}$ the support of the solution $\Gama$, we can write
	\begin{equation}
	\X = \D\Gama = \sum_{t\in \mathcal{T}} \Gamma_t \d_t.
	\label{GlobalExpression3}
	\end{equation}
	Suppose, without loss of generality, that the sparsest solution has its largest coefficient (in absolute value) in $\Gamma_i$. 
	For the first step of the OMP to choose one of the atoms in the support, we require
	\begin{equation}
	|\d_i^T \X | > \max_{j\notin\mathcal{T}} | \d_j^T \X |.
	\end{equation}
	Substituting Equation \eqref{GlobalExpression3} in this requirement we obtain
	\begin{equation} \label{eq:inequality3}
	\left| \sum_{t\in \mathcal{T}}\Gamma_t\d_t^T\d_i \right| > \max_{j\notin\mathcal{T}} \left| \sum_{t\in \mathcal{T}}\Gamma_t\d_t^T\d_j \right|.
	\end{equation}
	Using the reverse triangle inequality, the assumption that the atoms are normalized, and that $|\Gamma_i|\geq|\Gamma_t|$, we construct a lower bound for the left hand side:
	\begin{align}
	\left| \sum_{t\in \mathcal{T}}\Gamma_t\d_t^T\d_i \right|
	\geq& |\Gamma_i| - \sum_{t \in \mathcal{T},t\neq i}|\Gamma_t|\cdot|\d_t^T\d_i | \\
	\geq& |\Gamma_i| - |\Gamma_i|\sum_{t \in \mathcal{T},t\neq i}|\d_t^T\d_i |.
	\end{align}
	Consider the stripe which completely contains the $i^{th}$ atom as shown in Figure \ref{Fig:details}. Notice that $\d_t^T\d_i$ is zero for every atom too far from $\d_i$ because the atoms do not overlap. Denoting the stripe which fully contains the $i^{th}$ atom as $p(i)$ and its support as $\mathcal{T}_{p(i)}$, we can restrict the summation as:
	\begin{equation} \label{Eq:EquationFromOMPproof}
	\left| \sum_{t\in \mathcal{T}}\Gamma_t\d_t^T\d_i \right| \geq |\Gamma_i| - |\Gamma_i|\sum_{t \in \mathcal{T}_{p(i)},t\neq i}|\d_t^T\d_i |.
	\end{equation}
	We can bound the right side by using the number of non-zeros in the support $\mathcal{T}_{p(i)}$, denoted by $n_{p(i)}$, together with the definition of the mutual coherence, obtaining:
	\begin{figure} 
		\includegraphics[trim = -80 80 80 0, width=.4\textwidth]{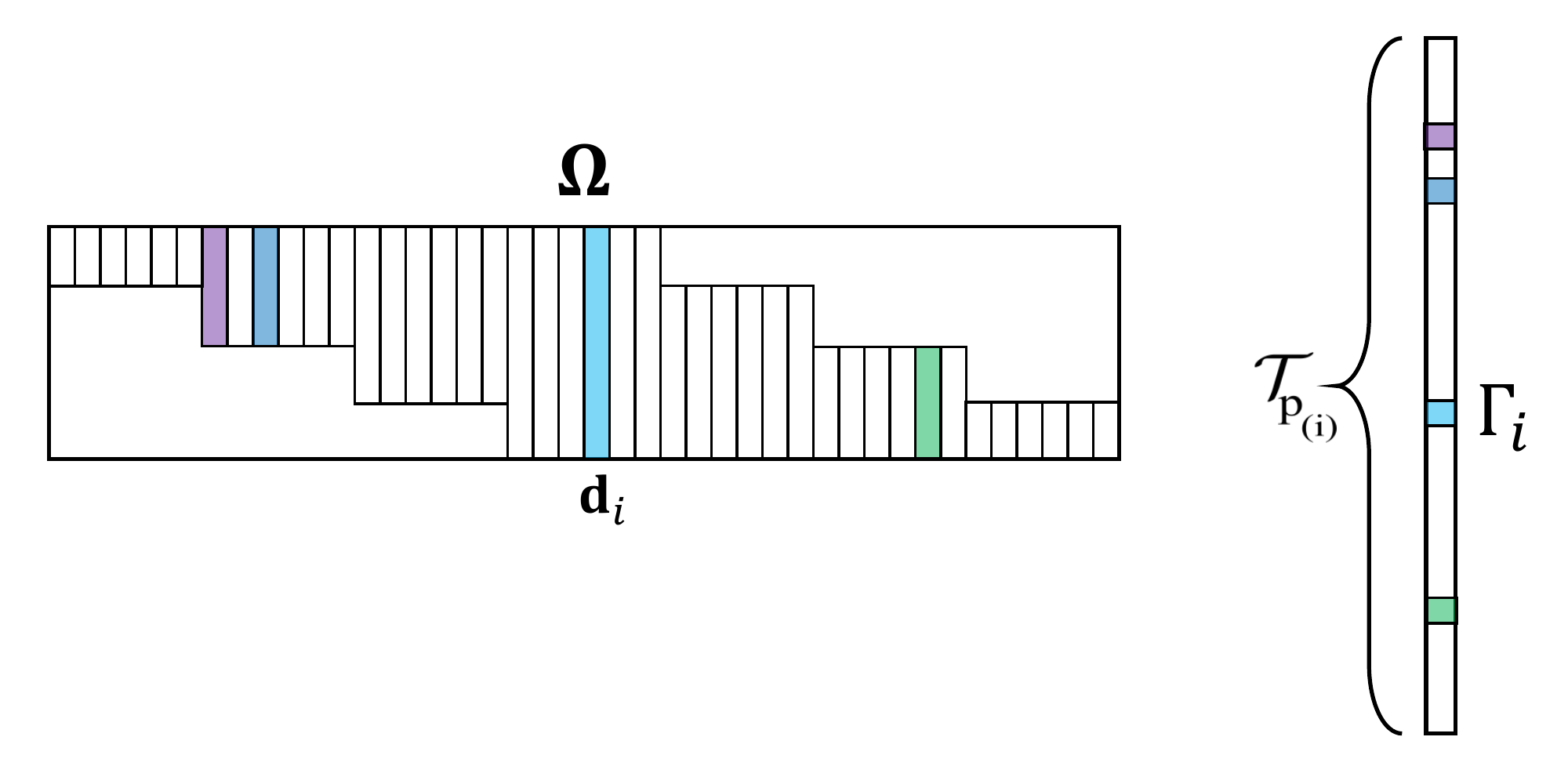}
		\caption{The $p_{(i)}$ stripe of atom $\d_i$.}
		\label{Fig:details}
		\vspace{-0.3cm}
	\end{figure}
	\begin{equation}
	\left| \sum_{t\in \mathcal{T}}\Gamma_t\d_t^T\d_i \right| \geq |\Gamma_i| -  |\Gamma_i| \cdot (n_{p(i)}-1)\cdot\mu(\D).
	\end{equation}
	Using the definition of the $\Loi$ norm, we obtain
	\begin{equation}
	\left| \sum_{t\in \mathcal{T}}\Gamma_t\d_t^T\d_i \right| \geq |\Gamma_i| -  |\Gamma_i| \cdot (\|\Gama\|_{0,\infty}-1)\cdot\mu(\D).
	\end{equation}
	Now, we construct an upper bound for the right hand side of Equation \eqref{eq:inequality3}, using the triangle inequality and the fact that $|\Gamma_i|$ is the maximal value in the sparse vector:
	\begin{align} \label{eq:SkipSteps2}
	\max_{j\notin\mathcal{T}}\left| \sum_{t\in \mathcal{T}}\Gamma_t\d_t^T\d_j \right| &\leq \max_{j\notin\mathcal{T}}\sum_{t \in \mathcal{T}} |\Gamma_t|\cdot|\d_t^T\d_j |\\
	&\leq |\Gamma_i|\max_{j\notin\mathcal{T}}\sum_{t \in \mathcal{T}}|\d_t^T\d_j |.
	%&= |\Gamma_i|\max_{j\notin\mathcal{T}}\sum_{t \in \mathcal{T}_{p(j)}}|\d_t^T\d_j |.
	\end{align}
	Relying on the same rational as above, we obtain:
	\begin{align}
	&& \max_{j\notin\mathcal{T}}\left| \sum_{t\in \mathcal{T}}\Gamma_t\d_t^T\d_j \right| & \leq |\Gamma_i|\max_{j\notin\mathcal{T}}\sum_{t \in \mathcal{T}_{p(j)}}|\d_t^T\d_j |\\
	&&\leq |\Gamma_i|\max_{j\notin\mathcal{T}}\ n_{p(j)}\cdot\mu(\D) & \leq |\Gamma_i|\cdot\|\Gama\|_{0,\infty}\cdot\mu(\D).
	\end{align}
	Using both bounds, we get
	\begin{align*}
	\left| \sum_{t\in \mathcal{T}}\Gamma_t\d_t^T\d_i \right|
	&\geq |\Gamma_i| -  |\Gamma_i| \cdot (\|\Gama\|_{0,\infty}-1)\cdot\mu(\D)\\
	&>|\Gamma_i|\cdot\|\Gama\|_{0,\infty}\mu(\D)
	\geq\max_{j\notin\mathcal{T}}\left| \sum_{t\in \mathcal{T}}\Gamma_t\d_t^T\d_j \right|.
	\end{align*}
	Thus,
	\begin{equation}
	1-(\|\Gama\|_{0,\infty}-1)\cdot\mu(\D) > \|\Gama\|_{0,\infty}\cdot\mu(\D).
	\end{equation}
	From this we obtain the requirement stated in the theorem. Thus, this condition guarantees the success of the first OMP step, implying it will choose an atom inside the true support.
	
	The next step in the OMP algorithm is an update of the residual. This is done by decreasing a term proportional to the chosen atom (or atoms within the correct support in subsequent iterations) from the signal. Thus, this residual is also a linear combination of the same atoms as the original signal. As a result, the $\Loi$ norm of the residual's representation is less or equal than the one of the true sparse code $\Gama$. Using the same set of steps we obtain that the condition on the $\Loi$ norm \eqref{eq:Loi_condition} guarantees that the algorithm chooses again an atom from the true support of the solution. Furthermore, the orthogonality enforced by the least-squares step guarantees that the same atom is never chosen twice. As a result, after $\|\Gama\|_0$ iterations the OMP will find all the atoms in the correct support, reaching a residual equal to zero.
\end{proof}

\begin{customthm}{9}{(Global Basis Pursuit recovery guarantee using the $\Loi$ norm):}
	For the system of linear equations $\D\Gama=\X$, if a solution $\Gama$ exists obeying
	\begin{equation}
		\|\Gama\|_{0,\infty}<\frac{1}{2}\left(1+\frac{1}{\mu(\D)}\right),
	\end{equation}
	then Basis Pursuit is guaranteed to recover it.
\end{customthm}

\begin{proof}
	Define the following set
	\begin{equation}
	\C=\set*{ \hat{\Gama} \given
		\begin{gathered}
		\begin{split}
		&&\hat{\Gama}&\neq\Gama,& \quad \D(\hat{\Gama}-\Gama)&=\mathbf{0}\\
		&&\|\hat{\Gama}\|_1&\leq\|\Gama\|_1,& \quad \|\hat{\Gama}\|_{0,\infty}&>\|\Gama\|_{0,\infty}
		\end{split}
		\end{gathered}
	}.
	\end{equation}
	This set contains all alternative solutions which have lower or equal $\ell_1$ norm and higher $\|\cdot\|_{0,\infty}$ norm. If this set is non-empty, the solution of the basis pursuit is different from $\Gama$, implying failure. In view of our uniqueness result, and the condition posed in this theorem on the $\Loi$ cardinality of $\Gama$, every solution $\hat{\Gama}$ which is not equal to $\Gama$ must have a higher $\|\cdot\|_{0,\infty}$ norm. Thus, we can omit the requirement $\|\hat{\Gama}\|_{0,\infty}>\|\Gama\|_{0,\infty}$ from $\C$.
	
	By defining $\Delt=\hat{\Gama}-\Gama$, we obtain a shifted version of the set,
	\begin{equation}
	\C_s=\set*{ \Delt \given
		\begin{gathered}
		\Delt\neq\mathbf{0}, \quad \D\Delt=\mathbf{0}\\
		\mathbf{0}\geq\|\Delt+\Gama\|_1-\|\Gama\|_1
		\end{gathered}
	}.
	\end{equation}
	In what follows, we will enlarge the set $\C_s$ and prove that it remains empty even after this expansion.
	Since $\D\Delt=\mathbf{0}$, then $	\D^T\D\Delt=\mathbf{0}$. By subtracting $\Delt$ from both sides, we obtain
	\begin{equation}
	-\Delt=(\D^T\D-\mathbf{I})\Delt.
	\label{Eq:equality_constraint_L1proof}
	\end{equation}
	Taking an entry-wise absolute value on both sides, we obtain
	\begin{equation}
	|\Delt|=|(\D^T\D-\mathbf{I})\Delt|\leq|\D^T\D-\mathbf{I}|\cdot|\Delt|,
	\end{equation}
	where we have applied the triangle inequality to the multiplication of the $i^{th}$ row of $(\D^T\D-\mathbf{I})$ by the vector $\Delt$. Note that in the convolutional case $\D^T\D$ is zero for inner products of atoms which do not overlap. Furthermore, the $i^{th}$ row of $\D^T\D$ is non-zero only in the indices which correspond to the stripe that fully contains the $i^{th}$ atom, and these non-zero entries can be bounded by $\mu(\D)$. Thus, extracting the $i^{th}$ row from the above equation gives
	\begin{equation*}
	|\Delta_i|\leq\mu(\D)\left(\|\delt_{p(i)}\|_1-|\Delta_i|\right), \label{Eq:Delta_i_abs}
	\end{equation*}
	where $p(i)$ is the stripe centered around the $i^{th}$ atom and $\delt_{p(i)}$ is the corresponding sparse vector of length $(2n-1)m$ extracted from $\Delt$, as can be seen in Figure \ref{proofL1}.
	\begin{figure}[t]
		\centering
		\includegraphics[trim =10 0 10 0 ,width=.33\textwidth]{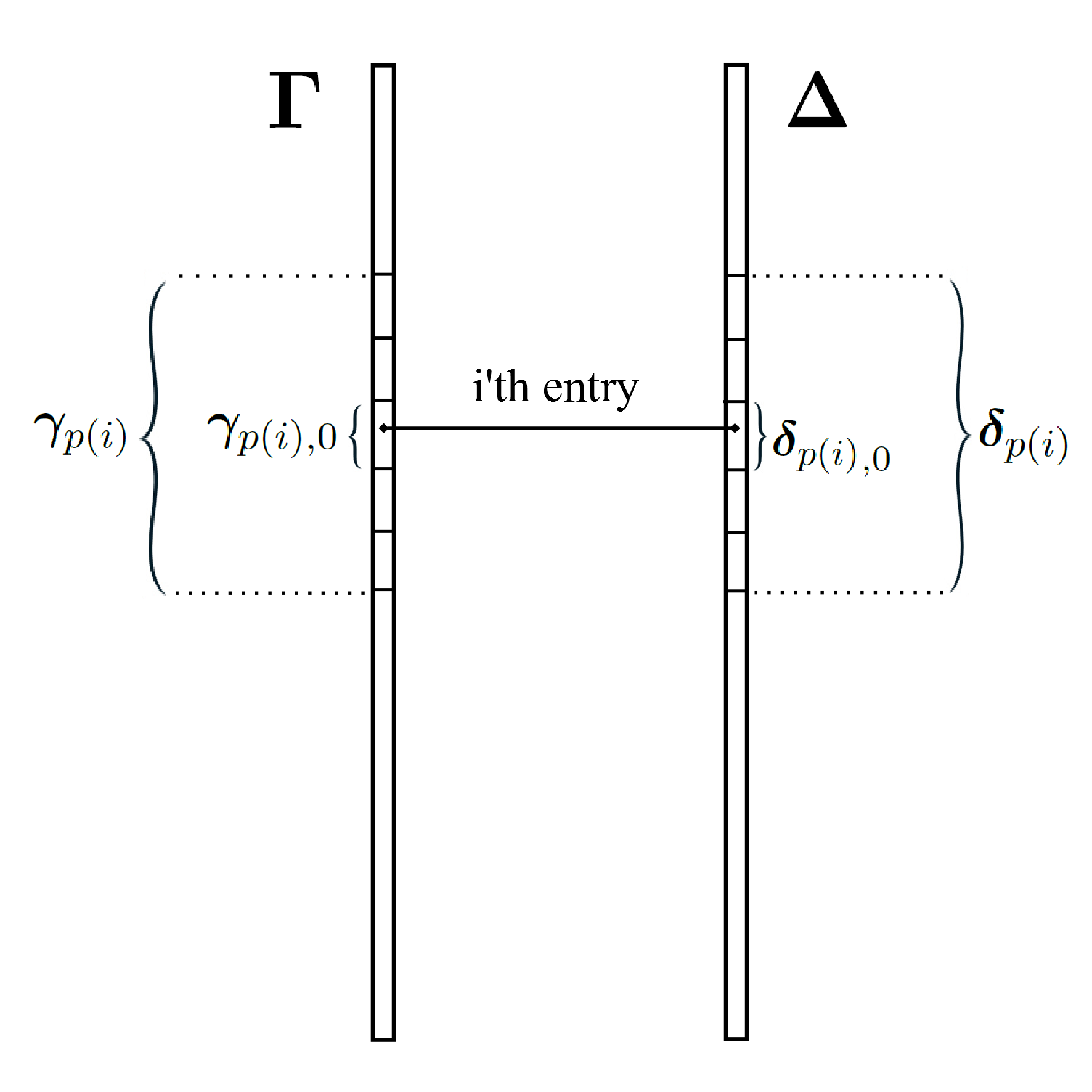}
		\caption{On the left we have the global sparse vector $\Gama$, a stripe  $\gama_{p(i)}$ (centered around the $i^{th}$ atom) extracted from it, and the center of this stripe $\gama_{p(i),0}$. The length of the stripe $\gama_{p(i)}$ is $(2n-1)m$ and the length of $\gama_{p(i),0}$ is $m$. On the right we have the corresponding global vector $\Delt$. Notice that if we were to consider the $i+1$ entry instead of the $i^{th}$, the vector corresponding to $\delt_{p(i)}$ would not change because the atoms $i$ and $i+1$ are fully overlapping.}
		\label{proofL1}
		\vspace{-0.2cm}
	\end{figure}
	This can be written as
	\begin{equation}
	|\Delta_i|\leq\frac{\mu(\D)}{\mu(\D)+1}\|\delt_{p(i)}\|_1.
	\end{equation}
	The above expression is a relaxation of the equality in Equation \eqref{Eq:equality_constraint_L1proof}, since each entry $\Delta_i$ is no longer constrained to a specific value, but rather bounded from below and above. Therefore, by putting the above into $\C_s$, we obtain a larger set $\C_s^1$:
	\begin{equation} 
	\C_s\subseteq\C_s^1=\set*{ \Delt \given
		\begin{gathered}
		\Delt\neq\mathbf{0}, \quad 	\mathbf{0}\geq\|\Delt+\Gama\|_1-\|\Gama\|_1 \\
		|\Delta_i|\leq\frac{\mu(\D)}{\mu(\D)+1}\|\delt_{p(i)}\|_1, \quad\forall i
		\end{gathered}
	}.
	\end{equation}
	Next, let us examine the second requirement
	\begin{align}
	\mathbf{0}\geq & \|\Delt+\Gama\|_1-\|\Gama\|_1 \\ \label{Eq:equality_constraint_L1proof_2}
	= & \sum_{i\in \mathcal{T}(\Gama)} \left(|\Delta_i+\Gamma_i|-|\Gamma_i|\right) + \sum_{i\notin \mathcal{T}(\Gama)} |\Delta_i|,
	\end{align}
	where, as before, $\mathcal{T}(\Gama)$ denotes the support of $\Gama$.
	Using the reverse triangle inequality, $|a+b|-|b|\geq-|a|$, we obtain
	\begin{align}\label{eq:inequality_bp_proof}
	\mathbf{0}&\geq\sum_{i\in \mathcal{T}(\Gama)} \left(|\Delta_i+\Gamma_i|-|\Gamma_i|\right) + \sum_{i\notin \mathcal{T}(\Gama)} |\Delta_i|\\ \nonumber
	&\geq\sum_{i\in \mathcal{T}(\Gama)} -|\Delta_i| + \sum_{i\notin \mathcal{T}(\Gama)} |\Delta_i| =\|\Delt\|_1-2\mathds{1}^T_{\mathcal{T}(\Gama)}|\Delt|, 
	\end{align}
	where the vector $\mathds{1}_{\mathcal{T}(\Gama)}$ contains ones in the entries corresponding to the support of $\Gama$ and zeros elsewhere. 
	Note that every vector satisfying Equation \eqref{Eq:equality_constraint_L1proof_2} will necessarily satisfy Equation \eqref{eq:inequality_bp_proof}. Therefore, by relaxing this constraint in $\C_s^1$, we obtain a larger set $\C_s^2$
	\begin{equation} \label{eq:CS2}
	\C_s^1\subseteq\C_s^2=\set*{ \Delt \given
		\begin{gathered}
		\Delt\neq\mathbf{0}, \quad 	\mathbf{0}\geq\|\Delt\|_1-2\mathds{1}^T_{\mathcal{T}(\Gama)}|\Delt|\\
		|\Delta_i|\leq\frac{\mu(\D)}{\mu(\D)+1}\|\delt_{p(i)}\|_1, \quad\forall i
		\end{gathered}
	}.
	\end{equation}
	Next, we will show the above defined set is empty for a small-enough support. We begin by summing the inequalities \mbox{$|\Delta_i|\leq\frac{\mu(\D)}{\mu(\D)+1}\|\delt_{p(i)}\|_1$} over the support of $\gama_{{p(i)},0}$. Recall that $\gama_{p(i)}$ is defined to be a stripe of length $(2n-1)m$ extracted from the global representation vector and $\gama_{p(i),0}$ corresponds to the central $m$ coefficients in the $p(i)$ stripe. Also, note that $\delt_{p(i)}$ is equal for all the entries inside the support of $\gama_{p(i),0}$. Since all the atoms inside the support of $\gama_{p(i),0}$ are fully overlapping, $\delt_{p(i)}$ does not change, as explained in Figure \ref{proofL1}.
	Thus, we obtain
	\begin{equation*}
	\mathds{1}^T_{\mathcal{T}(\gama_{p(i),0})}|\Delt|\leq\frac{\mu(\D)}{\mu(\D)+1}\cdot\|\gama_{p(i),0}\|_0\cdot\|\delt_{p(i)}\|_1.
	\end{equation*}
	Summing over all different $p(i)$ we obtain
	\begin{equation} \label{eq:goingback}
	\mathds{1}^T_{\mathcal{T}(\Gama)}|\Delt|\leq\frac{\mu(\D)}{\mu(\D)+1}\sum_k\|\gama_{k,0}\|_0\cdot\|\delt_k\|_1.
	\end{equation}
	Notice that in the sum above we multiply the $\ell_0$-norm of the \emph{local sparse vector} $\gama_{k,0}$ by the $\ell_1$ norm of the \emph{stripe} $\delt_k$. In what follows, we will show that, instead, we could multiply the $\ell_0$-norm of the \emph{stripe} $\gama_k$ by the $\ell_1$ norm of the \emph{local sparse vector} $\delt_{k,0}$, thus changing the order between the two. As a result, we will obtain the following inequality:
	\begin{equation}
	\mathds{1}^T_{\mathcal{T}(\Gama)}|\Delt|\leq\frac{\mu(\D)}{\mu(\D)+1}\sum_k\|\gama_{k}\|_0\cdot\|\delt_{k,0}\|_1.
	\end{equation}	
	Returning to Equation \eqref{eq:goingback}, we begin by decomposing the $\ell_1$ norm of the stripe $\delt_k$ into all possible shifts ($m-$dimensional chunks) and pushing the sum outside, obtaining:
	\begin{align}
	\mathds{1}^T_{\mathcal{T}(\Gama)}|\Delt|&\leq\frac{\mu(\D)}{\mu(\D)+1}\sum_k\|\gama_{k,0}\|_0\cdot\|\delt_k\|_1\\
	&=\frac{\mu(\D)}{\mu(\D)+1}\sum_k\left(\|\gama_{k,0}\|_0\sum_{j=k-n+1}^{k+n-1}\|\delt_{j,0}\|_1\right)\\
	&=\frac{\mu(\D)}{\mu(\D)+1}\sum_k\sum_{j=k-n+1}^{k+n-1}\|\gama_{k,0}\|_0\|\delt_{j,0}\|_1. \label{eq:matrix_sum}
	\end{align}
	Define a banded matrix $\A$ (with a band of width $2n-1$) such that $\A_{k,j}=\|\gama_{k,0}\|_0\cdot\|\delt_{j,0}\|_1$,	where $k-n+1\leq j\leq k+n-1$. Notice that the summation in \eqref{eq:matrix_sum} is equal to the sum of all entries in this matrix, where the first sum considers all its rows $k$ while the second sum considers all its columns $j$ (the second sum is restricted to the non-zero band). Instead, this interpretation suggests that we could first sum over all the columns $j$, and only then sum over all the rows $k$ which are inside the band. As a result, we obtain that
	\begin{align}
	\mathds{1}^T_{\mathcal{T}(\Gama)}|\Delt|&\leq\frac{\mu(\D)}{\mu(\D)+1}\sum_k\sum_{j=k-n+1}^{k+n-1}\|\gama_{k,0}\|_0\cdot\|\delt_{j,0}\|_1\\
	&=\frac{\mu(\D)}{\mu(\D)+1}\sum_j\sum_{k=j-n+1}^{j+n-1}\|\gama_{k,0}\|_0\cdot\|\delt_{j,0}\|_1\\
	&=\frac{\mu(\D)}{\mu(\D)+1}\sum_j\left(\|\delt_{j,0}\|_1\sum_{k=j-n+1}^{j+n-1}\|\gama_{k,0}\|_0\right).
	\end{align}
	Summing over all possible shifts we obtain the $\ell_0$-norm of the stripe $\gama_j$; i.e.,
	\begin{equation}
	\mathds{1}^T_{\mathcal{T}(\Gama)}|\Delt| \leq \frac{\mu(\D)}{\mu(\D)+1}\sum_j\|\delt_{j,0}\|_1\cdot\|\gama_{j}\|_0.
	\end{equation}
	Using the definition of $\|\cdot\|_{0,\infty}$
	\begin{align}
	\mathds{1}^T_{\mathcal{T}(\Gama)}|\Delt|&\leq\frac{\mu(\D)}{\mu(\D)+1}\sum_j\|\delt_{j,0}\|_1\cdot\|\gama_{j}\|_0\\
	&\leq\frac{\mu(\D)}{\mu(\D)+1}\sum_j\|\delt_{j,0}\|_1\cdot\|\Gama\|_{0,\infty}\\
	&\leq\frac{\mu(\D)}{\mu(\D)+1}\cdot\|\Delt\|_1\cdot\|\Gama\|_{0,\infty}. \label{eq:second_inequality}
	\end{align}
	For the set $\C_s^2$ to be non-empty, there must exist a $\Delt$ which satisfies
	\begin{align}
	\mathbf{0}\geq & \|\Delt\|_1-2\mathds{1}^T_{\mathcal{T}(\Gama)}|\Delt| \\ \geq & \|\Delt\|_1-2\frac{\mu(\D)}{\mu(\D)+1}\cdot\|\Delt\|_1\cdot\|\Gama\|_{0,\infty},
	\end{align}
	where the first and second inequalities are given in \eqref{eq:inequality_bp_proof} and \eqref{eq:second_inequality}, respectively.
	Rearranging the above we obtain $\|\Gama\|_{0,\infty}\geq\frac{1}{2}\left(1+\frac{1}{\mu(\D)}\right)$. However, we have assumed that $\|\Gama\|_{0,\infty}<\frac{1}{2}\left(1+\frac{1}{\mu(\D)}\right)$ and thus the previous inequality is not satisfied. As a result, the set we have defined is empty, implying that BP leads to the desired solution.
\end{proof}

%% ---------------------------------------------------------------------------------------------------------------

\section{On the Shifted Mutual Coherence and Stripe Coherence} 
\label{App:MoreOnShiftedMu}

\begin{customdef}{10}
Define the shifted mutual coherence $\mu_s$ by
\begin{equation}
\mu_s = \underset{i,j}{\max} \quad |\langle \d^0_i,\d^s_j  \rangle|,
\end{equation}
where $\d^0_i$ is a column extracted from $\O_0$, $\d^s_j$ is extracted from $\O_s$, and we require\footnote{The condition $i\neq j$ if $s=0$ is necessary so as to avoid the inner product of an atom by itself.} that $i\neq j$ if $s=0$.
\end{customdef}

\noindent
The shifted mutual coherence exhibits some interesting properties: 
\begin{enumerate}[a)]
\item $\mu_s$ is symmetric with respect to the shift $s$, i.e. $\mu_s=\mu_{-s}$.
\item Its maximum over all shifts equals the global mutual coherence of the convolutional dictionary: $\mu(\D) = \underset{s}{\max}\ \mu_s$.
\item The mutual coherence of the local dictionary is bounded by that of the global one: $\mu(\D_L) = \mu_0 \leq\underset{s}{\max}\ \mu_s=\mu(\D)$.
\end{enumerate}
We now briefly remind the definition of the maximal stripe coherence, as we will make use of it throughout the rest of this section.
Given a vector $\Gama$, recall that the stripe coherence is defined as $\zeta(\gama_i) = \sum_{s=-n+1}^{n-1} n_{i,s}\ \mu_s$, where $n_{i,s}$ is the number of non-zeros in the $s^{th}$ shift of $\gama_i$, taken from $\Gama$. The reader might ponder how the maximal stripe coherence might be computed.
Let us now define the vector $\mathbf{v}$ which contains in its $i^{th}$ entry the number $n_{i,0}$. Using this definition, the coherence of every stripe can be calculated efficiently by convolving the vector $\mathbf{v}$ with the vector of the shifted mutual coherences $[\mu_{-n+1},\dots,\mu_{-1},\mu_0,\mu_1,\dots,\mu_{n-1}]$.

Next, we provide an experiment in order to illustrate the shifted mutual coherence. To this end, we generate a random local dictionary with $m=8$ atoms of length $n=64$ and afterwards normalize its columns. We then construct a convolutional dictionary which contains global atoms of length $N = 640$. We exhibit the shifted mutual coherences for this dictionary in Figure \ref{fig:Mus}. 

\begin{figure}[t]
\begin{center}
	\begin{subfigure}[t]{.24\textwidth}
		\includegraphics[trim = 0 0 50 20, width=\textwidth]{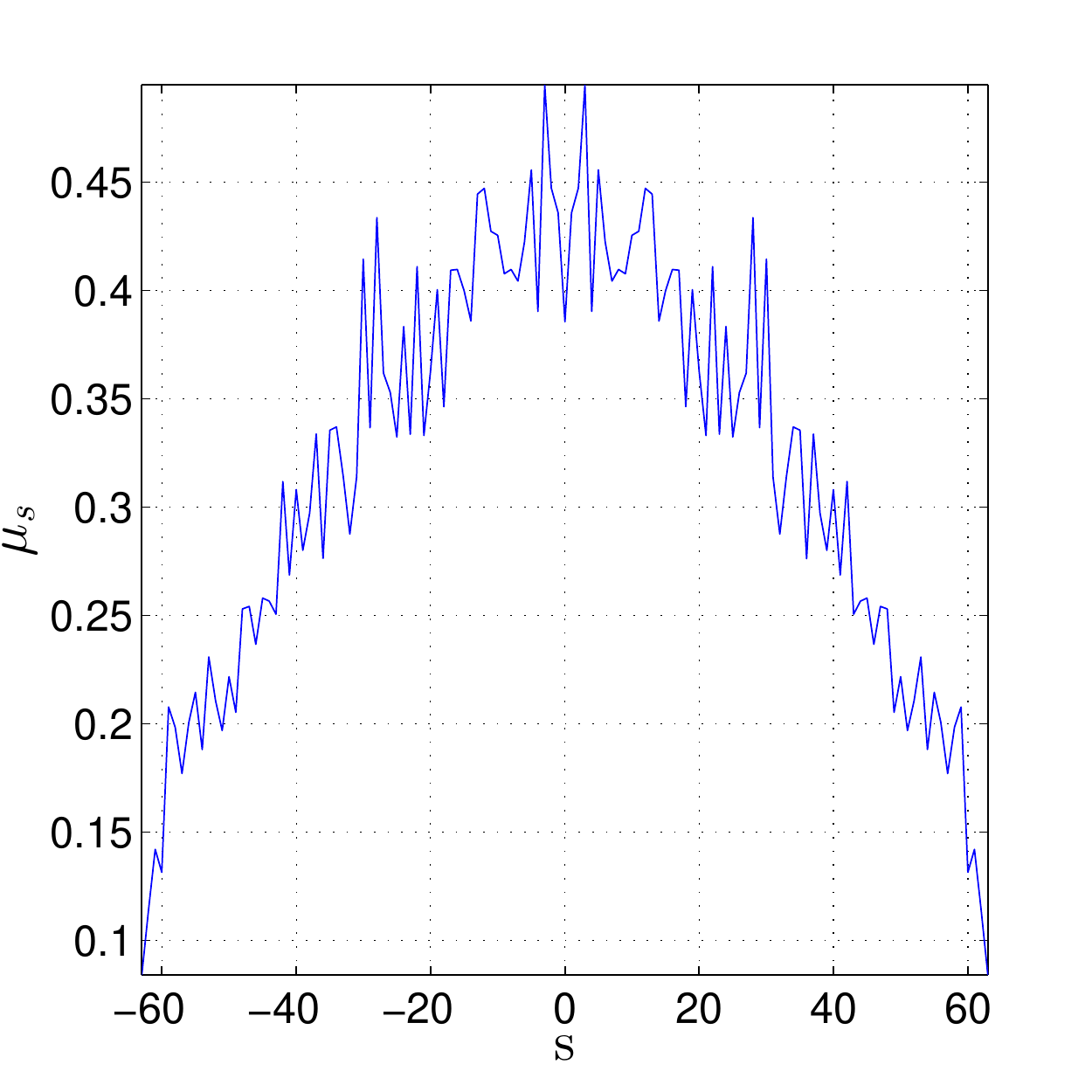}
		\caption{}
		\label{fig:Mus}
	\end{subfigure}
	\begin{subfigure}[t]{.24\textwidth}
		\includegraphics[trim = 50 0 0 50, width=\textwidth]{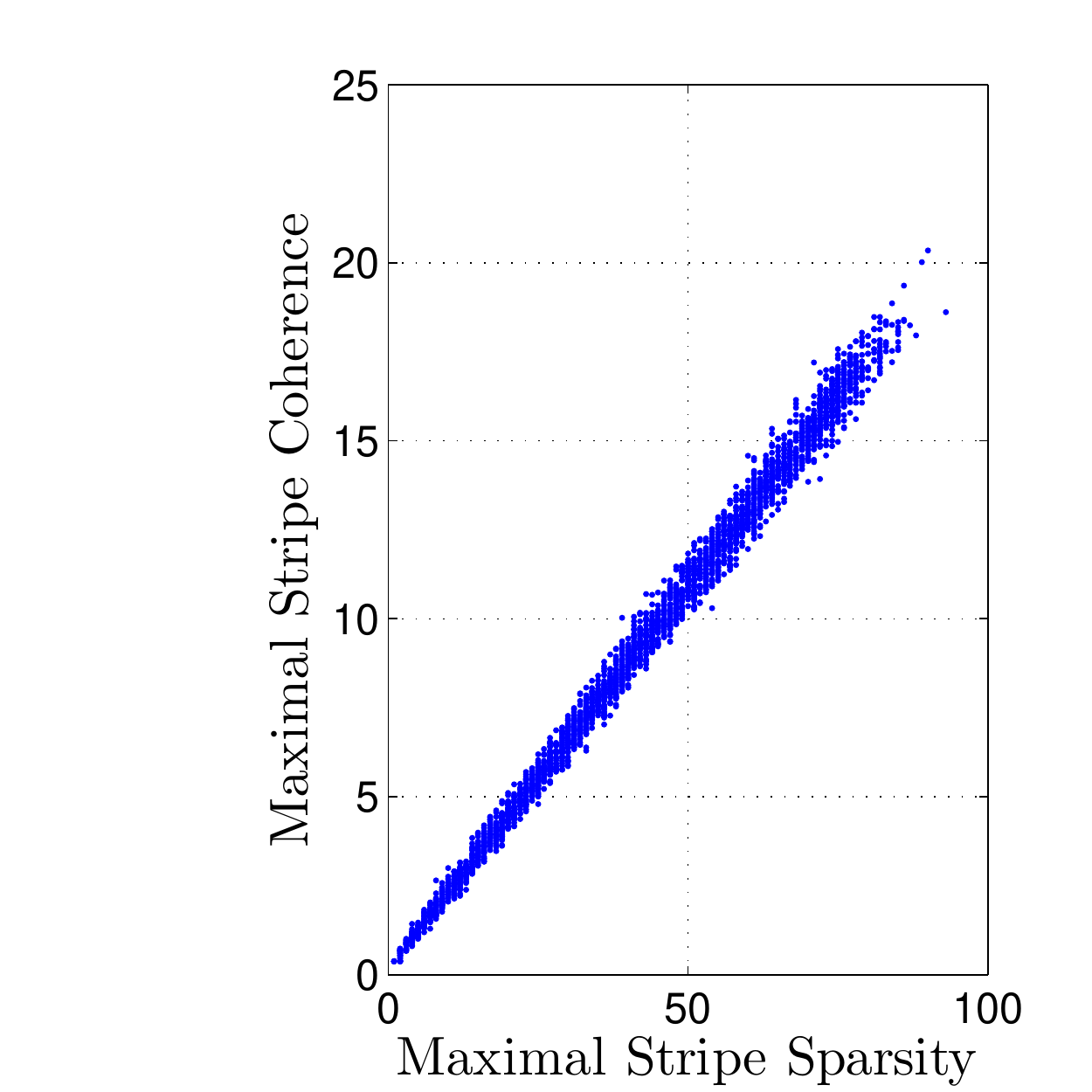}
		\caption{}
		\label{fig:L0inf_sc}
	\end{subfigure}
	\caption{Left: the shifted mutual coherence as function of the shift. The larger the shift between the atoms, the lower $\mu_s$ is expected to be. Right: the maximal stripe coherence as a function of the $\Loi$ norm, for random realizations of global sparse vectors.}
	\vspace{-0.5cm}
\end{center}
\end{figure}

Given this dictionary, we generate sparse vectors with random supports of cardinalities in the range $\left[1,300\right]$. For each sparse vector we compute its $\Loi$ norm by searching for the densest stripe, and its maximal stripe coherence using the convolution mentioned above. In Figure \ref{fig:L0inf_sc} we illustrate the connection between the $\Loi$ norm and the maximal stripe coherence for this set of sparse vectors. As expected, the $\Loi$ norm and the maximal stripe coherence are highly correlated. Although the theorems based on the stripe coherence are sharper, they are harder to comprehend. In this experiment we attempted to alleviate this by showing an intuitive connection between the two.

We now present a theorem relating the stripe coherences of related sparse vectors.
\begin{thm} \label{Thm:max_stripe_coherence_contained}
Let $\Gama_1$ and $\Gama_2$ be two global sparse vectors such that the support of $\Gama_1$ is contained in the support of $\Gama_2$. Then the maximal stripe coherence of $\Gama_1$ is less or equal than that of $\Gama_2$.
\end{thm}

\begin{proof}
Denote by $\gama_i^1$ and $\gama_i^2$ the $i^{th}$ stripe extracted from $\Gama_1$ and $\Gama_2$, respectively. Also, denote by $n_{i,s}^{1}$ and $n_{i,s}^{2}$ the number of non-zeros in the $s^{th}$ shift of $\gama_i^1$ and $\gama_i^2$, respectively. Since the support of $\Gama_1$ is contained in the support of $\Gama_2$, we have that $\forall i, s \quad n_{i,s}^{1}\leq n_{i,s}^{2}$. As a result, we have that
\begin{equation}
\max_i\sum_{s=-n+1}^{n-1} n_{i,s}^{1} \mu_s\leq
\max_i\sum_{s=-n+1}^{n-1} n_{i,s}^{2} \mu_s.
\end{equation}
The left-hand side of the above inequality is the maximal stripe coherence of $\Gama_1$, while the right-hand side is the corresponding one for $\Gama_2$, proving the hypothesis.
\end{proof}

\begin{customthm}{12}{(Global OMP recovery guarantee using the stripe coherence):}
	Given the system of linear equations $\X = \D\Gama$, if a solution $\Gama$ exists satisfying \begin{equation}
	\max_i\ \zeta_i = \max_i\sum_{s=-n+1}^{n-1} n_{i,s} \mu_s < \ \frac{1}{2}\left(1+\mu_0\right),
	\end{equation}
	then OMP is guaranteed to recover it.
\end{customthm}
\begin{proof}
	
	The first steps of this proof are exactly those derived in proving Theorem \ref{thm:OMPSuccess}, and are thus omitted for the sake brevity. Recall that in order for the first step of OMP to succeed, we require
	\begin{equation} \label{eq:inequality2}
	\left| \sum_{t\in \mathcal{T}}\Gamma_t\d_t^T\d_i \right| > \max_{j\notin\mathcal{T}} \left| \sum_{t\in \mathcal{T}}\Gamma_t\d_t^T\d_j \right|.
	\end{equation}
	Lower bounding the left hand side of the above inequality, we can write
	\begin{equation}
	\left| \sum_{t\in \mathcal{T}}\Gamma_t\d_t^T\d_i \right| \geq |\Gamma_i| - |\Gamma_i|\sum_{t \in \mathcal{T}_{p(i)},t\neq i}|\d_t^T\d_i |,
	\end{equation}
	as stated previously in Equation \eqref{Eq:EquationFromOMPproof}.
	Instead of summing over the support $\mathcal{T}_{p(i)}$, we can sum over all the supports $\mathcal{T}_{{p(i)},s}$, which correspond to all possible shifts. We can then write
	\begin{equation}
	\left| \sum_{t\in \mathcal{T}}\Gamma_t\d_t^T\d_i \right| \geq |\Gamma_i| -  |\Gamma_i|  \sum_{s = -n+1}^{n-1} \ \sum_{\substack{t \in \mathcal{T}_{{p(i)},s} \\ t\neq i}} |\d_t^T\d_i |.
	\end{equation}
	We can bound the right term by using the number of non-zeros in each sub-support $\mathcal{T}_{{p(i)},s}$, denoted by $n_{{p(i)},s}$, together with the corresponding shifted mutual coherence $\mu_s$. Also, we can disregard the constraint $t\neq i$ in the above summation by subtracting an extra $\mu_0$ term, obtaining:
	\begin{equation}
	\left| \sum_{t\in \mathcal{T}}\Gamma_t\d_t^T\d_i \right| \geq |\Gamma_i| -  |\Gamma_i| \left( \sum_{s = -n+1}^{n-1}\mu_s n_{p(i),s}-\mu_0 \right).
	\end{equation}
	Bounding the above by the maximal stripe coherence, we obtain
	\begin{equation}
	\left| \sum_{t\in \mathcal{T}}\Gamma_t\d_t^T\d_i \right| \geq |\Gamma_i| -  |\Gamma_i|  \left( \max_k\sum_{s = -n+1}^{n-1}\mu_s n_{k,s}-\mu_0 \right).
	\end{equation}
	In order to upper bound the right hand side of Equation \eqref{eq:inequality2} we follow the steps leading to Equation \eqref{eq:SkipSteps2}, resulting in
	\begin{equation}
	\max_{j\notin\mathcal{T}}\left| \sum_{t\in \mathcal{T}}\Gamma_t\d_t^T\d_j \right| \leq |\Gamma_i|\max_{j\notin\mathcal{T}}\sum_{t \in \mathcal{T}_{p(j)}}|\d_t^T\d_j |.
	\end{equation}
	Using a similar decomposition of the support and the definition of the shifted mutual coherence, we have
	\begin{align}
	\max_{j\notin\mathcal{T}}\left| \sum_{t\in \mathcal{T}}\Gamma_t\d_t^T\d_j \right| &\leq |\Gamma_i|\max_{j\notin\mathcal{T}} \sum_{s=-n+1}^{n-1} \ \sum_{t\in\mathcal{T}_{p(j),s}} |\d_t^T\d_j |\\
	& \leq |\Gamma_i|\max_{j\notin\mathcal{T}} \sum_{s=-n+1}^{n-1} \mu_s n_{p(j),s}.
	\end{align}
	Once again bounding this expression by the maximal stripe coherence, we obtain
	\begin{equation}
	\max_{j\notin\mathcal{T}}\left| \sum_{t\in \mathcal{T}}\Gamma_t\d_t^T\d_j \right| \leq |\Gamma_i|\cdot\max_k\sum_{s=-n+1}^{n-1} \mu_s n_{k,s}.
	\end{equation}
	Using both bounds, we have that
	\begin{align*}
	\left| \sum_{t\in \mathcal{T}}\Gama_t \d_t^T\d_i \right|
	& \geq|\Gamma_i| -  |\Gamma_i| \left( \max_k\sum_{s=-n+1}^{n-1} \mu_s n_{k,s}-\mu_0 \right) \\
	& >|\Gamma_i|\cdot \max_k\sum_{s=-n+1}^{n-1} \mu_s n_{k,s} \\
	& \geq\max_{j\notin\mathcal{T}}\left| \sum_{t\in \mathcal{T}}\Gamma_t\d_t^T\d_j \right|.
	\end{align*}
	Thus,
	\begin{equation}
	1-\max_k\sum_{s=-n+1}^{n-1} \mu_s n_{k,s}+\mu_0 \ > \ \max_k\sum_{s=-n+1}^{n-1} \mu_s n_{k,s}.
	\end{equation}
	Finally, we obtain
	\begin{equation}
	\max_k\ \zeta_k=\max_k\sum_{s=-n+1}^{n-1} \mu_s n_{k,s}\ < \ \frac{1}{2}\left(1+\mu_0\right),
	\end{equation}
	which is the requirement stated in the theorem. Thus, this condition guarantees the success of the first OMP step, implying it will choose an atom inside \mbox{the true support $\mathcal{T}$.}
		
	The next step in the OMP algorithm is an update of the residual. This is done by decreasing a term proportional to the chosen atom (or atoms within the correct support in subsequent iterations) from the signal. Thus, the support of this residual is contained within the support of the true signal. As a result, according to the previous theorem, the maximal stripe coherence corresponding to the residual is less or equal to the one of the true sparse code $\Gama$. Using the same set of steps we obtain that the condition on the maximal stripe coherence guarantees that the algorithm chooses again an atom from the true support of the solution. Furthermore, the orthogonality enforced by the least-squares step guarantees that the same atom is never chosen twice. As a result, after $\|\Gama\|_0$ iterations the OMP will find all the atoms in the correct support, reaching a residual equal to zero.
\end{proof}

We have provided two theorems for the success of the OMP algorithm. Before concluding, we aim to show that assuming $\mu(\D)=\mu_0$, the guarantee based on the stripe coherence is at least as strong as the one based on the $\Loi$ norm. Assume the recovery condition using the $\Loi$ norm is met and as such $\|\Gama\|_{0,\infty} = \max_i \ n_i < \ \frac{1}{2}\left(1+\frac{1}{\mu(\D)}\right)$, where $n_i$ is equal to $\|\gama_i\|_0$. Multiplying both sides by $\mu(\D)$ we obtain \mbox{$\max_i \ n_i\cdot\mu(\D) < \ \frac{1}{2}\left(1+\mu(\D)\right)$}. Using the above inequality and the properties:
\begin{equation}
1) \sum\limits_{s=-n+1}^{n-1} n_{i,s}=n_i, \qquad 2)\  \forall s \quad \mu_s\leq\mu(\D),
\end{equation}
we have that
\begin{align}
\max_i\sum_{s=-n+1}^{n-1} n_{i,s} \mu_s&\leq\max_i \sum_{s=-n+1}^{n-1} n_{i,s}\mu(\D)\\
&=\max_i \ n_i\cdot\mu(\D) < \ \frac{1}{2}\left(1+\mu(\D)\right).
\end{align}
Thus, we obtain that
\begin{equation}
\max_i\sum_{s=-n+1}^{n-1} n_{i,s} \mu_s < \ \frac{1}{2}\left(1+\mu(\D)\right)=\frac{1}{2}\left(1+\mu_0\right),
\end{equation}
where we have used our assumption that $\mu(\D)=\mu_0$. We conclude that if the recovery condition based on the $\Loi$ norm is met, then so is the one based on the stripe coherence. As a result, the condition based on the stripe coherence is at least as strong as the one based on the $\Loi$ norm.

As a final note, we mention that assuming $\mu(\D)= \underset{s}{\max}\ \mu_s = \mu_0$ is in fact a reasonable assumption. Recall that in order to compute $\mu_s$ we evaluate inner products between atoms which are $s$ indexes shifted from each other. As a result, the higher the shift $s$ is, the less overlap the atoms have, and the less $\mu_s$ is expected to be. Thus, we expect the value $\mu_0$ to be the largest or close to it in most cases.

%% ---------------------------------------------------------------------------------------------------------------

\section{Theoretical Analysis of Corrupted Signals}

\begin{customthm}{15}{(Upper bounding the SRIP via the mutual coherence):}
	For a convolutional dictionary $\D$ with global mutual coherence $\mu(\D)$, the SRIP can be upper-bounded by
	\begin{equation}
	\delta_k\leq(k-1)\mu(\D).
	\end{equation}
\end{customthm}

\begin{proof}
Consider the sub-dictionary $\D_{\mathcal{T}}$, obtained by restricting the columns of $\D$ to a support $\mathcal{T}$ with $\Loi$ norm equal to $k$. Lemma 1 states that the eigenvalues of the Gram matrix $\D_{\mathcal{T}}^T\D_{\mathcal{T}}$ are bounded by
\begin{equation} \label{Gersh}
1-(k-1)\mu(\D)\leq\lambda_i(\D_{\mathcal{T}}^T\D_{\mathcal{T}})\leq 1+(k-1)\mu(\D).
\end{equation}
Now, for every $\Delt$ we have that
\begin{align*}
(1-(k-1)\mu(\D))\|\Delt\|_2^2 \leq & \lambda_{min}(\D_{\mathcal{T}}^T\D_{\mathcal{T}})\|\Delt\|_2^2 \\
\leq & \|\D_{\mathcal{T}} \Delt\|_2^2 \leq \lambda_{max}(\D_{\mathcal{T}}^T\D_{\mathcal{T}})\|\Delt\|_2^2 \\
\leq & (1+(k-1)\mu(\D))\|\Delt\|_2^2,
\end{align*}
where $\lambda_{max}$ and $\lambda_{min}$ are the maximal and minimal eigenvalues, respectively.
As a result, we obtain that $\delta_k\leq(k-1)\mu(\D)$.
\end{proof}

\begin{customthm}{16}{(Stability of the solution to the $\Poie$ problem):}
Consider a sparse vector $\Gama$ such that $\|\Gama\|_{0,\infty} = k < \frac{1}{2}\left( 1 + \frac{1}{\mu(\D)} \right) $, and a convolutional dictionary $\D$ satisfying the SRIP property for $\Loi=2k$ with coefficient $\delta_{2k}$. Then, the distance between the true sparse vector $\Gama$ and the solution to the $\Poie$ problem $\hat{\Gama}$ is bounded by
\begin{equation}
\|\Gama-\hat{\Gama}\|_2^2\leq \frac{4\epsilon^2}{1-\delta_{2k}}\leq\frac{4\epsilon^2}{1-(2k-1)\mu(\D)}.	
\end{equation}
\end{customthm}

\begin{proof}
The solution to the $\Poie$ problem satisfies $\|\Y-\D\hat{\Gama}\|_2^2\leq\epsilon^2$, and it must also satisfy $\|\hat{\Gama}\|_{0,\infty}\leq\|\Gama\|_{0,\infty}$ (since $\hat{\Gama}$ is the solution with the minimal $\Loi$ norm).
Defining $\Delt=\Gama-\hat{\Gama}$, using the triangle inequality, we have that $\|\D\Delt\|_2^2=\|\D\Gama-\Y+\Y-\D\hat{\Gama}\|_2^2\leq4\epsilon^2$. Furthermore, since the $\Loi$ norm satisfies the triangle inequality as well, we have that $\|\Delt\|_{0,\infty}=\|\Gama-\hat{\Gama}\|_{0,\infty}\leq\|\Gama\|_{0,\infty}+\|\hat{\Gama}\|_{0,\infty} \leq 2k$. 
Using the SRIP of $\D$, we have that
\begin{equation}
(1-\delta_{2k})\|\Delt\|_2^2\leq\|\D\Delt\|_2^2\leq 4\epsilon^2,
\end{equation}
where in the first inequality we have used the lower bound provided by the definition of the SRIP. Finally, we obtain the following stability claim:
\begin{equation}
\|\Delt\|_2^2=\|\Gama-\hat{\Gama}\|_2^2\leq \frac{4\epsilon^2}{1-\delta_{2k}}.
\end{equation}
Using our bound of the SRIP in terms of the mutual coherence, we obtain that
\begin{equation}
\|\Delt\|_2^2=\|\Gama-\hat{\Gama}\|_2^2\leq \frac{4\epsilon^2}{1-\delta_{2k}}\leq\frac{4\epsilon^2}{1-(2k-1)\mu(\D)}.
\end{equation}
For the last inequality to hold, we have assumed $k = \|\Gama\|_{0,\infty}<\frac{1}{2}(1+\frac{1}{\mu(\D)})$.

\end{proof}

\begin{customthm}{17}{(Stable recovery of global OMP in the presence of noise):}
	Suppose a clean signal $\X$ has a representation $\D\Gama$, and that it is contaminated with noise $\E$ to create the signal $\Y=\X+\E$, such that $\|\Y-\X\|_2\leq\epsilon$. Denote by $\epsilon_{_L}$ the highest energy of all $n$-dimensional local patches extracted from $\E$. Assume $\Gama$ satisfies
	\begin{equation}
	\|\Gama\|_{0,\infty} < \frac{1}{2}\left( 1+\frac{1}{\mu(\D)} \right)-\frac{1}{\mu(\D)}\cdot\frac{\epsilon_{_L}}{|\Gamma_{min}|},
	\end{equation}
	where $|\Gamma_{min}|$ is the minimal entry in absolute value of the sparse vector $\Gama$.
	Denoting by $\Gama_\text{OMP}$ the solution obtained by running OMP for $\|\Gama\|_0$ iterations, we are guaranteed that
	\begin{enumerate}[ a) ]
	\item OMP will find the correct support; And,
	\item $\|\Gama_\text{OMP}-\Gama\|_2^2\leq\frac{\epsilon^2}{1-\mu(\|\Gama\|_{0,\infty}-1)}$.
	\end{enumerate}
\end{customthm}

\begin{proof}
	We shall first prove that the first step of OMP succeeds in recovering an element from the correct support.
	Denoting by $\mathcal{T}$ the support of $\Gama$, we can write
	\begin{equation}
	\Y = \D\Gama + \E = \sum_{t\in \mathcal{T}} \Gamma_t \d_t + \E.
	\label{GlobalExpression}
	\end{equation}
	Suppose that $\Gama$ has its largest coefficient in absolute value in $\Gamma_i$. For the first step of OMP to choose one of the atoms in the support, we require
	\begin{equation}
	|\d_i^T \Y | > \max_{j\notin\mathcal{T}} | \d_j^T \Y |.
	\end{equation}
	Substituting Equation \eqref{GlobalExpression} in this requirement we obtain
	\begin{equation} \label{eq:inequality}
	\left| \sum_{t\in \mathcal{T}}\Gamma_t\d_t^T\d_i + \E^T\d_i\right| > \max_{j\notin\mathcal{T}} \left| \sum_{t\in \mathcal{T}}\Gamma_t\d_t^T\d_j + \E^T\d_j\right|.
	\end{equation}
	Using the reverse triangle inequality we can construct a lower bound for the left hand side:
	\begin{align}
	\left| \sum_{t\in \mathcal{T}}\Gamma_t\d_t^T\d_i + \E^T\d_i\right|
	&\geq\left| \sum_{t\in \mathcal{T}}\Gamma_t\d_t^T\d_i\right| - \left|\E^T\d_i\right|.
	\end{align}
	Our next step is to bound the absolute value of the inner product of the noise and the atom $\d_i$. A na\"ive approach, based on the Cauchy-Schwarz inequality and the normalization of the atoms, would be to bound the inner product as $|\E^T\d_i| \leq \|\E\|_2\cdot\|\d_i\|_2 \leq \epsilon$. However, such bound would disregard the local nature of the atoms. Due to their limited support we have that $\d_i=\R_i^T\R_i\d_i$ where, as previously defined, $\R_i$ extracts a $n$-dimensional patch from a $N$-dimensional signal. Based on this observation, we have that
	\begin{equation}
	|\E^T\d_i| = |\E^T\R_i^T\R_i\d_i| \leq \|\R_i\E\|_2\cdot\|\d_i\|_2  \leq \epsilon_{_L},
	\end{equation}		
	where we have used the fact that \mbox{$\|\R_i \E\|_2 \leq \epsilon_{_L} \ \forall\ i$}. By exploiting the locality of the atom, together with the assumption regarding the maximal local energy of the noise, we are able to obtain a much tighter bound, because $\epsilon_{_L} \ll \epsilon$ in general. As a result, we obtain
	\begin{equation}
	\left| \sum_{t\in \mathcal{T}}\Gamma_t\d_t^T\d_i + \E^T\d_i\right| \geq\left| \sum_{t\in \mathcal{T}}\Gamma_t\d_t^T\d_i\right| - \epsilon_{_L}.
	\end{equation}
	Using the reverse triangle inequality, the normalization of the atoms and the fact that $|\Gamma_i|\geq|\Gamma_t|$, we obtain
	\begin{align}
	\left| \sum_{t\in \mathcal{T}}\Gamma_t\d_t^T\d_i + \E^T\d_i\right|
	&\geq |\Gamma_i| - \sum_{t \in \mathcal{T},t\neq i}|\Gamma_t|\cdot|\d_t^T\d_i |-\epsilon_{_L} \\
	&\geq |\Gamma_i| - |\Gamma_i|\sum_{t \in \mathcal{T},t\neq i}|\d_t^T\d_i |-\epsilon_{_L}.
	\end{align}
	Notice that $\d_t^T\d_i$ is zero for every atom too far from $\d_i$ because the atoms do not overlap. Denoting the stripe which fully contains the $i^{th}$ atom as $p(i)$ and its support as $\mathcal{T}_{p(i)}$, we can restrict the summation as:
	\begin{equation}
	\left| \sum_{t\in \mathcal{T}}\Gamma_t\d_t^T\d_i + \E^T\d_i\right| \geq |\Gamma_i| - |\Gamma_i|\sum_{\substack{t \in \mathcal{T}_{p(i)},\\ t\neq i}}|\d_t^T\d_i |-\epsilon_{_L}.
	\end{equation}
	Denoting by $n_{p(i)}$ the number of non-zeros in the support $\mathcal{T}_{p(i)}$ and using the definition of the mutual coherence we obtain:
	\begin{align*}
	\left| \sum_{t\in \mathcal{T}}\Gamma_t\d_t^T\d_i + \E^T\d_i\right| &	\geq |\Gamma_i| - |\Gamma_i| (n_{p(i)}-1)\mu(\D) - \epsilon_{_L}\\
	& \geq |\Gamma_i| - |\Gamma_i| (\|\Gama\|_{0,\infty}-1)\mu(\D) - \epsilon_{_L}.
	\end{align*}
	In the last inequality we have used the definition of the $\Loi$ norm.
	
	Now, we construct an upper bound for the right hand side of equation \eqref{eq:inequality}, once again using the triangle inequality and the fact that $|\E^T\d_j|\leq\epsilon_{_L}$:
	\begin{align}
	\max_{j\notin\mathcal{T}}\left| \sum_{t\in \mathcal{T}}\Gamma_t\d_t^T\d_j + \E^T\d_j \right|
	% & \leq \max_{j\notin\mathcal{T}}\left| \sum_{t\in \mathcal{T}}\Gamma_t\d_t^T\d_j\right| + \left|\E^T\d_j \right|\\
	& \leq \max_{j\notin\mathcal{T}}\left| \sum_{t\in \mathcal{T}}\Gamma_t\d_t^T\d_j\right| + \epsilon_{_L}.
	\end{align}
	Using the same rationale as before we get
	\begin{align}
	\max_{j\notin\mathcal{T}}\left| \sum_{t\in \mathcal{T}}\Gamma_t\d_t^T\d_j + \E^T\d_j \right| 
	&\leq |\Gamma_i|\max_{j\notin\mathcal{T}} \sum_{t \in \mathcal{T}}|\d_t^T\d_j |+\epsilon_{_L}\\
	&\leq |\Gamma_i|\max_{j\notin\mathcal{T}} \sum_{t \in \mathcal{T}_{p(j)}}|\d_t^T\d_j |+\epsilon_{_L}\\
	&\leq |\Gamma_i|\cdot\|\Gama\|_{0,\infty}\cdot\mu(\D)+\epsilon_{_L}.
	\end{align}
%	where, as before, $p(j)$ denotes the stripe which fully contains the $j^{th}$ atom. Once again, we have constrained the summation to elements which have non-zero inner product with the atom $\d_j$, those found inside the support $\mathcal{T}_{p(j)}$. Denoting by $n_{p(j)}$ the number of non-zeros in the support $\mathcal{T}_{p(j)}$ and using the definitions of the mutual coherence and the $\Loi$ norm we obtain:
%	\begin{align}
%	\max_{j\notin\mathcal{T}}\left| \sum_{t\in \mathcal{T}}\Gamma_t\d_t^T\d_j + \E^T\d_j \right|
%	&\leq |\Gamma_i|\max_{j\notin\mathcal{T}}\sum_{t \in \mathcal{T}_{p(j)}}|\d_t^T\d_j |+ \epsilon_{_L} \\
%	&\leq |\Gamma_i|\max_{j\notin\mathcal{T}}\ n_{p(j)}\cdot\mu(\D)+\epsilon_{_L} \\
%	&\leq |\Gamma_i|\cdot\|\Gama\|_{0,\infty}\cdot\mu(\D)+\epsilon_{_L}.
%	\end{align}
	Using both bounds, we obtain
	\begin{align*}
	& \left| \sum_{t\in \mathcal{T}}\Gamma_t\d_t^T\d_i + \E^T\d_i\right|
	\geq |\Gamma_i| - |\Gamma_i| (\|\Gama\|_{0,\infty}-1)\mu(\D) - \epsilon_{_L} \\
	\geq & |\Gamma_i|\cdot \|\Gama\|_{0,\infty}\mu(\D)+\epsilon_{_L}
	\geq\max_{j\notin\mathcal{T}}\left| \sum_{t\in \mathcal{T}}\Gamma_t\d_t^T\d_j + \E^T\d_j \right|.
	\end{align*}
	From this, it follows that
	\begin{align}
	\|\Gama\|_{0,\infty}&\leq\frac{1}{2}\left(1+\frac{1}{\mu(\D)}\right)-\frac{1}{\mu(\D)}\cdot\frac{\epsilon_{_L}}{|\Gamma_i|}.
	\label{Eq:ThisIneq}
	\end{align}
	Note that the theorem's hypothesis assumes that the above holds for $|\Gamma_{\text{min}}|$ instead of $|\Gamma_i|$. However, because $|\Gamma_i| \geq |\Gamma_{min}|$, this condition holds for every $i$. Therefore, Equation \eqref{Eq:ThisIneq} holds and we conclude that the first step of OMP succeeds.  
	
	Next, we address the success of subsequent iterations of the OMP.
	Define the sparse vector obtained after $k<\|\Gama\|_0$ iterations as $\Lamda^k$, and denote its support by $\mathcal{T}^k$. Assuming that the algorithm identified correct atoms (i.e., has so far succeeded), $\mathcal{T}^k = \text{supp}\{\Lamda_k\}\subset\text{supp}\{\Gama\}$. The next step in the algorithm is the update of the residual. This is done by decreasing a term proportional to the chosen atoms from the signal; i.e., 
	\begin{equation}
	\Y^k = \Y - \sum_{i\in \mathcal{T}^k} \d_i \Lambda_i^k.
	\end{equation}
	Moreover, $\Y^k$ can be seen as containing a clean signal $\X^k$ and the noise component $\E$, where 
	\begin{equation}
	\X^k = \X - \sum_{i\in \mathcal{T}^k} \d_i \Lambda_i^k = \D\Gama^k.
	\end{equation}
	The objective is then to recover the support of the sparse vector corresponding to $\X^k$, $\Gama^k$, defined as\footnote{Note that if $k=0$, $\X_0 = \X$, $\Y_0 = \Y$, and $\Gama_0 = \Gama$.}
	\begin{equation}  \label{eq:defGama_k}
	\Gamma_i^k =
	\left\{
	\begin{array}{ll}
	\Gamma_i - \Lambda_i^k & \mbox{if } \ i\in \mathcal{T}^k \\
	\Gamma_i  & \mbox{if } \ i\notin \mathcal{T}^k.
	\end{array}
	\right.
	\end{equation}	
	Note that $\text{supp}\{\Gama^k\} \subseteq \text{supp}\{\Gama\}$ and so
	\begin{equation} \label{eq:omp_first_item}
	\|\Gama^k\|_{0,\infty}\leq \|\Gama\|_{0,\infty}.
	\end{equation}
	In words, the $\Loi$ norm of the underlying solution of $\X^k$ does not increase as the iterations proceed. Note that this representation is also unique in light of the uniqueness theorem presented in part I.
	From the above definitions, we have that
	\begin{align}
	\Y^k - \X^k &= \Y - \sum_{i\in \mathcal{T}^k} \d_i \Lambda_i^k - \X + \sum_{i\in \mathcal{T}^k} \d_i \Lamda_i^k \\
	&= \Y - \X = \E.
	\end{align}
	Hence, the noise level is preserved, both locally and globally; both $\epsilon$ and $\epsilon_L$ remain the same. 
	
	Note that $\Gama^k$ differs from $\Gama$ in at most $k$ places, following Equation \eqref{eq:defGama_k} and that $|\mathcal{T}^k|=k$. As such, $\|\Gama^k\|_{\infty}$ is greater than the $(k+1)^{th}$ largest element in absolute value in $\Gama$. This implies that $\|\Gama^k\|_{\infty}\geq |\Gamma_{\min}|$.
	Finally, we obtain that
	\begin{align}
	\|\Gama^k\|_{0,\infty} \leq \|\Gama\|_{0,\infty}
	& < \frac{1}{2}\left( 1+\frac{1}{\mu(\D)} \right)-\frac{1}{\mu(\D)}\cdot\frac{\epsilon_{_L}}{|\Gamma_{min}|} \\
	& \leq \frac{1}{2}\left( 1+\frac{1}{\mu(\D)} \right)-\frac{1}{\mu(\D)}\cdot\frac{\epsilon_{_L}}{\|\Gama^k\|_\infty}.
	\end{align}
	The first inequality is due to \eqref{eq:omp_first_item}, the second is the assumption in \eqref{omp_hypothesis} and the third was just obtained above. Thus,
	\begin{equation}
	\|\Gama^k\|_{0,\infty} < \frac{1}{2}\left( 1+\frac{1}{\mu(\D)} \right)-\frac{1}{\mu(\D)}\cdot\frac{\epsilon_{_L}}{\|\Gama^k\|_\infty}.
	\end{equation}		
	Similar to the first iteration, the above inequality together with the fact that the noise level is preserved, guarantees the success of the next iteration of the OMP algorithm. From this follows that the algorithm is guaranteed to recover the true support after $\|\Gama\|_0$ iterations.
	
	Finally, we move to prove the second claim. In its last iteration OMP solves the following problem:
	\begin{equation}
	\Gama_{OMP}=\arg\min_\Delt\|\D_{\mathcal{T}}\Delt-\Y\|_2^2,
	\end{equation}
	where $\D_{\mathcal{T}}$ is the convolutional dictionary restricted to the support $\mathcal{T}$ of the true sparse code $\Gama$. Denoting $\Gama_\mathcal{T}$ the (dense) vector corresponding to those atoms, the solution to the above problem is simply given by
	\begin{align}
	\Gama_{OMP}
	&=\D_\mathcal{T}^{\dagger}\Y
	=\D_\mathcal{T}^{\dagger}\left( \D\Gama + \E \right) \\
	&=\D_\mathcal{T}^{\dagger}\left( \D_\mathcal{T}\Gama_\mathcal{T} + \E \right)
	=\Gama_\mathcal{T} + \D_\mathcal{T}^{\dagger}\E,
	\end{align}
	where we have denoted by $\D_\mathcal{T}^{\dagger}$ the Moore-Penrose pseudoinverse of the sub-dictionary $\D_\mathcal{T}$.
	Thus,
	\begin{align}
	\|\Gama_{OMP}-\Gama_\mathcal{T}\|_2^2 = \|\D_\mathcal{T}^{\dagger}\E\|_2^2
	&\leq \|\D_\mathcal{T}^{\dagger}\|^2_2\cdot\|\E\|^2_2 \\
	= \frac{1}{\lambda_{\min}\left(\D_\mathcal{T}^T\D_\mathcal{T} \right)}\|\E\|^2_2
	&\leq\frac{\epsilon^2}{1-\mu(\D)(\|\Gama\|_{0,\infty}-1)}.
	\end{align}
	In the last inequality we have used the bound on the eigenvalues of $\D_\mathcal{T}^T\D_\mathcal{T}$ derived in Lemma 1.
	
\end{proof}

\begin{customthm}{18}{(ERC in the convolutional sparse model):}
	For a convolutional dictionary $\D$ with mutual coherence $\mu(\D)$, the ERC condition is met for every support $\mathcal{T}$ that satisfies
	\begin{equation}
		\|\mathcal{T}\|_{0,\infty} < \frac{1}{2}\left(1+\frac{1}{\mu(\D)}\right).
	\end{equation}
\end{customthm}
\begin{proof}
	For the ERC to be satisfied, we must require that, for every $i \notin \mathcal{T}$, 
	\begin{equation}
	\| \D^{\dagger}_\mathcal{T} \d_i \|_1  =  \left\| \left( \D^T_\mathcal{T} \D_\mathcal{T} \right)^{-1} \D^T_\mathcal{T} \d_i \right\|_1 < 1.
	\end{equation}
	Using properties of induced norms, we have that
	\begin{equation} \label{eq:ERC_multiplicative}
	\left\| \left( \D^T_\mathcal{T} \D_\mathcal{T} \right)^{-1} \D^T_\mathcal{T} \d_i \right\|_1
	\leq \left\| \left( \D^T_\mathcal{T} \D_\mathcal{T} \right)^{-1} \right\|_1 \left\| \D^T_\mathcal{T} \d_i \right\|_1.
	\end{equation}
	Using the definition of the mutual coherence, it is easy to see that the absolute value of the entries in the vector $\D^T_\mathcal{T} \d_i$ are bounded by $\mu(\D)$. Moreover, due to the locality of the atoms, the number of non-zero inner products with the atom $\d_i$ is equal to the number of atoms in $\mathcal{T}$ that overlap with it. This number can, in turn, be bounded by the maximal number of non-zeros in a stripe from $\mathcal{T}$, i.e., its $\Loi$ norm, denoted by $k$. Therefore, $\left\| \D^T_\mathcal{T} \d_i \right\|_1 \leq k \mu(\D)$. 
	
	Addressing now the first term in Equation \eqref{eq:ERC_multiplicative}, note that
	\begin{equation} \label{eq:ERC_gram_l1_bound}
	\left\| \left( \D^T_\mathcal{T} \D_\mathcal{T} \right)^{-1} \right\|_1 = \left\| \left( \D^T_\mathcal{T} \D_\mathcal{T} \right)^{-1} \right\|_\infty,
	\end{equation}
	since the induced $\ell_1$ and $\ell_\infty$ norms are equal for symmetric matrices. Next, using the Ahlberg-Nilson-Varah bound and similar steps to those presented in Lemma 1, we have that
	\begin{equation} \label{eq:ERC_gram_infinity_bound}
	\left\| \left( \D^T_\mathcal{T} \D_\mathcal{T} \right)^{-1} \right\|_\infty \leq \frac{1}{1 - (k-1) \mu(\D) }.
	\end{equation}
	In order for this to hold, we must require the Gram $\D^T_\mathcal{T} \D_\mathcal{T}$ to be diagonally dominant, which is satisfied if $1-(k-1)\mu(\D) > 0$. This is indeed the case, as follows from the assumption on the $\Loi$ norm of $\mathcal{T}$. Plugging the above into Equation \eqref{eq:ERC_multiplicative}, we obtain
	\begin{align} \label{eq:ERC_theta_bound}
	\left\| \left( \D^T_\mathcal{T} \D_\mathcal{T} \right)^{-1} \D^T_\mathcal{T} \d_i \right\|_1
	& \leq \left\| \left( \D^T_\mathcal{T} \D_\mathcal{T} \right)^{-1} \right\|_1 \left\| \D^T_\mathcal{T} \d_i \right\|_1 \\
	& \leq \frac{k \mu(\D)}{1 - (k-1) \mu(\D) }.
	\end{align}
	Our assumption that $k < \frac{1}{2}\left(1+\frac{1}{\mu(\D)}\right)$ implies that the above term is less than one, thus showing the ERC is satisfied for all supports $\mathcal{T}$ that satisfy $\|\mathcal{T}\|_{0,\infty} < \frac{1}{2}\left(1+\frac{1}{\mu(\D)}\right)$.
\end{proof}

\begin{customthm}{19}{(Stable recovery of global Basis Pursuit in the presence of noise):}
	Suppose a clean signal $\X$ has a representation $\D\Gama$, and that it is contaminated with noise $\E$ to create the signal $\Y=\X+\E$. Denote by $\epsilon_{_L}$ the highest energy of all $n$-dimensional local patches extracted from $\E$. Assume $\Gama$ satisfies
	\begin{equation}
	\|\Gama\|_{0,\infty} \leq \frac{1}{3} \left( 1 + \frac{1}{\mu(\D)} \right).
	\end{equation}
	Denoting by $\Gama_{\text{BP}}$ the solution to the Lagrangian BP formulation with parameter $\lambda=4\epsilon_L$, we are guaranteed that
	\begin{enumerate}
	\item The support of $\Gama_{\text{BP}}$ is contained in that of $\Gama$.
	\item $\|\Gama_{\text{BP}}-\Gama\|_\infty < \frac{15}{2}\epsilon_L$.
	\item In particular, the support of $\Gama_{\text{BP}}$ contains every index $i$ for which $|\Gamma_i|>\frac{15}{2}\epsilon_L$.
	\item The minimizer of the problem, $\Gama_{\text{BP}}$, is unique.
	\end{enumerate}
\end{customthm}

We first state and prove a Lemma that will become of use while proving the stability result of BP.

\begin{customlemma}{2}{}
\label{lemma:noise_correlation}
Suppose a clean signal $\X$ has a representation $\D\Gama$, and that it is contaminated with noise $\E$ to create the signal $\Y=\X+\E$. Denote by $\epsilon_{_L}$ the highest energy of all $n$-dimensional local patches extracted from $\E$. Assume that
\begin{equation} \label{eq:lemma_assumption}
\|\Gama\|_{0,\infty} \leq \frac{1}{2} \left( 1 + \frac{1}{\mu(\D)} \right).
\end{equation}
Denoting by $\X_{\text{LS}}$ the best $\ell_2$ approximation of $\Y$ over the support $\mathcal{T}$, we have that\footnote{We suspect that, perhaps under further assumptions, the constant in this bound can be improved from 2 to 1. This is motivated by the fact that the bound in \cite{Tropp2006}, for the traditional sparse model, is $1\cdot \epsilon$ -- where $\epsilon$ is the global noise level.}
\begin{equation} \label{eq:DT_res}
\| \D^T(\Y-\X_{\text{LS}}) \|_\infty \leq 2\epsilon_L.
\end{equation}
\end{customlemma}

\begin{proof}
Using the expression for the least squares solution (and assuming that $\DT$ has full-column rank), we have that
\begin{align*}
\DT^T(\Y-\X_{\text{LS}})
& = \DT^T \left( \Y-\DT \left( \DT^T \DT \right)^{-1} \DT^T \Y \right) \\
& = \left( \DT^T - \DT^T \DT \left( \DT^T \DT \right)^{-1} \DT^T \right) \Y = \mathbf{0}.
\end{align*}
%The $i^{th}$ entry in the vector $\D^T(\Y-\X_{\text{LS}})$ is equal to the inner product of the atom $\d_i$ with the vector $\Y-\X_\text{LS}$. 
This shows that all inner products between atoms inside $\mathcal{T}$ and the vector $\Y-\X_{\text{LS}}$ are zero, and thus $\| \DTB^T(\Y-\X_{\text{LS}}) \|_\infty = \| \D^T(\Y-\X_{\text{LS}}) \|_\infty$. We have denoted by $\overline{\mathcal{T}}$ the complement to the support, containing all atoms not found in $\mathcal{T}$, and by $\DTB$ the corresponding dictionary.
Denoting by $\Gama_\mathcal{T}$ the vector $\Gama$ restricted to its support, and expressing $\X_{\text{LS}}$ and  $\Y$ conveniently, we obtain
\begin{align}
& \| \DTB^T(\Y-\X_{\text{LS}}) \|_\infty \\
= &\| \DTB^T \left( \mathbf{I} - \DT \left( \DT^T \DT \right)^{-1} \DT^T \right) \Y \|_\infty \\
= &\| \DTB^T \left( \mathbf{I} - \DT \left( \DT^T \DT \right)^{-1} \DT^T \right) (\DT\Gama_\mathcal{T} + \E) \|_\infty. \label{eq:DTB_res}
\end{align}
It is easy to verify that 
\begin{align}
\phantom{=} \left( \mathbf{I} - \DT \left( \DT^T \DT \right)^{-1} \DT^T \right) \DT\Gama_\mathcal{T} = \mathbf{0}.
\end{align}
Plugging this into the above, we have that 
\begin{align*}
 \| \DTB^T(\Y-\X_{\text{LS}}) \|_\infty  =  \left\| \DTB^T \left( \mathbf{I} - \DT \left( \DT^T \DT \right)^{-1} \DT^T \right) \E \right\|_\infty.
%= & \left\| \DTB^T \left( \mathbf{I} - \DT \left( \DT^T \DT \right)^{-1} \DT^T \right) (\DT\Gama_\mathcal{T} + \E) \right\|_\infty \\
\end{align*}
Using the triangle inequality for the $\ell_\infty$ norm, we obtain
\begin{align}
& \| \DTB^T(\Y-\X_{\text{LS}}) \|_\infty \\
= & \left\| \DTB^T\E - \DTB^T\DT \left( \DT^T \DT \right)^{-1} \DT^T\E \right\|_\infty \\
\leq & \left\| \DTB^T\E \right\|_\infty + \left\| \DTB^T\DT \left( \DT^T \DT \right)^{-1} \DT^T\E \right\|_\infty. \label{eq:triangle_inequality_step}
\end{align}
In what follows, we will bound both terms in the above expression with $\epsilon_L$. First, due to the limited support of the atoms, $\d_i=\R_i^T\R_i\d_i$, where $\R_i$ extracts the $i^{th}$ local patch from the global signal, as previously defined. Thus,
\begin{align} \label{eq:local_noise}
\left\| \DTB^T\E \right\|_\infty
 = \max_{i\in\overline{\mathcal{T}}} | \d_i^T\E | & = \max_{i\in\overline{\mathcal{T}}} | \d_i^T\R_i^T\R_i\E | \\
& \leq \max_{i\in\overline{\mathcal{T}}} \|\R_i\d_i\|_2 \cdot \|\R_i\E\|_2 \leq \epsilon_L,
\end{align}
where we have used the Cauchy-Schwarz inequality, the normalization of the atoms and the fact that $\|\R_i\E\|_2\leq\epsilon_L$ $\forall \ i$. Next, we move to the second term in Equation \eqref{eq:triangle_inequality_step}. Using the definition of the induced $\ell_\infty$ norm, and the bound $\| \DT^T\E \|_\infty \leq \epsilon_L$, we have that
\begin{align*}
\left\| \DTB^T\DT \left( \DT^T \DT \right)^{-1} \DT^T\E \right\|_\infty \leq &\left\| \DTB^T\DT \left( \DT^T \DT \right)^{-1} \right\|_\infty \epsilon_L.
\end{align*}
Recall that the induced infinity norm of a matrix is equal to the maximal $\ell_1$ norm of its rows. Notice that a row in the above matrix can be written as $\d_i^T\DT \left( \DT^T \DT \right)^{-1}$, where $i\in \overline{\mathcal{T}}$. Then,
\begin{align}
&\left\| \DTB^T\DT \left( \DT^T \DT \right)^{-1} \DT^T\E \right\|_\infty \\
\leq &\max_{i\in\overline{\mathcal{T}}} \left\| \d_i^T\DT \left( \DT^T \DT \right)^{-1} \right\|_1 \cdot \epsilon_L.
\end{align}
Using the definition of induced $\ell_1$ norm and Equation \eqref{eq:ERC_gram_l1_bound} and \eqref{eq:ERC_gram_infinity_bound}, we obtain that
\begin{align}
&\left\| \DTB^T\DT \left( \DT^T \DT \right)^{-1} \DT^T\E \right\|_\infty \\
\leq &\max_{i\in\overline{\mathcal{T}}} \left\| \d_i^T\DT \right\|_1 \cdot \left\| \left( \DT^T \DT \right)^{-1} \right\|_1 \cdot \epsilon_L \\
\leq &\max_{i\in\overline{\mathcal{T}}} \left\| \d_i^T\DT \right\|_1 \cdot \frac{1}{1-(k-1)\mu(\D)} \cdot \epsilon_L,
\end{align}
where we have denoted by $k$ the $\Loi$ norm of $\mathcal{T}$. Notice that due to the limited support of the atoms, the vector $\d_i^T\DT$ has at most $k$ non-zeros entries. Additionally, each of these is bounded in absolute value by the mutual coherence of the dictionary. Therefore, $\|\d_i^T\DT\|_1 \leq k\mu(\D)$ (note that $i\notin\mathcal{T}$). Plugging this into the above equation, we obtain
\begin{align}
\left\| \DTB^T\DT \left( \DT^T \DT \right)^{-1} \DT^T\E \right\|_\infty \leq \frac{k\mu(\D)}{1-(k-1)\mu(\D)} \cdot \epsilon_L.
\end{align}
Rearranging our assumption in Equation \eqref{eq:lemma_assumption}, we get $\frac{k\mu(\D)}{1-(k-1)\mu(\D)}\leq 1$. Therefore, the above becomes
\begin{align} \label{eq:using_assumption}
\left\| \DTB^T\DT \left( \DT^T \DT \right)^{-1} \DT^T\E \right\|_\infty \leq \epsilon_L.
\end{align}
Finally, plugging Equation \eqref{eq:local_noise} and \eqref{eq:using_assumption} into Equation \eqref{eq:triangle_inequality_step}, we conclude that
\begin{align}
& \| \DTB^T(\Y-\X_{\text{LS}}) \|_\infty \\
\leq & \left\| \DTB^T\E \right\|_\infty + \left\| \DTB^T\DT \left( \DT^T \DT \right)^{-1} \DT^T\E \right\|_\infty \\
\leq & \ \epsilon_L + \epsilon_L = 2\epsilon_L.
\end{align}
\vspace{-0.2cm}
\end{proof}

For completeness, and before moving to the proof of the stability of BP, we now reproduce Theorem 8 from \cite{Tropp2006}.
\begin{thm}{(Tropp):} \label{Theorem:Tropp}
	Suppose a clean signal $\X$ has a representation $\D\Gama$, and that it is contaminated with noise $\E$ to create the signal $\Y=\X+\E$. Assume further that $\Y$ is a signal whose best $\ell_2$ approximation over the support of $\Gama$, denoted by $\mathcal{T}$, is given by $\X_{\text{LS}}$, and that $\X_{\text{LS}} = \D \Gama_{\text{LS}}$. Moreover, consider $\Gama_{\text{BP}}$ to be the solution to the Lagrangian BP formulation with parameter $\lambda$. If the following conditions are satisfied:
	\begin{enumerate}[ a) ]
		\item \label{eq:ERC_assumption} The ERC is met with constant $\theta \geq 0$ for the support $\mathcal{T}$; And
		\item \label{eq:Corr_assumption} $\| \D^T(\Y-\X_{\text{LS}}) \|_\infty \leq \lambda\theta$,
	\end{enumerate}
	then the following hold:
	\begin{enumerate}
		\item The support of $\Gama_{\text{BP}}$ is contained in that of $\Gama$.
		\item \label{thesis_2} $\|\Gama_{\text{BP}}-\Gama_{\text{LS}} \|_\infty \leq \lambda \left\| \left( \D_\mathcal{T}^T \D_\mathcal{T} \right)^{-1} \right\|_\infty$.
		\item \label{thesis_3} In particular, the support of $\Gama_{\text{BP}}$ contains every index $i$ for which $|{{\Gama_{\text{LS}}}_{ }}_i| > \lambda \left\| \left( \D_\mathcal{T}^T \D_\mathcal{T} \right)^{-1} \right\|_\infty$.
		\item The minimizer of the problem, $\Gama_{\text{BP}}$, is unique.
	\end{enumerate}
\end{thm}

Armed with these, we now proceed to proving Theorem \ref{Theorem:StabilityBP}.

\begin{proof}
In this proof we shall show that Theorem \ref{Theorem:Tropp} can be reformulated in terms of the $\Loi$ norm and the mutual coherence of $\D$, thus adapting it to the convolutional setting. Our strategy will be first to restrict its conditions \eqref{eq:ERC_assumption} and \eqref{eq:Corr_assumption}, and then to derive from its theses the desired claims.

To this end, we begin by converting the assumption on the ERC into another one relying on the $\Loi$ norm. This can be readily done using Theorem \ref{Theorem:ERC_Loi}, which states that the ERC is met assuming the $\Loi$ norm of the support is less than $\frac{1}{2} \left( 1 + \frac{1}{\mu(\D)} \right)$ -- a condition that is indeed satisfied due to our assumption in Equation \eqref{eq:BP_assumption}.
Next, we move to assumption \eqref{eq:Corr_assumption} in Theorem \ref{Theorem:Tropp}. We can lower bound the ERC constant $\theta$ by employing the inequality in \eqref{eq:ERC_theta_bound}, thus obtaining
\begin{equation} \label{eq:theta_inequality}
\theta =  1 - \underset{i \notin \mathcal{T}}{\max} \|\D^{\dagger}_{\mathcal{T}} \d_i \|_1 \geq 1 - \frac{\|\Gama\|_{0,\infty} \mu(\D)}{1 - (\|\Gama\|_{0,\infty}-1) \mu(\D) }.
\end{equation}
Using the assumption that $\|\Gama\|_{0,\infty}\leq\frac{1}{3} \left( 1 + \frac{1}{\mu(\D)} \right)$, as stated in Equation \eqref{eq:BP_assumption}, the above can be simplified into
\begin{equation} \label{eq:lower_bound_theta}
\theta =  1 - \underset{i \notin \mathcal{T}}{\max} \|\D^{\dagger}_{\mathcal{T}} \d_i \|_1 \geq \frac{1}{2}.
\end{equation}
Bringing now the fact that $\lambda = 4\epsilon_L$, as assumed in our Theorem, and using the just obtained inequality \eqref{eq:lower_bound_theta}, condition \eqref{eq:Corr_assumption} must hold since
\begin{equation}
\| \D^T(\Y-\X_{\text{LS}}) \|_\infty \leq 2 \epsilon_L \leq \theta  \lambda.
\end{equation}
Note that the leftmost inequality is Lemma \eqref{lemma:noise_correlation}, and the implication here is that $\lambda \geq 4 \epsilon_L$.

Thus far, we have addressed the conditions in Theorem \ref{Theorem:Tropp}, showing that they hold in our convolutional setting. In the remainder of this proof we shall expand on its results, in particular point 2 and 3. We can upper bound the term $\left\| \left( \D_\mathcal{T}^T \D_\mathcal{T} \right)^{-1} \right\|_\infty$ using Equation \eqref{eq:ERC_gram_infinity_bound}, obtaining
\begin{equation}
\left\| \left( \D_\mathcal{T}^T \D_\mathcal{T} \right)^{-1} \right\|_\infty \leq \frac{1}{1 - ( \|\Gama\|_{0,\infty} -1) \mu(\D) }.
\label{eq:inequationFinal}
\end{equation}
Using once again the assumption that $\|\Gama\|_{0,\infty}\leq\frac{1}{3}\left(1+\frac{1}{\mu(\D)}\right)$, we have that $\|\Gama\|_{0,\infty}<\frac{1}{3}\left(3+\frac{1}{\mu(\D)}\right)$. From this last inequality, we get $\left( \|\Gama\|_{0,\infty} - 1 \right) \mu(\D) < \frac{1}{3}$. Thus, it follows that
\begin{equation}
\frac{1}{1 - ( \|\Gama\|_{0,\infty} -1) \mu(\D) } < \frac{3}{2}.
\end{equation}
Based on the above inequality, and Equation \eqref{eq:inequationFinal}, we get
\begin{equation}\label{eq:InverseGramBound}
\left\| \left( \D_\mathcal{T}^T \D_\mathcal{T} \right)^{-1} \right\|_\infty < \frac{3}{2}.
\end{equation}
Plugging this into the second result in Tropp's theorem, together with the above fixed $\lambda$, we obtain that
\begin{equation} \label{eq:BoundGamaBPLS}
\|\Gama_{\text{BP}}-\Gama_{\text{LS}} \|_\infty \leq \lambda \left\| \left( \D_\mathcal{T}^T \D_\mathcal{T} \right)^{-1} \right\|_\infty < 4\epsilon_L \cdot \frac{3}{2} = 6\epsilon_L.
\end{equation}
On the other hand, looking at the distance to the real $\Gama$,
\begin{align} \label{eq:BoundLsLast}
	\|\Gama_{\text{LS}}-\Gama\|_{\infty} & = \| \left( \D_\mathcal{T}^T \D_\mathcal{T} \right)^{-1} \D_\mathcal{T}^T\left( \Y - \X \right)\|_{\infty}\\ 
									    & \leq  \| \left( \D_\mathcal{T}^T \D_\mathcal{T} \right)^{-1} \|_{\infty}\cdot \| \D_\mathcal{T}^T \E \|_{\infty} < \frac{3}{2} \epsilon_L.
\end{align}
For the first inequality we have used the definition of the induced $\ell_\infty$ norm, and the second one follows from \eqref{eq:InverseGramBound} and a similar derivation to that in \eqref{eq:local_noise}. Finally, using triangle inequality and Equations \eqref{eq:BoundLsLast} and \eqref{eq:BoundGamaBPLS} we obtain
\begin{equation}
\| \Gama_{\text{BP}} - \Gama \|_{\infty} \leq \| \Gama_{\text{BP}} - \Gama_{\text{LS}} \|_{\infty} + \| \Gama_{\text{LS}} - \Gama \|_{\infty} < \frac{15}{2} \epsilon_L.
\end{equation}
The third result in the theorem follows immediately from the above.

\end{proof}

\section{Global Pursuit Through Local Processing}

Let us consider the Iterative Soft Thresholding algorithm which minimizes the global BP problem by iterating the following updates
\begin{equation}
\Gama^k = \mathcal{S}_{\lambda/c}\left( \Gama^{k-1} + \frac{1}{c} \D^T(\Y-\D\Gama^{k-1}) \right),
\end{equation}
where $\mathcal{S}$ applies an entry-wise soft thresholding operation with threshold $\lambda/c$. 
Defining as $\P_i$ the operator which extracts the $i^{th}$ $m-$dimensional vector from $\Gama$, we can break the above algorithm into local updates by
\begin{equation}
\P_i\Gama^k = \mathcal{S}_{\lambda/c}\left( \P_i\Gama^{k-1} + \frac{1}{c} \P_i\D^T(\Y-\D\Gama^{k-1}) \right).
\end{equation}
As a first observation, the matrix $\P_i\D^T$, which is of size $m \times N$, is in-fact $\D_L^T$ padded with zeros. As a consequence, the above can be rewritten as follows:
\begin{equation}
\P_i\Gama^k = \mathcal{S}_{\lambda/c}\left( \P_i\Gama^{k-1} + \frac{1}{c} \P_i\D^T\R_i^T \R_i (\Y-\D\Gama^{k-1}) \right),
\end{equation}
where we have used $\R_i$ as the operator which extracts the $i^{th}$ $n-$dimensional patch from an $N-$dimensional global signal. The operator $\P_i$ extracts $m$ rows from $\D^T$, while $\R_i^T$ extracts its non-zero columns. Therefore, $\P_i\D^T\R_i^T = \D_L^T$, and so we can write
\begin{equation}
\P_i\Gama^k = \mathcal{S}_{\lambda/c}\left( \P_i\Gama^{k-1} + \frac{1}{c} \D_L^T \R_i (\Y-\D\Gama^{k-1}) \right).
\end{equation}
Noting that $\alfa_i^k=\P_i\Gama^k$ is the $i^{th}$ local sparse code, and defining $\r_i^k=\R_i (\Y-\D\Gama^{k-1})$ as the corresponding patch-residual at iteration $k$, we obtain our final update (for every patch)
\begin{equation}
\alfa_i^k = \mathcal{S}_{\lambda/c}\left( \alfa_i^{k-1} + \frac{1}{c} \ \D_L^T \ \r_i^{k-1} \right).
\end{equation}
We summarize the above derivations in Section IX in the paper.

\end{document}